\documentclass[acmsmall]{acmart}

\usepackage{calc}

\usepackage{amsthm,amsmath,stmaryrd}
\usepackage[aboveskip=1pt,belowskip=2pt]{subcaption}
\usepackage{thmtools, thm-restate}
\usepackage[conf={no link to proof,restate,text proof translated=}]{proof-at-the-end}

\usepackage{tikz}
\usepackage{pgfplots}
\usetikzlibrary{calc,matrix,arrows,automata,positioning,chains,cd,arrows.meta,shapes.geometric,decorations.pathreplacing}

\theoremstyle{definition}
\newtheorem{definition}{Definition}
\newtheorem{example}{Example}[section]
\newtheorem{lemma}{Lemma}

\newtheorem{proposition}{Proposition}
\newtheorem{corollary}{Corollary}
\newtheorem{theorem}{Theorem}
\newtheorem*{theorem*}{Theorem}

\newtheorem{observation}{Observation}

\newcommand{\tuple}[1]{\left\langle #1 \right\rangle}     % n-tuples
\newcommand{\sem}[1]{\left\llbracket #1 \right\rrbracket} % Semantic brackets
\newcommand{\set}[1]{\left\lbrace #1 \right\rbrace}       % sets

\newcommand{\false}{\textit{false}}
\newcommand{\defeq}{\triangleq}

\definecolor{black}{RGB}{0,0,0}
\definecolor{green}{RGB}{25,160,80}
\definecolor{red}{RGB}{192,57,43}
\definecolor{blue}{RGB}{41,128,185}
\definecolor{orange}{RGB}{255,99,0}
\definecolor{gray}{RGB}{75,75,75}
\definecolor{lightgrey}{RGB}{240,240,240}
\definecolor{purple}{RGB}{155, 89, 182}
\definecolor{darkgrey}{RGB}{52, 73, 94}

\usepackage{algorithm2e}
\SetKwInOut{Input}{Input}
\SetKwInOut{Output}{Output}
\SetKwRepeat{Do}{do}{while}

\SetCommentSty{mycommfont}
\SetKwProg{Fn}{Subroutine}{ begin}{end}
\SetKwComment{Inv}{\rm\textbf{invariant:} }{}

\usepackage{listings}
\lstdefinestyle{base}{
  language=C,
  emptylines=1,
  breaklines=true,
  basicstyle=\color{black}\fontfamily{pzc}\selectfont\tt,
  commentstyle=\color{gray}\rm\itshape,
  keywordstyle=\rm\bfseries,
  identifierstyle=\rm\itshape,
  escapechar=|,
  morekeywords=[1]{then,do},
  morekeywords=[2]{assert},
  morekeywords=[3]{requires,ensures},
  keywordstyle	= [2]\color{red}\bf,
  keywordstyle	= [3]\bf\color{gray},
  numberstyle=\footnotesize\color{gray},
  moredelim=**[is][\bf\color{green}]{@!}{!@},
  moredelim=**[is][\bf\color{red}]{@?}{?@},
  literate=
    {<=}{{$\leq$}}1
    {>=}{{$\geq$}}1
    {!}{{$\neg$}}1
    {!=}{{$\neq$}}1
    {||}{{$\lor$}}1
    {&&}{{$\land$}}1
    {->}{{$\rightarrow$}}1
    {_1}{$_{1}$}2
    {_2}{$_{2}$}2
    {_3}{$_{3}$}2
}
\newcommand{\cinline}[1]{\mbox{\lstinline[style=base,mathescape]{#1}}}

\usepackage{mathpartir} % Inference rules
\usepackage{booktabs}   % Tables
\usepackage{hyperref}   % hyperlinks
\newcommand{\src}{\textit{src}}
\newcommand{\dst}{\textit{dst}}

\newcommand{\seq}{\cdot}
\newcommand{\Ob}[1]{\textit{Ob}(\mathbf{#1})}
\newcommand{\proj}[1]{\phi_{\mathbf{#1}}}
\newcommand{\invimg}[2]{#1^{-1}(#2)}

\newcommand{\cat}[1]{\mathbf{#1}}

\newcommand{\functor}[1]{{#1}}
\newcommand{\lifted}[1]{\overline{#1}}

\DeclareFontFamily{U}  {MnSymbolC}{}
\DeclareFontShape{U}{MnSymbolC}{m}{n}{
    <-6>  MnSymbolC5
   <6-7>  MnSymbolC6
   <7-8>  MnSymbolC7
   <8-9>  MnSymbolC8
   <9-10> MnSymbolC9
  <10-12> MnSymbolC10
  <12->   MnSymbolC12}{}
\DeclareSymbolFont{MnSyC}{U}{MnSymbolC}{m}{n}
\DeclareMathSymbol{\righthalfcup}{\mathrel}{MnSyC}{184}
\NewCommandCopy\oldexample\example
\NewCommandCopy\endoldexample\endexample
\renewenvironment{example}[1][]{%
  \if\relax\detokenize{#1}\relax%
    \oldexample%
  \else%
    \oldexample[#1]%
  \fi%
  \pushQED{\qed}%
 }{%
   \popQED%
 \endoldexample%
}

\newcommand{\dom}{\text{dom}}

\newcommand{\tr}[3]{#2 \rightarrow_{#1} #3}

\renewcommand{\vec}[1]{\mathbf{#1}}

\newcommand{\TF}{\mathbf{TF}}

\newcommand{\zak}[1]{}
\newcommand{\zsw}[1]{}

\newcommand{\barr}[3]{#1 \mathrel{\smash{\Vdash\!\!\xrightarrow{\raisebox{-0.25ex}[0ex][0ex]{$\scriptstyle#2$}}}} #3}

\newcommand{\bcarr}[3]{#1 \mathrel{\smash{\Vdash\!\!\xrightarrow{\raisebox{-0.4ex}[0ex][0ex]{$\scriptstyle#2$}}\mathrel{\mkern-14mu}\rightarrow}} #3 }

\acmJournal{PACMPL}
\acmVolume{1}
\acmNumber{CONF} % CONF = POPL or ICFP or OOPSLA
\acmArticle{1}
\acmYear{2018}
\acmMonth{1}
\acmDOI{} % \acmDOI{10.1145/nnnnnnn.nnnnnnn}
\startPage{1}
\setcopyright{none}
\acmConference[Conference acronym 'XX]{Make sure to enter the correct
  conference title from your rights confirmation emai}{June 03--05,
  2018}{Woodstock, NY}
\acmPrice{15.00}
\acmISBN{978-1-4503-XXXX-X/18/06}


\begin{document}

\title{A Categorical Basis for Robust Program Analysis}

\author{Zachary Kincaid}
\email{zkincaid@cs.princeton.edu}
\orcid{0000-0002-7294-9165}
\affiliation{%
  \institution{Princeton University}
  \city{Princeton}
  \state{New Jersey}
  \country{USA}
}

\author{Shaowei Zhu}
\email{shaoweiz@cs.princeton.edu}
\orcid{0000-0002-0335-1151}
\affiliation{%
  \institution{Princeton University}
  \city{Princeton}
  \state{New Jersey}
  \country{USA}
}

\begin{abstract}
Users of program analyses expect that results change predictably in response to changes in their programs, but many analyses fail to provide such robustness.
This paper introduces a theoretical framework that provides a unified language to articulate robustness properties. By modeling programs and their properties as objects in a category, diverse notions of robustness—from variable renaming to semantic refinement and structural transformation—can be characterized as structure-preserving functors.

Beyond formulating the meaning of robustness, this paper provides methods for achieving it.  The first is a general recipe for designing robust analyses, by lifting a sound and robust analysis from a restricted (sub-Turing) model of computation to a sound and robust analysis for general programs.
This recipe demystifies the design of several existing loop summarization and termination analyses by showing they are instantiations of this general recipe, and furthermore elucidates their robustness properties.  The second is a characterization of a sense in which an algebraic program analysis is robust, provided that it is comprised of robust operators.  In particular, we show that such analyses behave predictably under common refactoring patterns, such as variable renaming and loop unrolling. 
\end{abstract}

% Tighten float spacing to fit 22-page limit
\setlength{\abovecaptionskip}{2pt}
\setlength{\belowcaptionskip}{0pt}
\setlength{\textfloatsep}{6pt plus 2pt minus 2pt}
\setlength{\floatsep}{6pt plus 2pt minus 2pt}
\setlength{\intextsep}{6pt plus 2pt minus 2pt}

\maketitle

\section{Introduction} \label{sec:introduction}

The goal of program analysis is to predict how a program's state evolves over time.  However, programs are not static artifacts---they too change.  This paper is motivated by the question of how we can predict how a \textit{program analysis} will behave in response to such changes.

Abstractly, we might think of a program analysis as a function mapping a program $P$ to some approximation $A(P)$ of $P$'s behavior (for instance, $A(P)$ might be a mapping from the loops of $P$ to associated invariants, or a yes/no decision about whether $P$ terminates).
Suppose that a given program $P$ is transformed in some particular way to yield a program $P'$.  What can be said about the relationship between $A(P)$ and $A(P')$?  Reasonable expectations might be:
\begin{itemize}
\item \textit{(Invariance under renaming)}: If  $P$ and $P'$ are equivalent modulo renaming of variables, $A(P)$ and $A(P')$ are also equivalent modulo renaming.
\item \textit{(Locality)}: $A(P)$ and $A(P')$ differ only in the places that $P$ and $P'$ do (for instance, one might expect that changing one function does not affect the information computed for another unrelated function).
\item \textit{(Monotonicity)}: If $P'$ is a refinement of $P$ (e.g., $P'$ is obtained from $P$ by inserting assumed invariants) then $A(P')$ should be at least as precise as $A(P)$.
\end{itemize}

Unfortunately, many program analysis tools do not provide guarantees on ``robustness'', which can make them difficult to use and pose challenges in their deployments in production environments.
For instance, an experience report on the commercial static analyzer Coverity notes that users expect consistent results when their code changes, even across tool versions, and discusses how Coverity eschews techniques like SMT solving, randomization, and use of time-outs to avoid unstable bug reports \cite{Coverity}.  Several reports lament failure of locality, or \textit{butterfly effects}---``a minor modification in
one part of the program source causes changes in the outcome of the verification
in other, unchanged and unrelated parts of the program'' \cite{CAV:LP2016}.  As \citet{UV:LM2010} put it succinctly: ``tools can understand our programs, but we cannot understand our tools.''

The main contributions of this paper are:
\begin{itemize}
\item (Section~\ref{sec:robustness}) We develop a theoretical framework for understanding robustness.  The high-level idea is that a program analysis $A$ should be considered to be robust with respect to a class of relationships $R$ (e.g., equivalence modulo renaming variables) exactly when, should some relationship $r \in R$ hold between two programs $P$ and $P'$, then $r$ also holds between $A(P)$ and $A(P')$.  We formalize this intuitive notion using the language of category theory: we think of both the space of programs and analysis results as forming categories (in which the arrows represent a class of relationships of interest), and define an analysis to be robust exactly when it preserves this categorical structure.

\item (Section~\ref{sec:lifting}) We present a general recipe for designing robust program analyses. 
 The
idea is to leverage weak (that is, not Turing-complete) models of computation
that have tractable dynamics.  To analyze a program we (1) extract a
\textit{model} of the program of some restricted form, (2) analyze the 
behavior of the model, and then (3) translate the  model analysis to an
\textit{approximate} analysis for the original program. We show that, supposing certain properties of the model extraction and result translation process,
this recipe preserves robustness properties of the model analysis.  Furthermore, we identify instances of this recipe in the literature, and thereby show that they are robust.
\item (Section~\ref{sec:apa}) We investigate robustness in the context of algebraic program analysis, driven by the question of
how robustness guarantees of the individual operators of an algebra translate to high-level guarantees of the whole analysis.  In particular, we
show that program analyses following the recipe of Section~\ref{sec:lifting} are robust with respect to \textit{loop-preserving stuttering simulations}.  This class of relationships accounts for both non-structural transformations (e.g., renaming a variable) and structural transformations (e.g., unrolling a loop).
\end{itemize}

%%% Local Variables:
%%% TeX-master: "main"
%%% End:

\subsection{Related Work}

\paragraph{Program analysis frameworks}

Abstract interpretation provides a uniform framework for \textit{sound} program analysis \cite{POPL:CC1977}.
It defines what it means for a program analysis to be sound in terms of an approximation relationship between a ``concrete'' semantics defining the program behavior and an ``abstract'' semantics defining the program analysis.  It also provides a general recipe for designing sound program analyses---the job of the analysis designer is to develop a lattice of program properties, transfer functions that ``lift'' program commands to act on these properties, and a widening operator to extrapolate the limit of a sequence of properties.

The goal of our paper is to do for \textit{robustness} what abstract interpretation does for \textit{soundness}.
While partial orders are sufficient to relate possible invariants for a single program ($X \sqsubseteq Y$ means that $Y$ approximates $X$), we require richer structure to relate possible invariants for different programs.
In our framework, the ``abstract domain'' forms a category.  The arrows of a category can be thought of as a generalization of the approximation order in abstract interpretation:
an arrow $f : X \rightarrow Y$ means that $Y$ approximates $X$ \textit{modulo the transformation $f$}.    For instance, if $f$ is a transformation that replaces a variable $x$ with $a$ and $y$ with $b$, we might write
$f : (0 \leq x \leq y \leq 10) \rightarrow (0 \leq a \land b \leq 100)$, meaning that the property on the left of the arrow under the transformation $f$ (that is, $0 \leq a \leq b \leq 10$) is stronger than the property on the right.
In this sense, one can think of an order-theoretic abstract domain as a category in which the only transformations of interest are identity transformations.  Table~\ref{tab:ai-glossary} provides a correspondence between some category-theoretic terminology used in this paper and their order-theoretic counterparts to provide intuition for how our work connects with abstract interpretation.

Prior categorical treatments of abstract interpretation~\cite{RAIRO:SJM1992,MFPS:KRD2023} use categories in a different way than ours: they use enriched categories in which objects are types, arrows are programs, and homsets carry a poset structure; analyses are lax functors, encoding the idea that the analysis of a composite program is more precise than the separate analysis of its components.  In contrast, our categorical construction uses objects for programs and arrows for \emph{relationships between programs}, where analyses are functors from programs to results, encoding the idea that robust analyses preserve relationships between programs.

\begin{table}[t]
\caption{Correspondence between order-theoretic terminology in abstract interpretation and category-theoretic terminology in this paper. \label{tab:ai-glossary}}
\begin{tabular}{cc}
\toprule
\textbf{Abstract interpretation} & \textbf{Robust program analysis}\\ \midrule
Lattice & Concrete category\\
Monotone function & Concrete functor\\
Galois connection & Weak adjoint\\
\bottomrule
\end{tabular}
\end{table}

Algebraic program analysis is a framework for compositional program analysis (that is, satisfying the \textit{locality} principle) \cite{CAV:KRC2021}.  In algebraic program analysis, the typical goal is to compute a summary that over-approximates the dynamics of a given program.  This summary is computed in ``bottom-up'' fashion: starting with single instructions, we build summaries for ever-larger  fragments of the program by combining summaries using operations that mirror those of regular expressions: sequencing, choice, and iteration.

A recent line of work has developed \textit{monotone} program analyses, meaning that the iteration operator satisfies the property that if one summary $F$ approximates another $G$, then $F^*$ must approximate $G^*$ (and each other operator is similarly robust) \cite{CAV:SK2019,PLDI:ZK2021,CAV:ZK2021,POPL:CK2024,CAV:ZK2024,OOPSLA:PK2024}.
These analyses offer better predictability guarantees while matching (or even exceeding) state-of-the-art reasoning capabilities.
Monotonicity can be seen as a special case of robustness: it concerns
the relationship between the analyses of two programs when one refines
the other.  Our framework generalizes this by considering richer classes
of relationships between programs (e.g., variable renamings, linear
changes of variables), and in this way both unifies the design principles
behind existing monotone analyses and extends their guarantees to
structural program transformations such as loop unrolling.
The analysis design strategy presented in Section~\ref{sec:lifting} is
inspired by the one from \cite{SAS:Kincaid2018}.  However, our
strategy is more general (applying, for instance, to termination
analysis as well as invariant generation) and more flexible in the
kinds of relationships it preserves.  Moreover, we show that our
strategy also preserves certain algebraic laws
(Section~\ref{sec:algebraic-laws-main}).

\paragraph{Program transformation and analysis}

Control-flow refinement is a family of techniques that aim to improve analysis precision by performing loop transformations (e.g., unrolling a loop, or splitting a loop into two or more phases).  While these techniques have been shown to improve analysis precision experimentally, there are also cases where these transformations lead to results that are incomparable or even less precise.  An exception is \cite{POPL:CBKR2019}, which gives a transformation that is guaranteed to improve precision, subject to certain conditions of the analysis.  In this paper, we give a recipe for designing analyses that meet these conditions.

\citet{SAS:APS2010} and \citet{SAS:RZ2018} observe that numerical program analyses are often not robust under invertible linear transformations---i.e., if one applies a linear change of variables $f$ to a program $P$ (substituting each variable $x$ with a linear combination of fresh variables) to obtain an isomorphic program $f[P]$, the computed invariants are generally not isomorphic.  The authors exploit this as a way of improving analysis precision: analyze both the programs, and take the conjunction of the invariants computed for $P$ and $f[P]$ (under the inverse transformation).  \citet{ENTCS:ML2012} make a similar observation about non-robustness with respect to projecting out variables, and proposes an analogous technique to analyze it.   In this paper, we show that some algebraic analyses that have previously been shown to be \emph{monotone} are in fact \textit{robust under linear transformations} (a generalization of the original notion of monotonicity which includes not only logical entailment of the results but also invariance under renaming of variables).
As a result, these robust analyses do not benefit from the aforementioned techniques, but the designer of the robust analyses need not be
concerned with \emph{which} linear transformation or projection to use in the first place.

\citet{CC:LF2008} observe that translating a program from a high-level language to a low-level language (e.g., bytecode) negatively impact the precision program analysis, and introduced the notion of \textit{relative completeness} to describe the situation in which precision loss does not occur.  \citet{PLDI:LL2024} develop a strategy for combining non-relational analyses with online SSA conversion to achieve analyses that are relatively complete in this sense.   Our work provides a general framework in which to understand program analyses ``preserving'' relationships between programs; relative completeness is an example of a robustness property in which the relationship is a transformation that introduces new variables to represent intermediate computations.

%%% Local Variables:
%%% TeX-master: "main"
%%% End:

\section{Overview} \label{sec:overview}

This section illustrates examples of non-robustness of program analysis, as well as an analysis that achieves robustness by following the recipe presented in Section~\ref{sec:lifting}.
First, we consider an example consisting of two programs $P_1$ and $P_2$, pictured in Figure~\ref{fig:overview-ex} (derived from \cite{ENTCS:ML2012}).
$P_1$ can be obtained from $P_2$ by (1) renaming the $x$ variable to be $i$, and (2) projecting out the $y$ variable.  Thus, one might expect that a loop invariant generation algorithm would produce an invariant for $P_2$ that is at least as precise as $P_1$ (modulo renaming $x$ to be $i$).
However, \citet{ENTCS:ML2012} observed that polyhedral analysis violates this expectation---the invariant computed for $P_1$ is $1 \leq i \leq 5$, while the invariant for $P_2$ is $x \leq 5$.
In the vocabulary of this paper, this example illustrates that polyhedral analysis is not robust with respect to linear transformations.
The essential reason is that polyhedral analysis uses a widening operator to extrapolate the limit of a sequence of successively weaker candidate invariants---the extrapolation step is essentially an educated guess, and is sensitive to the particular way that a polyhedron is represented.

We now consider an alternative method to analyze this loop, which makes use of polyhedra in a different way, and which \textit{is} robust with respect to linear transformations.  Instead of approximating the reachable states of a machine, our goal is to summarize its dynamics.  This summary can be represented by a \textit{transition formula}---a logical formula over a set of program variables and their ``primed'' copies, representing the program state before and after a computation. For instance, the instruction \cinline{x++} in the program $P_2$ can be represented by the transition formula $x' = x + 1 \land y' = y$.

\textit{Algebraic program analysis} computes such summaries in ``bottom-up'' fashion: it begins by summarizing individual statements, then builds up summaries for larger and larger fragments of the program until it obtains a summary for the whole program.  A central problem of interest in algebraic program analysis is the design of \textit{iteration operators}, which given a summary for the body of a loop, computes a summary that represents any number of iterations of that loop body formula.

We consider two particular ways to use the domain of convex polyhedra to define an iteration operator, inspired by the analyses of \citet{ENTCS:ACI2010} and \citet{FMCAD:FK2015}.  For any linear integer/rational arithmetic formula $F$, we use $\textit{conv}(F)$ to denote the closed convex hull of $F$; i.e.,
a formula of the form $\bigwedge_{i=0}^n t_i \geq b_i$ where each $t_i$ is a linear term (over the free variables of $F$) and $b_i$ is a rational constant, such that
$F$ entails $\textit{conv}(F)$ and $\textit{conv}(F)$ entails any linear inequality that is entailed by $F$.
Suppose that $G$ is a transition formula over the variables $X$ and $X'$; our goal is to compute a formula $G^\star$ that over-approximates the reflexive transitive closure of $G$.
\begin{itemize}
\item \textit{(Polyhedral guard analysis)}  Define $\textit{Pre}(G) \defeq \textit{conv}(\exists X'. G)$ and $\textit{Post}(G) \defeq \textit{conv}(\exists X. G)$, which represent the pre-condition and post-condition of the transition $G$, respectively.  As long as at least one iteration of $G$ is taken both $\textit{Pre}(G)$ and $\textit{Post}(G)$ must hold, and so
\[ G^{\star_{\text{PGA}}} \defeq X' = X \lor (\textit{Pre}(G) \land \textit{Post}(G))
\]
over-approximates the reflexive transitive closure of $G$.
\item \textit{(Linear recurrence analysis)} A \textit{linear recurrence} is a linear inequality of the form $t' \leq t + b$ where $t$ is a linear term over the pre-state variables $X$, $t'$ is $t$ under the substitution that replaces each unprimed variable with its primed counter-part, and $b$ is a constant.  If a loop body summary $G$ entails a linear recurrence $t' \leq t + b$, then $t$ may increase by at most $b$ in one iteration of the loop (and thus may increase by at most $kb$ after $k$ iterations of the loop).
We can compute linear recurrences from a loop body formula $G$ as follows.  First, we introduce a
set of variables
$\Delta_X \defeq \set{ \delta_x : x \in X}$, where each $\delta_x$ represents the change in $x$, and define  \[ \Delta(G) \defeq \exists X,X'. G \land \bigwedge_{x \in X} \delta_x = x' - x\ .\]
Then $\textit{conv}(\Delta(G))$ takes the form $\bigwedge_{i=1}^n t_i \leq b_i$, where each $t_i$ is a linear term over $\Delta_X$.  Since $\Delta(G)$ entails each $t_i \leq b_i$, we have that $G$ entails $t_i[\Delta_X \mapsto X'] \leq t_i[\Delta_X \mapsto X] + b_i$ (where $t[\Delta_X \mapsto X']$ denotes the parallel substitution that replaces each $\delta_x$ with $x'$ in the linear term $t$)---i.e., each constraint in the constraint representation of the convex hull of $\Delta(G)$ corresponds to a linear recurrence inequality entailed by $G$.  Thus,
\[ G^{\star_{LRA}} \defeq \exists k \in \mathbb{N}. k \geq 0 \land \bigwedge_{i=1}^n t_i[\Delta_X \mapsto X'] \leq t_i[\Delta_X \mapsto X] + kb_i \]
over-approximates the reflexive transitive closure of $G$.
\end{itemize}
Naturally, these two ideas can be combined into one iteration operator, which we will refer to as the \textit{polyhedral iteration operator}:
\newcommand{\pstar}{{\tikz[baseline=(char.base)]{\node[shape=regular polygon,regular polygon sides=5,draw,inner sep=-0.5pt] (char) {\scriptsize $\star$};}}}

\[ G^{\pstar} \defeq \exists k \in \mathbb{N}. \left( \begin{array}{l@{}l}&\bigwedge_{i=1}^n t_i[\Delta_X \mapsto X'] \leq t_i[\Delta_X \mapsto X] + kb_i\\ \land & ((k = 0 \land X' = X) \lor (k \geq 1 \land \textit{Pre}(G) \land \textit{Post}(G)))
\end{array}\right)
\]

\begin{figure}
\begin{subfigure}{0.3\textwidth}
\begin{lstlisting}[style=base]
i = 1;
while(i < 5) {
  i++;
}\end{lstlisting}
\caption{$P_1$ \label{fig:overview-ex1}}
\end{subfigure}
\begin{subfigure}{0.3\textwidth}
\begin{lstlisting}[style=base]
x = 1; y = 0;
while(x < 5) {
  x++;
  y += x;
}\end{lstlisting}
\caption{$P_2$ \label{fig:overview-ex2}}
\end{subfigure}
\begin{subfigure}{0.3\textwidth}
\begin{lstlisting}[style=base]
x = 1; y = 0;
while(x < 5) {
  x++; z = x;
  while(z-- > 0) y++;
}\end{lstlisting}
\caption{$P_3$ \label{fig:overview-ex3}}
\end{subfigure}
\caption{Three related programs $P_1$, $P_2$, and $P_3$.  $P_1$ is an abstraction of $P_2$ and $P_3$, and $P_2$ and $P_3$ are weakly bisimilar. \label{fig:overview-ex}}
\end{figure}
\begin{table}
\caption{Loop summary computation for the programs $P_1$ and $P_2$ from Figure~\ref{fig:overview-ex}.  Observe that the same linear transformation that relates $P_1$ and $P_2$ also relates the summary for $P_1$ and $P_2$. \label{tab:summaries}}
\small
\begin{tabular}{l|l|l}
\toprule
 & \multicolumn{1}{c|}{$P_1$} & \multicolumn{1}{c}{$P_2$} \\ \midrule
 Loop body $G$ & $G_1 \defeq i < 5 \land i' = i + 1$ & $G_2 \defeq x < 5 \land x' = x + 1 \land y' = y + x + 1$\\
 $\Delta(G)$ & $\exists i,i'.\, G_1 \land \delta_i = i'-i$ & $\exists x,x',y,y'.\, G_2 \land \delta_x = x'\!-\!x \land \delta_y = y'\!-\!y$\\
 $\textit{conv}(\Delta(G))$ & $-\delta_i \leq -1 \land \delta_i \leq 1$ & $-\delta_x \leq -1 \land \delta_x \leq 1 \land \delta_y \leq 5$\\
 Recurrences & $i' = i + 1$ & $x' = x + 1\land y' \leq y + 5$\\
 Precondition & $i < 5$ & $x < 5$\\
 Postcondition & $i' \leq 5$ & $x' \leq 5$\\
 Loop summary $G^{\pstar}$ & $\exists k. \begin{array}{l@{}l}&k \geq 0 \land i' = i + k\\ \land &(k \geq 1 \Rightarrow (i < 5 \land i' \leq 5))\end{array}$ & $\exists k. \begin{array}{l@{}l} &k \geq 0 \land x' = x + k\\ \land & y' \leq y + 5k\\ \land & (k \geq 1 \Rightarrow (x < 5 \land x' \leq 5))\end{array}$\\
 \bottomrule
\end{tabular}
\end{table}

The process of computing loop summaries using the polyhedral iteration operator for the programs $P_1$ and $P_2$ in Figure~\ref{fig:overview-ex} is depicted in Table~\ref{tab:summaries}.  Observe that the relationship between $P_1$ and $P_2$ also holds between the summary for the loops in $P_1$ and $P_2$: if we take the loop summary for $P_2$, project out the variable $y$ and $y'$ by existential quantification, and rename $x$ and $x'$ to $i$ and $i'$, we get the loop summary for $P_1$.
In fact, this holds for any pair of programs that are related by a linear transformation.  Thus, we say that this analysis is \textit{robust with respect to linear simulations}.

The programs $P_3$ and $P_2$ are also semantically related---the inner loop of $P_3$ corresponds exactly to the instruction \cinline{y += x} in $P_2$.  It happens to be the case that $G^{\pstar}$ preserves this relationship between $P_2$ and $P_3$; however, the theory of robustness developed in this paper does not explain why.  While it would certainly be desirable for a program analysis to be robust in the sense that it produces semantically equivalent results for semantically equivalent programs, this is impossible due to Rice's theorem.

%%% Local Variables:
%%% TeX-master: "main"
%%% End:

\section{Background and Notation} \label{sec:background}

\subsection{Transition Systems and Transition Formulas}

A \textbf{transition system} $T$ consists of a 
set of states $S_T$ and a transition relation ${\rightarrow_T} \subseteq S_T \times S_T$.
We use the notation $\tr{T}{x}{x'}$ to denote
that the pair $\tuple{x,x'}$ belongs to the transition relation of $T$.
Let $T$ and $U$ be transition systems.    A \textbf{simulation function from $T$ to $U$} is a function
$s : S_T \rightarrow S_U$ such that for all $x,x' \in S_T$ such that
$\tr{T}{x}{x'}$, we have $\tr{U}{s(x)}{s(x')}$. 

For transition systems $T$ and $U$ over the same state space, define their sequential composition
\[ T \seq U \defeq \set{\tuple{s,s''} : \exists s' \in S_T. \tr{T}{s}{s'} \land \tr{U}{s'}{s''}\ } . \]
We use $T^n$ to denote the
$n$-fold composition of $T$ with itself, and $T^*$ to denote the
reflexive transitive closure of $T$.  We use $T^\omega$ to denote the non-terminating states of $T$; i.e., the set of all states $s \in S_T$ such that there exists an infinite sequence
$s_0, s_1,\dots \in S_T$ such that $s_0 = s$ and
$\tr{T}{s_{i}}{s_{i+1}}$ for all $i \in \mathbb{N}$.

Let $X$ be a set of variable symbols.  A \textbf{transition formula} over $X$ is a first-order formula whose free variables range over $X$ and a set of ``primed copies'' $X' \defeq \set{x' : x \in X}$; these variables designate the state of a program before and after a transition, respectively.
For the sake of concreteness, we suppose that transition formulas are expressed in the language of linear integer/rational arithmetic.
We let $\TF(X)$ denote the set of all transition formulas over $X$.
Similarly, $\mathbf{SF}(X)$ denotes the state formulas over $X$---formulas whose free variables range over $X$.

A transition formula $F \in \TF(X)$ defines a transition system where $S_F = \mathbb{Q}^X$ (i.e., its states are maps from $X$ to $\mathbb{Q}$) and
$\tr{F}{s}{s'}$ iff the formula $F$ is satisfied by a model in which the variables in $X$ are interpreted according to $s$ and the variables in $X'$ are interpreted according to $s'$.  Note that transition formulas are closed under sequential composition, in the sense that for any transition formulas $F$ and $G$, the transition system $F \seq G$ is represented by the transition formula $\exists X''. F[X'\mapsto X''] \land G[X \mapsto X'']$.  However, it is not generally the case that
$F^*$ or $F^\omega$ are representable by a formula.

If $F \in \TF(X)$ and $G \in \TF(Y)$ are transition formulas, a \textbf{linear simulation} from $F$ to $G$ is a linear map $f : \mathbb{Q}^X \rightarrow \mathbb{Q}^Y$ that is also a simulation (if $\tr{F}{s}{s'}$, then $\tr{G}{f(s)}{f(s')}$).  Since formulas are \textit{constraints} on transitions, it is convenient to think of a linear simulation in \textit{transposed form}---that is, a substitution $\sigma$ that maps the variables in $Y$ to linear terms over the variables in $X$. Here the substitution $\sigma$ can be understood as the linear map that sends $s \in \mathbb{Q}^X$ to $\lambda y. \sem{\sigma(y)}(s)$, where $\sem{t}(s)$ denotes the evaluation of the term $t$ over $X$ in the state $s$.
We also introduce a convenient notation for representing simulations using term substitutions as follows.
For any substitution $\sigma$ mapping a set of variables $Y$ to terms over a set of variables $X$, we let $\sigma'$ denote ``primed'' variant of $\sigma$, such that $y' \in Y'$ is mapped to $\sigma(y)[X \mapsto X']$.  Then we have that
$\sigma$ is a simulation from $F$ to $G$ iff
$F \models G[\sigma, \sigma']$, where $G[\sigma,\sigma']$ denotes the result of applying both substitutions $\sigma$ and $\sigma'$ to $G$.  That is, we may think of $G[\sigma,\sigma']$ as the \textit{weakest} among the transition formulas $H$ such that
$\sigma$ is a simulation from $H$ to $G$.

\begin{example} \label{ex:linear-sim-p1-p2}
Consider the programs $P_1$ and $P_2$ from Figure~\ref{fig:overview-ex}. The loop body of $P_1$ is the transition formula $G_1 \defeq i < 5 \land i' = i + 1$ over $\{i\}$, and the loop body of $P_2$ is $G_2 \defeq x < 5 \land x' = x + 1 \land y' = y + x + 1$ over $\{x, y\}$.
The projection that maps $(x,y) \mapsto x$, which corresponds to the substitution $\sigma = [i \mapsto x]$, is a linear simulation from $G_2$ to $G_1$ (one may verify that $G_2 \models G_1[\sigma, \sigma'] \equiv (x < 5 \land x' = x + 1)$).
Additional of transition formulas and linear simulations appear in
Figure~\ref{fig:dats}.
%Figure~\ref{fig:cp} applies this idea to programs related by a linear simulation---in that case an invertible one (i.e., an isomorphism), and the inverse is a simulation in the other direction.
\end{example}

\subsection{Algebraic Program Analysis} \label{sec:apa-background}

Algebraic program analyses are a family of program analyses based on
Tarjan's path expression algorithm
\cite{JACM:Tarjan1981,JACM:Tarjan1981b}.  These analyses work
compositionally, in the same ``bottom-up'' style demonstrated in
Section~\ref{sec:overview}.  However, rather than defining the
analysis by recursion on the syntactic structure of the program, it is
defined by recursion on a regular expression recognizing the paths
through the program.  This allows algebraic analyses to be applied to
programs with arbitrary control structure (e.g. \cinline{goto}s) and
to be generalized to an interprocedural setting
\cite{FMCAD:FK2015,PLDI:KBFR17}.

A \textbf{Pre-Kleene algebra}\footnote{This definition of Pre-Kleene algebras first appeared in \cite{POPL:CBKR2019}, which observes that the laws of Kleene algebra are often not satisfied by algebraic program analyses.  For instance, the algebra of transition formulas cannot be extended with a $*$ operator to form a Kleene algebra, owing to the fact that least fixed points are not first-order definable.} (PKA) is an algebraic structure $\tuple{A,0,1,+,\cdot,(-)^*}$ equipped with two distinguished elements $0$ and $1$, two binary operations $+$ and $\cdot$, and a unary operation $(-)^{*}$
that satisfies the laws of idempotent semirings as well as the following laws governing the iteration operator: 
\noindent \begin{minipage}{0.4\textwidth}
\begin{itemize}
\item \textit{(Reflexivity)} $1 \leq a^*$
\item \textit{(Extensivity)} $a \leq a^*$
\item \textit{(Transitivity)} $a^*a^* = a^*$
\end{itemize}
\end{minipage}
\begin{minipage}{0.5\textwidth}
\begin{itemize}
\item \textit{(Monotonicity)} if $a \leq b$ then $a^* \leq b^*$
\item \textit{(Unrolling)} $(a^n)^* \leq a^*$ for all $n \in \mathbb{N}$
\end{itemize}
\end{minipage}

\noindent where $\leq$ is the order defined by the $+$ operator ($a \leq b$ iff $a + b = b$).
An example of a PKA is the algebra $\TF(X)$ in which the universe is the set of
 transition formulas over $X$, $0^{\textbf{TF}}$ is $\false$, $1$
 is the identity relation $\bigwedge_{x \in X} x' = x$, $+$ is
 disjunction, $\cdot$ is relational composition $\seq$, and
 $(-)^{*}$ is the over-approximate transitive closure operator
 $(-)^{\pstar}$.

Let $A$ be a PKA.  A \textbf{flow graph} over $A$
consists of a finite set of vertices $V_G$, a finite set of
directed edges $E_G \subseteq V_G \times V_G$, a distinguished root vertex
$r_G \in G$, and a weight function $w_G : E \rightarrow A$ assigning each edge a non-zero weight.  We extend $w_G$
to a function $E^* \rightarrow A$ as $w_G(e_1e_2 \dots e_n) \defeq w_G(e_1)
\cdot w_G(e_2) \cdot \dots \cdot w_G(e_n)$.

\subsection{Category Theory}

First, we recall some of the basic definitions of category theory.  A \textbf{category} $\mathbf{C}$ consists of
\begin{itemize}
\item A collection of objects $\Ob{\mathbf{C}}$
\item For each pair of objects $X,Y$, a collection $\mathbf{C}(X,Y)$
  of arrows
\item For each triple of objects $X,Y,Z$, a composition function $\circ : \mathbf{C}(Y,Z) \times \mathbf{C}(X,Y) \rightarrow \mathbf{C}(X,Z)$
\item For each object $X$, an arrow $1_X \in \mathbf{C}(X,X)$
\end{itemize}
such that composition is associative, and for each $X,Y$ and each
$f \in \mathbf{C}(X,Y)$ we have $f \circ 1_X = f = 1_Y \circ f$.
Examples of categories include $\textbf{Set}$, in which the objects are sets and arrows are functions; $\textbf{Vect}_{\mathbb{Q}}$, in which the objects are rational vector spaces and arrows are linear transformations; and $\textbf{TS}$, in which the objects are transition systems and arrows are simulations.

Let $\mathbf{C}$ and $\mathbf{D}$ be categories.  A \textbf{functor}
is a map $F$ sending each object $X$ of $\mathbf{C}$ to an object
$F(X)$ of $\mathbf{D}$ and each arrow $f \in \mathbf{C}(X,Y)$ to an
arrow $F(f) \in \mathbf{D}(F(X),F(Y))$ such that: 
\begin{itemize}
\item For any $f \in \mathbf{C}(Y,Z)$ and $g \in \mathbf{C}(X,Y)$, we have $F(f \circ g) = F(f) \circ F(g)$
\item For each object $X$ of $\mathbf{C}$, we have $F(1_X) = 1_{F(X)}$
\end{itemize}
A functor $F$ is \textbf{faithful} if it is injective on arrows; i.e., for any objects $X$ and $Y$ of $\mathbf{C}$ and any arrows $f,g \in \mathbf{C}(X,Y)$, if $F(f) = F(g)$, then $f=g$.

Suppose that $F$ and $G$ are functors from $\cat{C}$ to $\cat{D}$.  A \textbf{natural transformation} $\eta : F \Rightarrow G$ associates each object of $X \in \Ob{\cat{C}}$
with an arrow $\eta_X \in \cat{D}(F(X),G(X))$ such that for any $f \in \cat{C}(X,Y)$, we have $\eta_Y \circ F(f) = G(f) \circ \eta_X$.

A \textbf{subcategory} of a category $\cat{C}$ is obtained by selecting a sub-collection of objects and arrows that includes all identity arrows of the selected objects and is closed under composition. There is an obvious faithful inclusion functor from a subcategory to the original category that maps each object and arrow in the subcategory to themselves.

%%% Local Variables:
%%% TeX-master: "main"
%%% End:

\section{A Category-Theoretic Foundation of Robust and Sound Analysis} \label{sec:robustness}

\textsc{Informally, } we say that a program analysis $A : \textit{Program} \rightarrow \textit{Property}$ is robust with respect to a class of relationships $R$ if for all programs $P$ and $P'$, if $P$ and $P'$ are related by some $r \in R$, then $A(P)$ and $A(P')$ are also related by $r$.  We formalize this by considering \textit{Program} and \textit{Property} to be
\textit{concrete categories} over a common base category which defines the class of relationships that we are interested in preserving, and defining an operation to be robust when it is a \textit{concrete functor}.

Let $\mathbf{B}$ be a category (the ``base'' category).  A \textbf{concrete category over $\mathbf{B}$} consists of a category $\mathbf{C}$ along with a faithful functor $\proj{C} : \mathbf{C} \rightarrow \mathbf{B}$ (see \cite{Book:AHS2009} for an introduction to concrete category theory).
For any objects $A, B \in \Ob{C}$ and any arrow $f \in \cat{B}(\proj{C}(A),\proj{C}(B))$, we use
$\barr{A}{f}{B}$ to denote that there exists some arrow $g \in \cat{C}(A,B)$ such that $\proj{C}(g) = f$ (note that the choice of $g$ is unique, since $\proj{C}$ is faithful).  Thus, the functor $\proj{C}$
allows us to understand the arrows of $\cat{C}$ as being drawn from $\cat{B}$, in the
sense that for any objects $A,B \in \Ob{C}$,
$\cat{B}(\proj{C}(A),\proj{C}(B))$ defines a space of ``possible'' relationships that may hold between $A$ and $B$, and
$\barr{A}{r}{B}$ is the proposition that the relationship $r$ does hold between $A$ and $B$.  The fact that $\mathbf{C}$ is a category means that these relationships compose:
\begin{itemize}
\item \textit{(Identity)}: for any object $A$ we have
$\barr{A}{1_{\proj{C}(A)}}{A}$.  (Since $\proj{C}$ is a functor, we must have $\proj{C}(1_A) = 1_{\proj{C}(A)}$)
\item \textit{(Composition)}: if $\barr{A}{f}{\barr{B}{g}{C}}$, then
$\barr{A}{g \circ f}{C}$.  (Letting $f' \in \mathbf{C}(A,B)$ and $g' \in \mathbf{C}(B,C)$ be such that $\proj{C}(f') = f$ and $\proj{C}(g') = g$, we have $g' \circ f' \in \cat{C}(A,C)$ and $\proj{C}(g' \circ f') = \proj{C}(g') \circ \proj{C}(f') = g \circ f$).
\end{itemize}

\begin{example} \label{ex:ts}
  The category of transition systems $\mathbf{TS}$ can be understood as a concrete category over $\mathbf{Set}$, where $\proj{TS}$ sends any transition system $T$ to its state space $S_T$, and sends any simulation $f : T \rightarrow U$ to itself (essentially ``forgetting'' that $f$ is a simulation, and just looking at it as a function from $S_T$ to $S_U$).
  For a function $f : S_T \rightarrow S_U$, we write $\barr{T}{f}{U}$ exactly when $f$ is a simulation; i.e., $\barr{T}{f}{U}$ is a proposition meaning that for all $s$ and $s'$ such that $s \rightarrow_T s'$, we have
  $f(s) \rightarrow_U f(s')$.

  Observe that the transitive closure operation is robust in the sense that the state spaces of $T$ and $T^*$ coincide, and if there is a simulation $f : T \rightarrow U$ between two transition systems, $f$ is also a simulation $T^* \rightarrow U^*$.
  The $(-)^\omega$ operation is similarly robust.  Define the category $\mathbf{SubSet}$ where objects are pairs $\tuple{U,S}$ where $U$ is a set and $S$ is a subset of $U$, and where an arrow $f : \tuple{U,S} \rightarrow \tuple{U',S'}$ is a function $U \rightarrow U'$ such that for all $x \in S$, we have $f(x) \in S'$.  Observe that $\mathbf{SubSet}$ is a concrete category over $\mathbf{Set}$, where $\proj{SubSet}$ maps each object $\tuple{U,S}$ to $U$, and each
  arrow $f : \tuple{U,S} \rightarrow \tuple{U',S'}$ to itself.  Then $(-)^\omega : \mathbf{TS} \rightarrow \mathbf{SubSet}$ satisfies the property that if
  $\barr{T}{f}{U}$, then $\barr{T^\omega}{f}{U^\omega}$.
\end{example}

The above example motivates the formal definition of robustness:

\begin{definition}[Robustness]
  Let $\mathbf{B}$ be a category, and let $\mathbf{C}$ and $\mathbf{D}$ be concrete categories over $\mathbf{B}$.  A function $F : \Ob{C} \rightarrow \Ob{D}$ is \textbf{robust} (with respect to $\mathbf{B}$, or arrows in $\mathbf{B}$) if for every object $X \in \Ob{C}$ we have
  $\proj{C}(X) = \proj{D}(F(X))$ and
  for every pair of objects $X,Y \in \mathbf{C}$ and every arrow $f \in \cat{B}(\proj{C}(X),\proj{C}(Y))$ such that
  $\barr{X}{f}{Y}$, we have
  $\barr{F(X)}{f}{F(Y)}$.
\end{definition}

Robustness coincides with an existing notion in category theory called \textit{concrete functors}.  For concrete categories $\mathbf{C}$ and $\mathbf{D}$ over some common base category $\mathbf{B}$, a functor $F : \mathbf{C} \rightarrow \mathbf{D}$ is \textbf{concrete} if $\proj{C} = \proj{D} \circ F$. We show that robust functions are the same as concrete functors.

\begin{restatable}[Robust operations are concrete functors]{lemma}{robustImpliesConcreteFunctor}
\label{lem:robust-implies-concrete-functor}
Let $\mathbf{B}$ be a category, and let $\mathbf{C}$ and $\mathbf{D}$ be concrete categories over $\mathbf{B}$.  Suppose that $F : \Ob{C} \rightarrow \Ob{D}$ is robust.
For any
$f \in \mathbf{C}(X,Y)$, define $F(f)$ to be the arrow $g \in \mathbf{D}(F(X),F(Y))$ such that $\proj{C}(f) = \proj{D}(g)$---the choice of $g$ is unique since $\proj{D}$ is faithful by assumption.  Then $F$ is a concrete functor.
\end{restatable}

Example~\ref{ex:ts} illustrates that semantic operations like reflexive transitive closure are robust; however, they are typically not computable.  Our goal in program analysis is to develop approximate analyses that are computable, while retaining \textit{some} robustness properties.

\begin{example}[Robustness under linear transformations] \label{ex:tf-b-category}
Suppose that we are interested in an over-approximating transitive closure operation for transition formulas that is robust with respect to linear transformations.  We endow the set of transition formulas with the structure of a concrete category over $\mathbf{Vect}_{\mathbb{Q}}$ as follows.  Let $\TF$ be a category where the objects are transition formulas over any set of variables (note that we use $\TF$ for the category and $\TF(X)$ for the set of transition formulas over a specific variable set $X$).  If $F \in \TF(X)$ and $G \in \TF(Y)$ are transition formulas over $X$ and $Y$ respectively, then the set of arrows between $F$ and $G$ is the set of linear maps $\mathbb{Q}^X \rightarrow \mathbb{Q}^Y$ that are simulations between the transition systems defined by $F$ and $G$ (i.e., $x \rightarrow_F x'$ implies $f(x) \rightarrow_G f(x')$).
 Observe that $\TF$ forms a concrete category over $\mathbf{Vect}_{\mathbb{Q}}$, where the concrete functor $\proj{TF}$ sends each transition formula $F \in \TF(X)$ to its corresponding state space $\mathbb{Q}^X$, and each simulation $f : F \rightarrow G$ to itself.  $\TF$ can be understood as a subcategory of $\mathbf{TS}$, which contains only transition systems that are representable using transition formulas, and only linear simulations as arrows.
 As we will see in Section~\ref{sec:lifting}, the iteration operator $(-)^\pstar$ (Section~\ref{sec:overview}) is robust with respect to $\mathbf{Vect}_{\mathbb{Q}}$.
\end{example}

\begin{example}[Invariance under renaming] \label{ex:cp-not-robust}
Consider constant propagation, which is robust w.r.t. variable renaming.
Formally, define a category of flow graphs $\textbf{FG}_{\alpha}$: objects are flow graphs weighted with transition formulas, and an arrow from $G$ (over variables $X$) to a structurally isomorphic $H$ (over $Y$) is a bijection $\sigma : X \rightarrow Y$ such that $w_G(e)[\sigma] \equiv w_H(e)$ for each edge $e$.
Given a flow graph $G$ with weights in $\TF(X)$, the goal of constant propagation is to compute a function
$k \in \textit{Const}(V_G,X) \defeq V_G \rightarrow (\set{\bot} \cup (X \rightarrow (\mathbb{Z} \cup \set{\top})))$,
which assigns to each control location $v \in V_G$ either the value $\bot$ (indicating that $v$ is statically unreachable) or a \textit{constant environment} that assigns each variable either an integer (the variable must always be equal to that value at $v$) or $\top$ (the variable may take on more than one value).
Let $\textbf{Bij}$ denote the category whose objects are finite sets (of variables) and whose arrows are bijections between them.
Let $\textbf{Const}$ be the category whose objects are constant assignments (belonging to $\textit{Const}(V,X)$ for some set of vertices $V$ and set of variables $X$), and where an arrow from $k \in \textit{Const}(V,X)$ to $k' \in \textit{Const}(V,Y)$ is a bijection $\sigma : X \rightarrow Y$ such that
$k(v)(x) = k'(v)(\sigma(x))$ for all $v \in V$ and $x \in X$.
Both $\textbf{FG}_{\alpha}$ and $\textbf{Const}$ are concrete categories over $\textbf{Bij}$, and constant propagation is a concrete functor from $\textbf{FG}$ to $\textbf{Const}$.
If we are interested in robustness with respect to linear transformations rather than variable renaming, then we may analogously consider a base category of $\textbf{Vect}_{\mathbb{Q}}$.
However, constant propagation is \textit{not} robust with respect to linear transformations.
For instance, the programs $P$ and $P'$ in Figure~\ref{fig:cp} are isomorphic under the linear map $(x,y) \mapsto (a{+}b, a{-}b)$; constant propagation correctly determines $x = 1$ in $P$, but cannot recover the corresponding fact $a + b = 1$ in $P'$.\footnote{Affine relations analysis, on the other hand, \textit{is} robust; one might think of affine relations analysis as the minimal generalization of constant propagation that is robust with respect to linear transformations.}
\end{example}

\begin{example}[Monotonicity] \label{ex:interval-monotonicity}
    Interval analysis computes an interval bounding the possible values of each variable at each control location.
    Define a category $\textbf{FG}_{mon}$ in which objects are flow graphs and there is a unique arrow $G \rightarrow H$ between structurally isomorphic graphs iff $w_G(e) \models w_H(e)$ for all edges $e$ (i.e., arrow in $\textbf{FG}_{mon}$ represent refinement).
    Similarly, define a category $\textbf{Intv}$ of interval assignments, with a unique arrow from $I$ (over variables $X$) to $J$ (over $Y$) iff $X \subseteq Y$ and $I(v)(x) \subseteq J(v)(x)$ for all $v$ and $x \in X$---i.e., $I$ is at least as precise as $J$ on the common variables.
    Both categories are concrete over $\textbf{FinSub}$ (finite sets of variables where arrows represent inclusion).
    Robustness with respect to $\textbf{FinSub}$ then requires that whenever $G$'s formulas are stronger than $H$'s, the analysis of $G$ is at least as precise as that of $H$---this is \textit{monotonicity}.
    The iterative interval analysis \cite{Cousot1977} is \textit{not} robust in this sense since it uses widening \cite{VMCAI:Cousot2015}, which is inherently non-monotone; on the other hand, policy iteration~\cite{CAV:CGGMP2005} (which computes strongest interval invariants) \textit{is} robust.
\end{example}

\begin{figure}
\begin{subfigure}{0.45\textwidth}
\begin{lstlisting}[style=base]
if (y > 0) {
  y = 0;
  x = y + 1;
} else { x = 1; }
|Post: $\set{x \mapsto 1, y\mapsto \top}$|\end{lstlisting}
\caption{A program $P$}
\end{subfigure}
\begin{subfigure}{0.5\textwidth}
\begin{lstlisting}[style=base]
if (a-b > 0) {
  a = (a+b)/2; b = a;
  a = a-b+1/2; b = 1/2;
} else { a = (a-b+1)/2; b = 1-a; }
|Post: $\set{a \mapsto \top, b\mapsto \top}$|\end{lstlisting}
\caption{A program $P'$}
\end{subfigure}
\caption{An example illustrating that constant propagation is not robust under linear transformations: The programs $P$ and $P'$ are isomorphic under the linear substitution $[x \mapsto a+b, y\mapsto a-b]$; however the result of constant propagation analysis on $P$ is not isomorphic to that of $P'$.  \label{fig:cp}}
\end{figure}

\paragraph{Using the Framework} 
We envision two modes for using this framework.  In the first mode, an analysis designer has a class of transformations of interest, and uses the framework to inform the design of a robust analysis (perhaps facilitated by the recipe of Section~\ref{sec:lifting}).
We expect that for the most part, analysis designers will use a small number of choices for the base category $\mathbf{B}$: e.g., $\textbf{Vect}_{\mathbb{Q}}$ captures robustness under linear transformations (Example~\ref{ex:tf-b-category}), $\textbf{Bij}$ captures robustness under variable renaming, and $\textbf{FinSub}$ captures monotonicity.
In the second mode, an analysis designer has an analysis of interest, and uses the framework to characterize the class of transformations under which it is robust.  For instance, the analyses in Table~\ref{tbl:instances} were designed with monotonicity in mind, but in the next section we show that in fact they robust with respect to the richer class of linear simulations.

\subsection{Soundness of Robust Analysis}

The main novelty of this section is to define what it means for a program analysis to be robust.  To understand what it means to be a program analysis, we must also define the meaning of \emph{soundness}, and it is convenient to do so in the same category-theoretic vocabulary that we use for robustness.  This section demonstrates how any concrete category can be equipped with an ``approximation pre-order'' to formalize soundness.

Example~\ref{ex:tf-b-category} is suggestive of a general mechanism by which we arrive at the structure of a concrete category.   While we might think of an algebraic analysis as a function mapping a program into the algebra of transition formulas, there is in fact a family of such algebras, parameterized by the set of variables over which they are defined.\footnote{The same phenomenon appears in iterative abstract interpretation.  For example, although we speak of ``the'' lattice of polyhedra, it is really a family of lattices, one for each dimension.}  We obtain $\TF$ as the disjoint union of transition formulas over all parameters, and the functor $\proj{TF}$ sends each transition formula back to its parameter.  Thus, we can recover the algebra $\TF(X)$ as the \textit{fiber} of $\mathbb{Q}^X$,
 $\proj{TF}^{-1}(\mathbb{Q}^X) \defeq \set{ F : \proj{TF}(F) = \mathbb{Q}^X } = \TF(X)$.  We recover the entailment relation on transition formulas as
 $F \models G \iff $ there is some arrow $f$ in $\TF$ between $F$ and $G$ such that $\proj{TF}(f) = 1_{\proj{TF}(X)}$.

More generally, for any concrete category $\mathbf{C}$ over $\mathbf{B}$, the fibers of $\mathbf{C}$ can be equipped with an ``approximation pre-order.''  That is, for any object $B \in \Ob{B}$, and any $C,C' \in \proj{C}^{-1}(B)$ write $C \preceq_B C'$ if and only if $\barr{C}{1_{\proj{C}(C)}}{C'}$ (i.e., there is an arrow $f \in \mathbf{C}(C,C')$ such that $\proj{C}(f) = 1_{\proj{C}(C)}$).  We use $C \preceq C'$ to denote that $\proj{C}(C) = \proj{C}(C')$ and $C \preceq_{\proj{C}(C)} C'$.  Observe that $\preceq$ is a pre-order (where reflexivity and transitivity follow from the \textit{Identity} and \textit{Composition} properties of the $\barr{}{1_{\proj{C}(C)}}{}$ relation, respectively).

The basic principle espoused by the theory of abstract interpretation~\cite{POPL:CC1977} is that soundness is a connection between a ``concrete semantics'' (the fundamental matter of interest, which is typically not computable) and an ``abstract semantics'' (a computable \textit{approximation} of the concrete semantics).  In our category-theoretic vocabulary, we may think of a concrete semantics as an object of some concrete category $\mathbf{R}^\natural$, and an abstract semantics as an object of another concrete category $\mathbf{R}^\sharp$.  A concrete functor $\gamma : \mathbf{R}^\sharp \rightarrow \mathbf{R}^\natural$ (the \textit{concretization} functor) maps each abstract object to its concrete meaning (e.g., mapping a polyhedron to the set of states it contains); being a concrete functor means that this mapping is compatible with arrows.  We say that an abstract semantics $a \in \Ob{R^\sharp}$ approximates a concrete semantics $c \in \Ob{R^\natural}$ if $c \preceq \gamma(a)$ (where $\preceq$ is the approximation order on $\mathbf{R}^\natural$).  Finally, if we let $\mathbf{C}$ be a category of programs,
$F^\natural : \mathbf{C} \rightarrow \mathbf{R}^\natural$ be a functor interpreting programs as their concrete semantics
and $F^\sharp : \mathbf{C} \rightarrow \mathbf{R}^\sharp$ be a functor interpreting programs as their abstract semantics, then
$F^\sharp$ is sound with respect to $F^\natural$ when
$F^\sharp(P)$ approximates $F^\natural(P)$ for all programs $P$.
Formally,
\begin{definition}[Soundness] \label{def:soundness}
    Let $\mathbf{C}$, $\mathbf{R}^\sharp$, and $\mathbf{R}^\natural$ be concrete categories over a common base category $\mathbf{B}$, and let $\gamma : \mathbf{R}^\sharp \rightarrow \mathbf{R}^\natural$,
    $F^\sharp : \mathbf{C} \rightarrow \mathbf{R}^\sharp$,
    and
    $F^\natural : \mathbf{C} \rightarrow \mathbf{R}^\natural$
    be concrete functors.
    We say that $F^\sharp$ is \textbf{sound} with respect to $F^\natural$
    iff for all $P \in \Ob{C}$, we have $F^\natural(P) \preceq \gamma(F^\sharp(P))$.
\end{definition}

Again, this notion of soundness coincides with an existing notion in category theory, called \textit{concrete natural transformations}.  For concrete categories $\mathbf{C}$ and $\mathbf{D}$ and concrete functors $F, G : \cat{C} \rightarrow \cat{D}$, a natural transformation $\epsilon : F \Rightarrow G$ is called concrete if for each object $A$ of $\cat{C}$, $\proj{D}(\epsilon_A) = 1_{\proj{D}(F(A))}$.

\begin{restatable}{lemma}{soundnessfromnattransformation} \label{lem:soundness-from-nat-transformation}
  $F^\sharp$ is sound w.r.t. $F^\natural$ iff there is a concrete natural transformation $\epsilon : F^\natural \Rightarrow \gamma \circ F^\sharp$.
\end{restatable}

As we will see in Section~\ref{sec:algebraic-laws-main}, concrete natural transformations are also a convenient mechanism with which to express algebraic laws.

\subsection{Best Abstractions } \label{sec:best-abstraction-weak-left-adj}
Another key element of abstract interpretation is the notion of ``best abstractions,'' which are formalized using Galois connections.  Intuitively, the ``best'' abstraction is the universal object among all abstractions, such that all other abstractions ``factor through'' it.  In this section, we formulate a notion of best abstractions for the purpose of program analysis in the language of concrete categories.

Best abstractions can be formalized as \textit{weak left adjoints} \cite{MZ:Kainen1971}.  We give a definition that is specialized towards our setting below, and provide a proof of equivalence to the standard definition in the appendix.

Let $\mathbf{C}^\natural$ and $\mathbf{C}^\sharp$ be concrete categories over some common base category $\cat{B}$, and let $\gamma : \mathbf{C}^\sharp \rightarrow \mathbf{C}^\natural$ be a concrete functor (one may think of $\mathbf{C}^\sharp$ as an abstraction of $\mathbf{C}^\natural$, and $\gamma$ as mapping each abstraction back to a representative in $\mathbf{C}^\natural$).
A \textbf{weak left adjoint} of $\gamma$ is a pair $\tuple{\alpha,\eta}$ where $\alpha$ associates every object $A \in \Ob{\mathbf{C}^\natural}$ with an object $\alpha(A) \in \Ob{\mathbf{C}^\sharp}$, and $\eta$ associates every object $A \in \Ob{\mathbf{C}^\natural}$ with an arrow $\eta_A \in \mathbf{B}(\proj{C^\natural}(A),\proj{C^\sharp}(\alpha(A)))$ such that the following properties hold.
First, we must have $\barr{A}{\eta_A}{\gamma(\alpha(A))}$---i.e., $\gamma(\alpha(A))$ approximates $A$ modulo the transformation $\eta_A$.  Second, $\alpha(A)$ is the closest approximation to $A$ within $\cat{C}^\sharp$, in the sense of the following universal property: for any object $B \in \cat{C}^\sharp$ that approximates $A$ modulo any transformation $f$, there is a way to factor the transformation $f$ as $\overline{f} \circ \eta_A$ such that $B$ approximates $\alpha(A)$ modulo $\overline{f}$.  More precisely:

\smallskip
\noindent\begin{minipage}{0.3\textwidth}
   \begin{tikzpicture}[thick]
     \matrix (m) [matrix of math nodes, row sep=2em, column sep=2.5em,
       text height=1.5ex, text depth=0.25ex] { & \proj{C^\sharp}(B)\\ \proj{C^\natural}(A) &
       \proj{C^\sharp}(\alpha(A))\\ }; \draw (m-2-1) edge[->] node[below]{$\eta_A$} (m-2-2); \draw (m-2-1) edge[->] node[above
       left]{$f$} (m-1-2); \draw (m-2-2) edge[->, dashed]
     node[right]{$\exists \overline{f}$} (m-1-2);
   \end{tikzpicture}
\end{minipage}%
\begin{minipage}{0.68\textwidth}
for any object $B \in \mathbf{C}^\sharp$ and any arrow $f \in \cat{B}(\proj{C^\natural}(A),\proj{C^\sharp}(B))$ such that $\barr{A}{f}{\gamma(B)}$, there is an arrow $\overline{f}$ such that $\barr{\alpha(A)}{\overline{f}}{B}$ and the diagram to the left commutes.
Note that $\alpha$ is not necessarily a functor, and $\overline{f}$ is not necessarily unique (if both hold, we recover the usual definition of an adjunction).
\end{minipage}
\smallskip

\begin{figure}[t]
\centering
  \begin{tikzpicture}
    \node (program) { \begin{minipage}{4.5cm}
  \begin{lstlisting}[style=base]
while (x >= 0 && y >= 0) {
  if (nondet()) {
     x -= z;
  } else {
     y -= z;
  }
  z++;
}
\end{lstlisting}
\end{minipage}};
    \node [right of=program, node distance=3.5cm,yshift=-0.5cm] (tf) { \(\begin{array}{l@{}l}
      &x \geq 0 \land y \geq 0\\
      \land & \left( \begin{array}{l@{}l}
       & (x' = x - z \land y' = y)\\
       \lor & (x' = x \land y' = y - z) \end{array} \right)\\
      \land & z' = z + 1
    \end{array}\)};
    \draw [decorate, decoration={brace, mirror, amplitude=4pt}]
      (tf.south west) -- (tf.south east) node [midway, below=4pt] {$F$};

  \node [right of=tf, node distance=6cm] (lds) { \( \begin{array}{l@{}l}
      &a \geq 0\\
      \land & a' = a - b\\
      \land & b' = b + 1
\end{array} \)};
    \draw [decorate, decoration={brace, mirror, amplitude=4pt}]
      (lds.south west) -- (lds.south east) node [midway, below=4pt] {$\alpha(F)$};

        \node [above of=lds,yshift=1cm] (dats2) { $q' = q + 2$ };

  \draw (tf) edge[->,thick] node[below]{$\underbrace{[a \mapsto x+y, b \mapsto z]}_{\eta_F} $} (lds);

    \draw (tf) edge[->,thick] node[above left]{$ [q \mapsto 2z] $} (dats2);
    \draw (lds) edge[->,thick,dashed] node[right]{$ [q \mapsto 2b] $} (dats2);

  \end{tikzpicture}
  \caption{Over-approximation of a loop by a deterministic affine transition systems. \label{fig:dats}}
\end{figure}
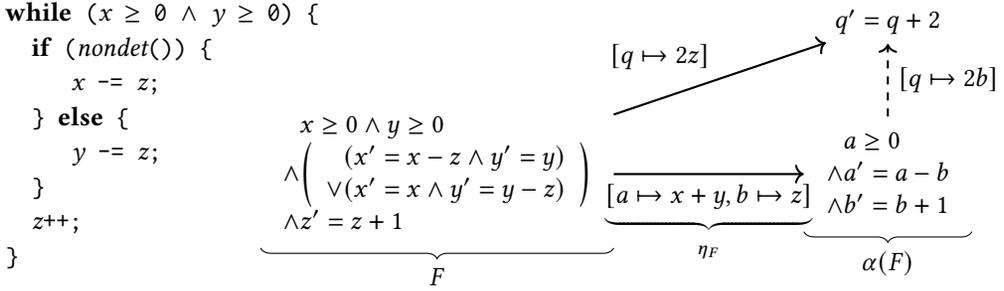

\begin{example} \label{ex:dats-adjoint}
\citet{CAV:ZK2021} present a termination analysis which is based on finding best abstractions of loops as deterministic affine transition systems (DATS).  Define a DATS over a set of variables $X$ to be a transition formula of the form $G \land T \in \TF(X)$ where $G$ is a guard (expressed over the pre-state variables $X$) and $T$ represents a linear transformation
$T = \bigwedge_{x \in X} x' = t_x + b_x$ (where $t_x$ is a linear term over $X$ and $b_x \in \mathbb{Q}$).
Let $\mathbf{DATS}$ denote the category of deterministic affine transition systems where arrows are linear simulations.
$\mathbf{DATS}$ is a subcategory of $\cat{TF}$, and so has a natural inclusion functor $\gamma : \cat{DATS} \rightarrow \cat{TF}$, and $\gamma$ has a weak left adjoint $\tuple{\alpha,\eta}$.
Figure~\ref{fig:dats} depicts an example of a transition formula $F$
along with its best abstraction, consisting of a DATS $\alpha(F)$ and a linear simulation $\eta_F$.  Observe that the DATS $\alpha(F)$ does not over-approximate $F$ (in the sense that $F$ does not entail $\alpha(F)$)--rather $\alpha(F)$ over-approximates $F$ \textit{modulo the transformation $\eta_F$} (i.e., $F$ entails $\alpha(F)[\eta_F,\eta_F']$).
\citet{CAV:ZK2021}'s termination analysis proves that the loop is terminating via asymptotic analysis of $\alpha(F)$.

The figure also depicts a second abstraction of $F$, consisting of a DATS $q' = q + 2$ and simulation $[q \mapsto 2x]$.
Intuitively, $(\alpha(F),\eta_F)$ is a closer approximation to $F$ than this second abstraction, and this fact is witnessed by the simulation $[q \mapsto 2b]$ from $\alpha(F)$ to $q' = q+2$.   One might thus think of weak left adjoints as Galois connections \textit{modulo transformations}.
\end{example}

%%% Local Variables:
%%% TeX-master: "main"
%%% End:

\section{A Lifting Strategy for Robust Analyses}
\label{sec:lifting}
This section presents a general strategy for obtaining robust  analyses.  The high-level idea is that if we have two concrete categories $\cat{C}$ and $\cat{D}$, and a concrete functor $\gamma : \cat{D} \rightarrow \cat{C}$ with a weak left adjoint, then a robust operation that is defined on $\mathbf{D}$ can be ``lifted'' to a robust operation that is defined on $\mathbf{C}$ (under certain technical assumptions).  For instance, one might think of $\mathbf{C}$ as a category of programs and $\cat{D}$ as a subcategory of programs of a restricted form, for which a (robust) operation $F$ is computable (for instance, programs in $\mathbb{D}$ have decidable reachability problems).  Our recipe ``lifts'' $F$ to an \textit{over-approximate} operation $\overline{F}$ that is defined for all programs (rather than just those in $\mathbf{D}$).  The analysis designer is responsible for the creative work of defining $\mathbf{D}$ and $F$ and showing that the functor $\gamma : \cat{D} \rightarrow \cat{C}$ has a (computable) weak left adjoint, but needn't prove soundness or robustness (since they follow from Theorems~\ref{thm:lift-soundness} and~\ref{thm:lift-robustness}).

We set the stage by considering the motivation behind a unified construction of several robust loop analyses (Section~\ref{sec:lt}).
We then describe the general strategy (Section~\ref{sec:lifted-op}) and formulate a soundness theorem (Section~\ref{sec:lifted-soundness}).  Finally, we show that lifted operations can be combined, and the lifting preserves algebraic laws (Section~\ref{sec:algebraic-laws-main}).

\subsection{Examples of Lifted Analyses} \label{sec:lt}

We begin by showing how the analyses introduced in Section~\ref{sec:overview} arise as instances of a common ``lifting'' pattern.  In the following, $\TF$ refers to the category in which the objects are transition formulas expressed in linear arithmetic and arrows are linear simulations (Example~\ref{ex:tf-b-category}). We think of $\textbf{TF}$ as the category that we wish to lift an operator to; the category that we lift from, and the operation we wish to lift, will vary.

\begin{paragraph}{Polyhedral Guard Analysis}

Define $\cat{PolyCart}$ (polyhedral cartesian transition formulas)
to be the subcategory of $\TF$
consisting of formulas of the form
$(\bigwedge_{i=1}^n s_i \geq a_i) \land (\bigwedge_{j=1}^m t_j \geq b_j)$ where each $s_i$ is a linear term over unprimed variables and each $t_i$ is a linear term over primed variables.
Observe that every polyhedral cartesian transition formula is transitively closed, so we can define a reflexive transitive closure operator $(-)^\star : \cat{PolyCart} \rightarrow \TF$ as $ F^\star \defeq X' = X \lor F$
for any $F$ (over the variables $X$).
Let $\gamma : \cat{PolyCart} \rightarrow \TF$ be the inclusion functor.
Observe that $\tuple{\alpha,\eta}$ is a weak left adjoint to $\gamma$ (see Corollary~\ref{cor:polyhedral-guard-analysis-sound-robust} in Appendix~\ref{sec:proofs}), where
$\alpha(G) \defeq \textit{Pre}(G) \land \textit{Post}(G)$, and $\eta$ is the identity (since $\alpha(G)$ lives in the same variable space as $G$, the substitution $[\eta_G,\eta_G']$ has no effect).  Polyhedral guard analysis can be understood as the operator
$G^{\star_{\text{PGA}}} = (\alpha(G)^\star)[\eta_G,\eta_G']$.
\end{paragraph}

\begin{paragraph}{Linear Recurrence Analysis}
Define $\mathbf{LT}_\bot$ (lossy translations) to be the subcategory of $\mathbf{TF}$ whose objects are formulas of the form $\bigwedge_{x \in X} x' \leq x + a_x$ or $\false$.
Computing the reflexive transitive closure of a formula in $\mathbf{LT}_\bot$ is straightforward, and gives rise to a concrete functor $(-)^\star : \mathbf{LT}_\bot \rightarrow \mathbf{TF}$, where for a formula $F \in \mathbf{LT}_\bot$ over variables $X$:
\[
\false^\star \defeq \bigwedge_{x \in X} x' = x \hspace*{2cm}
\left(\bigwedge_{x \in X} x' \leq x + a_x\right)^\star \defeq \exists k \in \mathbb{N}. k \geq 0 \land \bigwedge_{x \in X} x' \leq x + k a_x
\]

Let $\gamma : \mathbf{LT}_\bot \rightarrow \mathbf{TF}$ be the inclusion functor.
Observe that $\tuple{\alpha,\eta}$ is a weak left adjoint of $\gamma$ (see Corollary~\ref{cor:linear-rec-analysis-sound-robust} in Appendix~\ref{sec:proofs}), where
for any transition formula $G$ with $\textit{conv}(\Delta(G)) = \bigwedge_{i=1}^n t_i \leq b_i$,
$\alpha(G)$ is $\bigwedge_{i=1}^n y_i' \leq y_i + b_i$ (over fresh variables $y_1,\dots,y_n$) and
$\eta_G$ is the linear map corresponding to the substitution $\set{ y_i \mapsto \hat{t}_i }_{i=1}^n$ where $\hat{t}_i \defeq t_i[\Delta_X \mapsto X]$.
Intuitively, $\alpha(G)$ introduces a fresh variable for each recurrence relation, and
the $\eta_G$ represents the correspondence between these fresh variables and the linear terms they represent.
Linear recurrence analysis can be understood as the operator
$G^{\star_{\text{LRA}}} = (\alpha(G)^\star)[\eta_G,\eta_G']$.  (Note the right-hand side of this equation is syntactically identical to that of $(-)^{\star_{\text{PGA}}}$ above, but $\alpha$, $\eta$, and $(-)^\star$ are defined with respect to the category $\cat{LT}_\bot$ instead of $\cat{PolyCart}$).
\end{paragraph}

\begin{example} \label{ex:pga-lra-trace}
Table~\ref{tab:summaries} illustrates the steps of the above analyses.
Consider the loop body $G_2$ of $P_2$.
For \textit{LRA}, we compute $\Delta(G_2)$ and its convex hull, which gives $-\delta_x \leq -1 \land \delta_x \leq 1 \land \delta_y \leq 5$.  Each inequality $t_i \leq b_i$ yields a fresh variable $y_i$, giving $\alpha(G_2) = (y_1' \leq y_1 - 1 \land y_2' \leq y_2 + 1 \land y_3' \leq y_3 + 5)$ with $\eta = [y_1 \mapsto -x,\, y_2 \mapsto x,\, y_3 \mapsto y]$.  Applying the star and substituting via $[\eta_{G_2},\eta_{G_2}']$ yields the recurrence $\exists k \geq 0.\, x' = x + k \land y' \leq y + 5k$.
Observe that the linear simulation $\sigma = [i \mapsto x]$ from $G_2$ to $G_1$ (Example~\ref{ex:linear-sim-p1-p2}) is preserved: $\sigma$ is also a simulation from $G_2^{\star_{LRA}}$ to $G_1^{\star_{LRA}}$ (a consequence of Theorem~\ref{thm:lift-robustness} in the following).
\end{example}

\begin{paragraph}{Polyhedral Iteration Operator}
The polyhedral iteration operator defined in Section~\ref{sec:overview} combines polyhedral guard analysis and linear recurrence analysis, resulting in an analysis that produces more precise invariants than both combined.  This can be understood as an iteration operator induced by the product category $\textbf{PolyCart} \times \textbf{LT}_\bot$ (see Section~\ref{sec:algebraic-laws-main} and Example~\ref{ex:polyhedral-iteration-operator-through-combination} in the Appendix).
\end{paragraph}

\begin{paragraph}{Termination analysis via DATS Abstraction}

\citet{CAV:ZK2021}'s termination analysis (Example~\ref{ex:dats-adjoint}) lifts a termination analysis $(-)^\omega : \cat{DATS} \rightarrow \cat{SF}$ for deterministic affine transition systems to general transition formulas
$(-)^{\overline{\omega}} : \cat{TF} \rightarrow \cat{SF}$, by defining
$F^{\overline{\omega}} \defeq (\alpha(F))^\omega[\eta_F]$, where $\tuple{\alpha,\eta}$ is weak left adjoint to the inclusion functor $\gamma : \cat{DATS} \rightarrow \cat{TF}$.
\end{paragraph}

\subsection{Lifted Operations and Robustness} \label{sec:lifted-op}

Table~\ref{tbl:instances} lists more examples of ``lifting'' an analysis from a sub-Turing model of computation.  This section formalizes this pattern, providing a unified understanding of these analyses and an account of why they are sound and robust.

Suppose that $\cat{D}$ is a subcategory of the category of transition formulas $\TF$ (Example~\ref{ex:tf-b-category}), such that the inclusion functor $\gamma : \cat{D} \rightarrow \TF$ has a weak left adjoint $\tuple{\alpha,\eta}$.
The previous section gave examples in which a loop summarization analysis was lifted from an operator $(-)^\star : \cat{D} \rightarrow \TF$ by defining
$F^{\lifted{\star}} \defeq (\alpha(F)^\star)[\eta_F,\eta_F']$, and termination analysis was lifted from an operator $(-)^\omega : \cat{D} \rightarrow \textbf{SF}$ by defining $F^{\overline{\omega}} \defeq (\alpha(F)^\omega)[\eta_F]$.  Naturally, this pattern can be generalized to arbitrary operations and concrete categories.  The picture below depicts the general pattern of ``lifted'' operations (center) along with these two particular cases of loop summarization (left) and termination (right).
\begin{center}
  \begin{tikzpicture}[thick]
    \matrix (m) [matrix of math nodes, row sep=1.8em, column sep=2em,
      text height=1.5ex, text depth=0.25ex] { \mathbf{D} & \mathbf{TF}\\
                                              \mathbf{TF}  & \mathbf{Vect}_{\mathbb{Q}}\\ };
    \draw (m-2-2) edge[<-] node[right]{$\phi_{\mathbf{TF}}$} (m-1-2);
    \draw (m-1-1) edge[->] node[above]{$\star$} (m-1-2);
    \draw (m-2-1) edge[<-] node[left]{$\gamma$} (m-1-1);
    \draw (m-2-1) edge[->] node[below]{$\phi_{\mathbf{TF}}$} (m-2-2);
    \draw (m-2-1) edge[->,dashed] node[above left]{$\overline{\star}$} (m-1-2);
  \end{tikzpicture}
  \hspace*{0.5cm}
  \begin{tikzpicture}[thick]
  \matrix (m) [matrix of math nodes, row sep=1.8em, column sep=2em,
      text height=1.5ex, text depth=0.25ex] { \mathbf{D} & \mathbf{R}\\
                                              \mathbf{C}  & \mathbf{B}\\ };
    \draw (m-2-2) edge[<-] node[right]{$\proj{R}$} (m-1-2);
    \draw (m-1-1) edge[->] node[above]{$F$} (m-1-2);
    \draw (m-2-1) edge[<-] node[left]{$\gamma$} (m-1-1);
    \draw (m-2-1) edge[->,dashed] node[above left]{$\overline{F}$} (m-1-2);
    \draw (m-2-1) edge[->,dashed] node[below]{$\proj{C}$} (m-2-2);
    \draw (m-1-1) edge[->] node[above]{$F$} (m-1-2);
  \end{tikzpicture}
  \hspace*{0.5cm}
  \begin{tikzpicture}[thick]
  \matrix (m) [matrix of math nodes, row sep=1.8em, column sep=2em,
      text height=1.5ex, text depth=0.25ex] { \mathbf{D} & \mathbf{SF}\\
                                              \mathbf{TF}  & \mathbf{Vect}_{\mathbb{Q}}\\ };
    \draw (m-2-2) edge[<-] node[right]{$\phi_{\mathbf{SF}}$} (m-1-2);
    \draw (m-1-1) edge[->] node[above]{$\omega$} (m-1-2);
    \draw (m-2-1) edge[<-] node[left]{$\gamma$} (m-1-1);
    \draw (m-2-1) edge[->] node[below]{$\phi_{\mathbf{TF}}$} (m-2-2);
    \draw (m-2-1) edge[->,dashed] node[above left]{$\overline{\omega}$} (m-1-2);
  \end{tikzpicture}
  \end{center}

\begin{table}\footnotesize
\caption{Summary of robust analyses that can be understood as ``lifted'' from a sub-Turing model.  As a consequence of Theorem~\ref{thm:lift-robustness}, all listed analyses are robust with respect to linear simulations. \label{tbl:instances}}
\begin{tabular}{l|l|l}\toprule
Source & Model & ``Lifted'' Analysis\\ \midrule
& Lossy translation & Reflexive transitive closure\\
& Cartesian relations & Reflexive transitive closure\\
\cite{CAV:ZK2021} & Deterministic affine transition systems & (Eventual) non-termination analysis\\
\cite{CAV:SK2019} & $\begin{array}{@{}l@{}}\text{Rational Vector Addition System}\\ \hspace*{5pt}\text{with Resets ($\mathbb{Q}$-VASR)}\end{array}$ & Reflexive transitive closure\\
\cite{OOPSLA:PK2024} & Labeled $\mathbb{Q}$-VASR & Context-free reachability\\
\cite{POPL:CK2024} & Solvable transition ideals & Polynomial invariant generation
\\\bottomrule
\end{tabular}
\end{table}

The last prerequisite for formalizing ``lifting'' is a category-theoretic reformulation of the translation step ($[\eta,\eta']$ for loop summarization, $[\eta]$ for termination analysis).
Intuitively, given a concrete category $\mathbf{R}$ over $\mathbf{B}$, an object $B$ of $\mathbf{R}$, and an arrow $f \in \mathbf{B}(X,\proj{R}(B))$, we may think of the inverse image of $B$ under $f$ as the universal object $A \in \proj{R}^{-1}(X)$ such that there is some $g \in \mathbf{R}(A,B)$ with $\proj{R}(g) = f$ (i.e., $A$ is the greatest object that looks like an ``inverse'' of $X$ with respect to approximation pre-order in $\cat{R}$).  This intuition is formalized using \textit{cartesian morphisms} and \textit{fibred categories}
(due to \citet{SGA1}; see \citet{Book:BW1995} for an introduction---we specialize to the language of concrete categories below).

Let $\mathbf{C}$ be a concrete category over $\mathbf{B}$.
Let $B$ be an object of $\mathbf{C}$, and $f : X \rightarrow \proj{C}(B)$ be an arrow of $\mathbf{B}$.  A $\cat{C}$-arrow $u : A \rightarrow B$ is \textbf{cartesian} (for $f$ and $B$) if $\proj{C}(A) = X$, $\proj{C}(u) = f$, and for any object $Z \in \mathbf{C}$ and arrow $h : \proj{C}(Z) \rightarrow \proj{C}(B)$ such that
$\barr{Z}{h \circ f}{B}$, we have $\barr{Z}{h}{A}$.
We use $\bcarr{A}{f}{B}$ to denote that there exists a cartesian arrow $u \in \cat{C}(A,B)$ for $f$ and $B$.
A particular case that will be used in this paper (obtained by specializing $h$ to the identity morphism) is:
\begin{itemize}
\item[($\dagger$)] \label{cartesian} If $\bcarr{A}{f}{B}$ and $\barr{Z}{f}{B}$, then $Z \preceq A$.
\end{itemize}

Say that a concrete category $\mathbf{C}$ is a \textbf{fibred} (over its base category $\mathbf{B}$) if (1) every object
$B$ of $\mathbf{C}$ and arrow $f : A \rightarrow \proj{C}(B)$ of $\mathbf{B}$ has a cartesian arrow, and (2) the composition of cartesian arrows is cartesian.
In the following, we suppose that there is a canonical choice of cartesian arrow for any $f$ and $B$, and denote its domain by $\invimg{f}{B}$.  Then we observe that
$\bcarr{f^{-1}(B)}{f}{B}$, and that if
$\bcarr{A}{f}{B}$ and $\bcarr{B}{g}{C}$, then
$\bcarr{A}{g\circ f}{C}$.

\begin{example}
  $\mathbf{TF}$ is fibred over $\mathbf{Vect}_{\mathbb{Q}}$.
   For any transition formula $F \in \TF(X)$ and any
  linear map $\sigma : \mathbb{Q}^Y \rightarrow \mathbb{Q}^X$, we have that $\sigma \in \textbf{TF}(F[\sigma,\sigma'], F)$ is cartesian.

  Similarly, $\mathbf{SF}$ is fibred over $\mathbf{Vect}_{\mathbb{Q}}$.
   For any state formula $G \in \textbf{SF}(X)$ and any
  linear map $\sigma : \mathbb{Q}^Y \rightarrow \mathbb{Q}^X$, we have that $\sigma \in \textbf{SF}(G[\sigma], G)$ is cartesian.
\end{example}

\begin{definition} \label{def:lifting}
 Let $\mathbf{B}$ be a category, let
 $\mathbf{C}$, $\mathbf{D}$, $\mathbf{R}$ be concrete categories over $\mathbf{B}$, and let
 $\gamma : \mathbf{D} \rightarrow \mathbf{C}$ and
 $F : \mathbf{D} \rightarrow \mathbf{R}$ be concrete functors.
 Suppose that $\gamma$ has a weak left adjoint $\tuple{\alpha,\eta}$ and that $\mathbf{R}$ is fibred.  Define the \textit{lifting of $F$} to be $ \overline{F}(A) \defeq \eta_A^{-1}(F(\alpha(A)))\ . $
\end{definition}

\begin{restatable}[Robustness of lifted operations]{theorem}{liftRobustness}
\label{thm:lift-robustness}
  We have that $\overline{F}$ is a concrete functor from $\mathbf{C}$ to $\mathbf{R}$.
\end{restatable}
\begin{proof}
We must show that for any $\barr{A}{f}{B}$ in $\mathbf{C}$, we have $\barr{\overline{F}(A)}{f}{\overline{F}(B)}$.
Since $\barr{A}{f}{B}$ and $\barr{B}{\eta_B}{\gamma(\alpha(B))}$, composing gives $\barr{A}{\eta_B \circ f}{\gamma(\alpha(B))}$.  The universal property of the weak left adjoint yields $\overline{f}$ with $\barr{\alpha(A)}{\overline{f}}{\alpha(B)}$ and $\overline{f} \circ \eta_A = \eta_B \circ f$.

\smallskip\noindent\begin{minipage}{0.3\textwidth}
  \begin{tikzpicture}[thick]
    \matrix (m) [matrix of math nodes, row sep=1.8em, column sep=2em,
      text height=1.5ex, text depth=0.25ex] {
      \proj{C}(A) & \proj{D}(\alpha(A))\\
    \proj{C}(B)  & \proj{D}(\alpha(B))\\ };
    \draw (m-2-2) edge[<-,dashed] node[right]{$\overline{f}$} (m-1-2);
    \draw (m-1-1) edge[->] node[above]{$\eta_A$} (m-1-2);
    \draw (m-2-1) edge[<-] node[left]{$f$} (m-1-1);
    \draw (m-2-1) edge[->] node[below]{$\eta_B$} (m-2-2);
  \end{tikzpicture}
\end{minipage}%
\begin{minipage}{0.68\textwidth}
Since $F$ is a concrete functor, we have $\barr{F(\alpha(A))}{\overline{f}}{F(\alpha(B))}$.  Composing with $\bcarr{\overline{F}(A)}{\eta_A}{F(\alpha(A))}$ and using $\overline{f} \circ \eta_A = \eta_B \circ f$, we get $\barr{\overline{F}(A)}{\eta_B \circ f}{F(\alpha(B))}$.
Meanwhile, $\bcarr{\overline{F}(B)}{\eta_B}{F(\alpha(B))}$ and $\bcarr{f^{-1}(\overline{F}(B))}{f}{\overline{F}(B)}$; composing cartesian arrows gives $\bcarr{f^{-1}(\overline{F}(B))}{\eta_B \circ f}{F(\alpha(B))}$.  We have two arrows to $F(\alpha(B))$ over $\eta_B \circ f$, the latter cartesian; by ($\dagger$), $\overline{F}(A) \preceq f^{-1}(\overline{F}(B))$, and thus $\barr{\overline{F}(A)}{f}{\overline{F}(B)}$.\qedhere
\end{minipage}
\end{proof}

\subsection{Soundness of Lifted Operations} \label{sec:lifted-soundness}

This section shows that lifting preserves soundness: a sound operation on $\mathbf{D}$ lifts to a sound operation on $\mathbf{C}$.

Say that a functor $F$ is \textbf{cartesian} if for every cartesian arrow $f$, $F(f)$ is also cartesian. If we think of a concretization functor $\gamma_R : \mathbf{R}^\sharp \rightarrow \mathbf{R}^\natural$, then $\gamma_R$ being cartesian intuitively means that inverse images in $\mathbf{R}^\sharp$ are sound with respect to $\mathbf{R}^\natural$. In particular, the concretization functor from  $\TF$ to $\cat{TS}$ is cartesian, since $\TF$ is closed under inverse images of linear maps. Similarly, $\gamma_R: \mathbf{SF} \rightarrow \mathbf{SubSet}$ is cartesian.

\begin{theorem}[Soundness of lifted operations]
\label{thm:lift-soundness}
  Suppose that we have the following (not commuting) diagram of concrete categories over $\mathbf{B}$:

 \begin{minipage}{0.2\textwidth}
  \begin{tikzpicture}[thick]
    \matrix (m) [matrix of math nodes, row sep=1.8em, column sep=2em,
      text height=1.5ex, text depth=0.25ex] { \mathbf{D} & \mathbf{R}^\sharp \\
      \mathbf{C}  & \mathbf{R}^\natural \\ };
    \draw (m-2-2) edge[<-] node[right]{$\gamma_R$} (m-1-2);
    \draw (m-1-1) edge[->] node[above]{$F^\sharp$} (m-1-2);
    \draw (m-1-1) edge[->] node[left]{$\gamma_D$} (m-2-1);
    \draw (m-2-1) edge[->] node[below]{$F^\natural$} (m-2-2);
    \draw (m-2-1) edge[->,dashed] node[xshift=-3pt,yshift=6pt]{$\overline{F}^\sharp$} (m-1-2);
  \end{tikzpicture}
  \end{minipage}
  \begin{minipage}{0.75\textwidth}
  where:
  \begin{itemize}
      \item $\mathbf{C}$, $\mathbf{D}$, $\mathbf{R}^\natural$, and $\mathbf{R}^\sharp$ are concrete categories over $\mathbf{B}$;
      \item $\gamma_D: \mathbf{D} \to \mathbf{C}$ is a concrete functor that has a weak left adjoint $\langle \alpha, \eta \rangle$;
      \item $F^\natural : \mathbf{C} \rightarrow \mathbf{R}^\natural$, $F^\sharp : \mathbf{D} \rightarrow \mathbf{R}^\sharp$; and $\gamma_R : \mathbf{R}^\sharp \rightarrow \mathbf{R}^\natural$ are concrete functors;
      \item $F^\sharp$ is sound with respect to $F^\natural \circ \gamma_D$.
  \end{itemize}
  \end{minipage}

   Supposing that $\mathbf{R}^\sharp$ is fibred (so that the lifted operation $\overline{F}^\sharp : \mathbf{C} \rightarrow \mathbf{R}^\sharp$ is well-defined) and $\gamma_R$ is cartesian, then $\overline{F}^\sharp$ is sound with respect to $F^\natural$.
\end{theorem}

\begin{proof}
    For any object $A$ of $\mathbf{C}$, we must show $F^\natural(A) \preceq \gamma_R(\overline{F}^\sharp(A))$.
Since $F^\natural$ is a concrete functor and $\barr{A}{\eta_A}{\gamma_D(\alpha(A))}$, we have $\barr{F^\natural(A)}{\eta_A}{F^\natural(\gamma_D(\alpha(A)))}$.  Soundness of $F^\sharp$ on $\alpha(A)$ gives $F^\natural(\gamma_D(\alpha(A))) \preceq \gamma_R(F^\sharp(\alpha(A)))$.  Composing the above we get $\barr{F^\natural(A)}{\eta_A}{\gamma_R(F^\sharp(\alpha(A)))}. $
By definition, $\overline{F}^\sharp(A) = \eta_A^{-1}(F^\sharp(\alpha(A)))$, giving a cartesian arrow $\bcarr{\overline{F}^\sharp(A)}{\eta_A}{F^\sharp(\alpha(A))}$.  Since $\gamma_R$ is a cartesian functor, $\bcarr{\gamma_R(\overline{F}^\sharp(A))}{\eta_A}{\gamma_R(F^\sharp(\alpha(A)))}$.  We now have two arrows to $\gamma_R(F^\sharp(\alpha(A)))$ over $\eta_A$, the latter cartesian; by ($\dagger$), $F^\natural(A) \preceq \gamma_R(\overline{F}^\sharp(A))$.
\end{proof}

\subsection{Combinations and Algebraic Properties of Lifted Operations} \label{sec:algebraic-laws-main}

This section shows that (in a certain sense) if an operation $F : \cat{D} \rightarrow \cat{R}$ satisfies a certain kind of (in)equational law, then that law is also satisfied by its lifting $\overline{F} : \cat{C} \rightarrow \cat{R}$ (obtained by Definition~\ref{def:lifting}).  Such laws amount to another kind of robustness guarantee (further explored in Section~\ref{sec:apa}).

As an example, suppose that $\cat{D}$ is a subcategory of $\cat{TF}$ equipped with a reflexive transitive closure operator $(-)^\star : \cat{D} \rightarrow \cat{TF}$, and let
$(-)^{\overline{\star}} : \cat{TF} \rightarrow \cat{TF}$ denote its lifting.  Since for any transition formula $F$ in $\cat{D}$, $F^\star$ is transitively closed, we have the law $\forall F \in \Ob{\cat{D}}. F^\star \seq F^\star \leq F^\star$, or equivalently
$\forall F \in \Ob{\cat{D}}. (F^\star)^2 \leq F^\star$
(where $(-)^2$ is a functor mapping each transition formula $G$ to $G \seq G$).  As a consequence of Theorem~\ref{thm:generalized_extension} below, we can raise this inequation to one for the lifted operator $(-)^{\overline{\star}}$, yielding the law
$\forall F \in \Ob{\cat{TF}}. (F^{\overline{\star}})^2 \leq F^{\overline{\star}}$.

In the technical statement of the lifting theorem (below), $G : \cat{D} \rightarrow \cat{R}$ is an operation that is to be lifted to $\overline{G} : \cat{C} \rightarrow \cat{R}$.  The left-hand-side and right-hand-sides of an inequational law is represented as the composition of functors $F_1GH$ and $F_2G$, respectively.  To instantiate this theorem to lift the transitivity property above, we
set $F_1$ to $(-)^2$, $G$ to $(-)^\star$, and $H$ and $F_2$ to identity functors.  In Appendix~\ref{sec:algebraic-laws}, we show how several more laws can be obtained by instantiating these functors appropriately, including:
\begin{itemize}
\item \textit{Transitivity:} $F^{\overline{\star}} \seq F^{\overline{\star}} = F^{\overline{\star}}$ (Corollary~\ref{cor:law-transitivity})
\item \textit{Unrolling:} if $(F^n)^{\overline{\star}} \leq F^{\overline{\star}}$ (Corollary~\ref{cor:law-unrolling})
\end{itemize}

\begin{restatable}[Lifting]{theorem}{generalizedExtension} \label{thm:generalized_extension}
Let $\mathbf{B}$ be a category, and let $\mathbf{C}$, $\mathbf{D}$, $\mathbf{R}$, and $\mathbf{R}'$ be concrete categories over $\mathbf{B}$. Let $\gamma: \mathbf{D} \to \mathbf{C}$ be a concrete functor that has a weak left adjoint $\tuple{\alpha, \eta}$.
Let $\functor{H}: \mathbf{D} \to \mathbf{D}$ and $\widetilde{\functor{H}}: \mathbf{C} \to \mathbf{C}$ be concrete functors such that $\widetilde{\functor{H}}\gamma = \gamma\functor{H}$.
Let $\functor{G}: \mathbf{D} \to \mathbf{R}$, $\functor{F}_1: \mathbf{R} \to \mathbf{R}'$, and $\functor{F}_2: \mathbf{R} \to \mathbf{R}'$ be concrete functors.
Suppose that $\mathbf{R}$ is fibred, and that $\functor{F}_2$ is a cartesian functor. Let $\lifted{\functor{G}}$ be the lifting of $\functor{G}$.
If $\functor{F}_1(\functor{G}(\functor{H}(d))) \preceq \functor{F}_2(\functor{G}(d))$ for all objects $d$ of $\cat{D}$, then $\functor{F}_1(\lifted{\functor{G}}(\widetilde{\functor{H}}(c))) \preceq \functor{F}_2(\lifted{\functor{G}}(c))$ for all objects $c$ of $\cat{C}$.
\end{restatable}

Finally, we observe that lifted operations can be combined.
If two subcategories $\mathbf{D}$ and $\mathbf{E}$ each admit best abstractions for $\cat{C}$, then so does the product $\mathbf{D} \times \mathbf{E}$ (Appendix~\ref{sec:domain-combinations}), yielding a more precise analysis that inherits all robustness and soundness properties. For example, products of domains like $\mathbb{Q}\text{-VASR} \times \mathbf{LT}_\bot$ (both from Table~\ref{tbl:instances}) can yield new robust analyses for computing reflexive transitive closure.

%%% Local Variables:
%%% TeX-master: "main"
%%% End:

\section{Robust Algebraic Program Analysis} \label{sec:apa}

The previous section demonstrated a recipe for designing robust program analysis operations,
which generalizes a pattern seen in several algebraic program analyses (Table~\ref{tbl:instances}).  In this section, we investigate the robustness properties of such analyses.  That is, assuming that each individual operation of an algebraic program analysis is robust, in what sense is the overall analysis robust?

An algebraic program analysis operates by (1) parsing the program to produce a flow graph $G$ (over some algebra $A$) that represents it, and then (2) using a \textit{summarization} algorithm (such as \citet{JACM:Tarjan1981}'s) to produce a function $s : V_G \rightarrow A$ that maps each vertex $v$ to a summary $s(v)$ that over-approximates the semantics of all program paths from the entry $r_G$ to $v$.  To formalize robustness of an algebraic analysis, we restrict our attention to step (2)---that is, we consider transformations in flow graphs, and the sense in which (some) summarization algorithms are robust (assuming that each algebraic operation is robust).

To motivate our discussion, consider the flow graphs in Figure~\ref{fig:overview-flowgraphs}.
$G_1$ simulates $G_2$ in the sense that there is a substitution
$f = [i \mapsto x]$ mapping the variables of $G_2$ to linear terms over the variables of $G_1$, and a function $h : V_{G_1} \rightarrow V_{G_2}$ from the vertices of $G_1$ to the vertices of $G_2$ such that for each edge $\tuple{u,v}$ in $G_2$, there is a corresponding path in $G_1$ from $h(u)$ to $h(v)$ that simulates it (modulo the transformation $f$).\footnote{For $\tuple{1,2}$, $\tuple{3,4}$, and $\tuple{4,5}$, the simulating path is simply  the corresponding edge in $G_1$; edges $\tuple{2,3}$ and $\tuple{5,3}$ do not change the value of $x$, and the simulating path is the empty path from B to B in $G_1$.}  When an algebraic program analysis computes a summary assignment $s_1 : V_{G_1} \rightarrow \mathbf{TF}(\set{i})$ for $G_1$ and $s_2 : V_{G_2} \rightarrow \mathbf{TF}(\set{x,y})$, they are similarly related: for each vertex $v$, $s_1(f(v))$ simulates $s_2(v)$ (modulo the transformation $f$).   This is the case for all of the analyses derived using the pattern in Section~\ref{sec:lifting}.

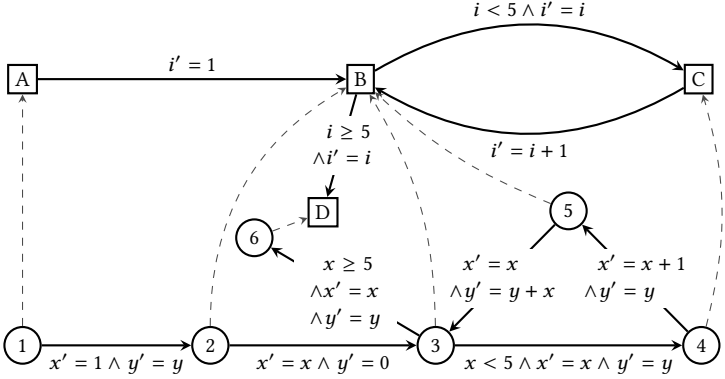
\begin{figure}
\footnotesize
\begin{tikzpicture}[>=stealth,thick]
  \node [rectangle,draw](A) {A};
  \node [rectangle,draw,right of=A, node distance=4.5cm] (B) {B};
  \node [rectangle,draw,right of=B, node distance=4.5cm] (C) {C};
  \node [rectangle,draw,below of=B, node distance=1.75cm,xshift=-0.5cm] (D) {D};
  \draw (A) edge[->] node[above]{$i' = 1$} (B);
  \draw (B) edge[->,bend left] node[above]{$i < 5 \land i' = i$} (C);
  \draw (C) edge[->,bend left] node[below]{$i' = i + 1$} (B);
  \draw (B) edge[->] node[fill=white,inner sep=1pt]{$\begin{array}{l@{}l}&i \geq 5\\ \land& i' = i\end{array}$} (D);

  \node [circle,draw,below of=A, node distance=3.5cm] (1) {1};
  \node [circle,draw,right of=1, node distance=2.5cm] (2) {2};
  \node [circle,draw,right of=2, node distance=3cm] (3) {3};
  \node [circle,draw,above right of=3, node distance=2.5cm] (5) {5};
  \node [circle,draw,below right of=5, node distance=2.5cm] (4) {4};

  \node [circle,draw,above left of=3, node distance=2cm,xshift=-1cm] (6) {6};
  \draw (1) edge[->] node[below]{$x' = 1 \land y' = y$} (2);
  \draw (2) edge[->] node[below]{$x' = x \land y' = 0$} (3);
\draw (3) edge[->] node[below]{$x < 5 \land x' = x \land y' = y$} (4);
\draw (4) edge[->] node[fill=white,inner sep=1pt]{$\begin{array}{l@{}l}&x' = x + 1\\\land& y' = y\end{array}$} (5);
\draw (5) edge[->] node[fill=white,inner sep=1pt]{$\begin{array}{l@{}l}&x' = x\\\land& y' = y + x\end{array}$} (3);

\draw (3) edge[->] node[fill=white,inner sep=1pt]{$\begin{array}{l@{}l}&x \geq 5\\ \land & x' = x\\ \land& y' = y\end{array}$} (6);

\draw (1) edge[->,dashed,gray,thin] (A);
\draw (2) edge[->,dashed,gray, bend left,thin] (B);
\draw (3) edge[->,dashed,gray,thin,bend right=15] (B);
\draw (4) edge[->,dashed,gray,thin, bend right=15] (C);
\draw (5) edge[->,dashed,gray,bend left=10,thin] (B);
\draw (6) edge[->,dashed,gray,thin] (D);

\end{tikzpicture}
\caption{Two flow graphs, corresponding to the programs $G_1$ from Figure~\ref{fig:overview-ex1} (above) and $G_2$ from Figure~\ref{fig:overview-ex2} (below).  The dashed lines indicate a stuttering simulation $\tuple{h, f}$ (where $h$ maps $1 \mapsto \text{A}$, $\{2,3,5\} \mapsto \text{B}$, $4 \mapsto \text{C}$, $6 \mapsto \text{D}$, and $f = [i \mapsto x]$) that preserves loops from the flow graph below to the one above.  The summaries at the loop headers B and 3 are the loop summaries of Table~\ref{tab:summaries}.  \label{fig:overview-flowgraphs}}
\end{figure}

The input to an algebraic program analysis is a flow graph $G$ with weights drawn from some algebra $A$, and the output is a summary assignment $s : V_G \rightarrow A$.  Following the framework introduced in Section~\ref{sec:robustness}, we wish to define concrete categories for flow graphs and summary assignments (over some common base category) such that summarization is robust (under the assumption that each algebraic operation is robust).  The remainder of this section formalizes robust algebraic program analysis and summary assignments (Section~\ref{sec:summary-assignemnt}), vertex elimination (Section~\ref{sec:robust-vertex-elimination}), and the class of transformations under which it is robust (Section~\ref{sec:continuous-maps}).

\subsection{Robust Algebraic Program Analysis}
\label{sec:summary-assignemnt}

This section defines the structure of a robust algebraic program analysis.  \textsc{Informally, } an algebraic program analysis is robust (with respect to some base category $\cat{B}$) if all of its operators are robust and satisfy the Pre-Kleene algebra laws.  To formalize this, we must generalize Pre-Kleene algebras to a setting wherein the the universe is a concrete category $\cat{C}$ over $\cat{B}$ rather than a set, and operators are concrete functors rather than functions.  An $n$-ary operator on $\cat{C}$ corresponds to functor from the $n$-fold product $\mathbf{C} \times_{\mathbf{B}} \cdots \times_{\mathbf{B}} \mathbf{C}$ to $\mathbf{C}$, where the product operator is defined as follows (and noting that ``$0$ times'' $\mathbf{C}$ is $\mathbf{B}$).

\begin{definition} \label{def:pullbacks}
    For any concrete categories $\mathbf{M}$ and $\mathbf{N}$ over a category $\mathbf{B}$, define $\mathbf{M} \times_{\mathbf{B}} \mathbf{N}$ to be the pullback of $\proj{M}$ and $\proj{N}$; that is, $\mathbf{M} \times_{\mathbf{B}} \mathbf{N}$ is the category in which the objects are pairs $\tuple{M, N}$ consisting of an object $M$ of $\mathbf{M}$ and an object $N$ of $\mathbf{N}$ such that $\proj{M}(M) = \proj{N}(N)$, and an arrow $(A,B) \rightarrow (A',B')$ is a pair $\tuple{f,g}$ such that $f \in \mathbf{M}(A,A')$ and $g \in \mathbf{N}(B,B')$ and $\proj{M}(f) = \proj{N}(g)$.  $\mathbf{M} \times_{\mathbf{B}} \mathbf{N}$ naturally forms a concrete category over $\cat{B}$.
\end{definition}

For instance, consider the category of transition formulas $\TF$.  The sequential composition operator is a \textit{partial} function on $\TF$, which is defined for any pair of transition formulas $F$ and $G$ that operate on the same set of variables (formalized in three equivalently ways: (1) $\phi_{\cat{TF}}(F) = \phi_{\cat{TF}}(G)$; (2) $(F,G)$ is an object of $\TF \times_{\cat{Vect}_{\mathbb{Q}}} \TF$; (3) $F$ and $G$ belong to the same fiber of $\phi_{\cat{TF}}$).   Sequential composition is robust in the sense that 
 for any transition formulas $F, F' \in \TF(X)$, and $G, G' \in \TF(Y)$ and any linear transformation $f : \mathbb{Q}^X \rightarrow \mathbb{Q}^Y$ such that $f$ is a simulation from $F$ to $G$ and $F'$ to $G'$, we have that $f$ is a simulation from $F \seq F'$ to $G \seq G'$---this corresponds precisely to the fact that sequential composition is a concrete functor $\TF \times_{\cat{Vect}_{\mathbb{Q}}} \TF \rightarrow \TF$.  
The $+$, $0$, and $1$ operators can similarly be understood as concrete functors, as any ``lifted'' iteration operator following the recipe of Section~\ref{sec:lifting}).
Generalizing this structure, we have the following definition.
\begin{definition} \label{def:robust_apa}
Let $\mathbf{B}$ be a category.
A \textbf{robust algebraic analysis} (over $\mathbf{B}$) is a concrete category $\mathbf{A}$ over $\mathbf{B}$ equipped with
concrete functors $0 : \mathbf{B} \rightarrow \mathbf{A}$, $1 : \mathbf{B} \rightarrow \mathbf{A}$,
$(-) \cdot (-) : \mathbf{A} \times_{\mathbf{B}} \mathbf{A} \rightarrow \mathbf{A}$,
$(-) + (-) : \mathbf{A} \times_{\mathbf{B}} \mathbf{A} \rightarrow \mathbf{A}$, and
$(-)^* : \mathbf{A} \rightarrow \mathbf{A}$ such that
\begin{enumerate}
\item Each fiber of $\mathbf{A}$ forms a Pre-Kleene algebra;
\item For any objects $a$ and $b$, we have $a \preceq b$ iff $a \leq b$ (i.e., the approximation order derived from the concrete category structure coincides with the natural order from the PKA structure).
\end{enumerate}
\end{definition}

Recall from Section~\ref{sec:algebraic-laws-main} that ``lifted'' operators (derived from a base operator using the recipe of Section~\ref{sec:lifting}) satisfy many of the same algebraic laws as their base.   As a consequence, $\TF$ forms a robust algebraic analysis (over ${\cat{Vect}_{\mathbb{Q}}}$) where the $(-)^*$ operator is instantiated to the polyhedral iteration operator from Section~\ref{sec:overview}, or the operator of \cite{POPL:CK2024}.

For any arrow $f$ of $\mathbf{B}$ and any $A,B$ in $\Ob{A}$ such that $\barr{A}{f}{B}$, we say that
\textbf{$B$ simulates $A$ modulo $f$}.
Observe that
the $\barr{}{f}{}$ relation is compatible with all of the algebraic operations and the natural order; that is,

\begin{observation}[Compatibility]
For all $A_1, A_2, B_1, B_2$ such that $\barr{A_1}{f}{B_1}$ and $\barr{A_2}{f}{B_2}$:
\begin{center}
\begin{minipage}{0.3\textwidth}
\begin{itemize}
\item $\barr{A_1 + A_2}{f}{B_1 + B_2}$
\item $\barr{A_1 \cdot A_2}{f}{B_1 \cdot B_2}$
\item $\barr{A_1^*}{f}{B_1^*}$
\end{itemize}
\end{minipage}
\begin{minipage}{0.6\textwidth}
\begin{itemize}
\item $\barr{0}{f}{0}$
\item $\barr{1}{f}{1}$
\item If $A_1 \leq A_2$, $\barr{A_2}{f}{B_1}$, and $B_1 \leq B_2$, then $\barr{A_1}{f}{B_2}$
\end{itemize}
\end{minipage}
\end{center}
\end{observation}

Finally we may define the category $\mathbf{S}_{\mathbf{A}}$ of summary assignments for a robust algebraic analysis $\mathbf{A}$.  The objects of $\mathbf{S}_{\mathbf{A}}$ are summary assignments---i.e., functions $s: U \rightarrow \proj{A}^{-1}(X)$ for some set $U$ and some $X \in \Ob{B}$.  An arrow in $\mathbf{S}_{\mathbf{A}}$ from a summary assignment $s : U \rightarrow \proj{A}^{-1}(X)$ to $t : V \rightarrow \proj{A}^{-1}(Y)$ consists of a pair $\tuple{h,f}$, comprising a ``structure transformation'' $h : U \rightarrow V$ and ``weight transformation'' $f \in \mathbf{B}(X,Y)$ such that for all $u \in U$, we have that $\barr{s(u)}{f}{t(h(u))}$ (that is, for each vertex $u \in U$, the summary assigned to $u$ by $s$ is simulated by the summary assigned to its corresponding vertex $h(u)$ by $t$, modulo the transformation $f$).  Thus, the natural ``base category'' over which to understand summary assignments is the product category $\mathbf{Set} \times \mathbf{B}$, where for a summary assignment $s : U \rightarrow \proj{A}^{-1}(X)$, we have $\proj{S_A}(s) = \tuple{U, X}$.

\subsection{Robustness of Summarization Under Isomorphisms} \label{sec:robust-vertex-elimination}

First we will present an abstract view of summarization algorithms based on \textit{vertex elimination}, which can be thought of as a variant of Gauss-Jordan elimination, or Kleene's algorithm for converting a finite automaton to a regular expression.  Practical summarization algorithms, such as \citet{JACM:Tarjan1981b}'s, can be thought of as vertex elimination algorithms that use various techniques to improve computational efficiency, but produce the same results as vertex elimination (assuming that the underlying algebra obeys the semiring laws).

In defining the vertex elimination operation, it is convenient to extend the weight function $w_G$ to the domain $V_G \times V_G$, by assigning each pair $u$ and $v$ with $\tuple{u,v} \notin E_G$ the weight $0$.  We do not distinguish between a weight function and its total extension.  Eliminating a vertex $v$ from a graph $G$ produces another flow graph $G/v$ in which $v$ has no outgoing edges, but the paths in $G$ are preserved in the sense that for any path $\pi$ in $G$ from $r_G$ to any vertex $u$, there is a path $\pi'$ in $G/v$ such that $w_G(\pi) \leq w_{G/v}(\pi')$. The graph $G/v$ has the same vertices and root as $G$, and its edges are defined as follows:
\begin{align*}
    E_{G/v} & \defeq \set{ \tuple{u,u'} \in E : u' \neq v} \cup \set{ \tuple{p,s} : \tuple{p,v} \in E \land \tuple{v,s} \in E}\\
    w_{G/v}(u,u') &\defeq \begin{cases}
    w_G(u,u') + w_G(u,v)w_G(v,v)^*w_G(v,u') & u \neq v, u' \neq v\\
    w_G(u,v)w_G(v,v)^* & u' = v\\
    0 & \text{otherwise}
    \end{cases}
\end{align*}
We extend the elimination notation to sequences by defining $G/\epsilon = G$ and $G/v_1,v_2,\dots,v_n = (G/v_1)/v_2,\dots,v_n$.
To compute a summary assignment using vertex elimination, we enumerate the non-root vertices of $G$ in some order $v_1,\dots,v_n$
and then eliminate them to yield a graph $H \defeq G/v_1,\dots,v_n$.  Since eliminated vertices have no outgoing edges, the only edges in $H$ are of the form $\tuple{r_G,v}$, and the weight of this edge must approximate the weight of any path from $r_G$ to $v$ in $G$.  Thus, the function
mapping $v \mapsto w_H(r_G,v)$ is a sound summary assignment.

If $A$ happens to be a Kleene algebra, then the elimination order is irrelevant.  But if $A$ is a \textit{Pre-Kleene algebra}, different elimination orders can produce different results (i.e., different sound over-approximations of program behavior).  From this abstract perspective, the crucial difference between different summarization algorithms is the selection of elimination orders.

The most basic question concerning the robustness of a summarization algorithm is whether it is robust with respect to graph isomorphisms---that is, given flow graphs $G$ and $H$ that differ only in the names of vertices, does the algorithm produce exactly the same result for $G$ and $H$?  The answer is, unfortunately, no. For instance, consider the flow graph in Figure~\ref{fig:tarjan-flowgraph}.  The map that swaps the vertices 2 and 3 is a structural isomorphism for this graph, so we cannot depend on the order that an algorithm will eliminate them.  However, the order of elimination matters---in particular, for the order $2,3$ the $(-)^*$ operator is applied to $dc$, and for the $3,2$ order it is applied to $cd$, and in general there is no relationship between $(dc)^*$ and $(cd)^*$.

\begin{figure}
\begin{subfigure}{0.25\textwidth}
\begin{tikzpicture}[>=stealth]
 \node [circle,draw] (A) {1};
 \node [below left of=A,circle,draw,node distance=1.3cm] (B) {2};
  \node [below right of=A,circle,draw, node distance=1.3cm] (C) {3};
  \draw (A) edge[->] node[above left]{$a$} (B);
\draw (A) edge[->] node[above right]{$b$} (C);
\draw (B) edge[->,bend left=15] node[above]{$c$} (C);
\draw (C) edge[->,bend left=15] node[below]{$d$} (B);
\end{tikzpicture}
\caption{A flow graph $G$ \label{fig:tarjan-flowgraph}}
\end{subfigure}
\begin{subfigure}{0.35\textwidth}
\begin{tikzpicture}[>=stealth]
 \node [circle,draw] (A) {1};
 \node [below left of=A,circle,draw,node distance=1.3cm] (B) {2};
  \node [below right of=A,circle,draw, node distance=1.3cm] (C) {3};
  \draw (A) edge[->] node[above left]{$a + b(dc)^*d$} (B);
\draw (A) edge[->] node[above right]{$(b+ac)(dc)^*$} (C);
\end{tikzpicture}
\caption{$G/2/3$ \label{fig:tarjan-reduce1}}
\end{subfigure}
\begin{subfigure}{0.35\textwidth}
\begin{tikzpicture}[>=stealth]
 \node [circle,draw] (A) {1};
 \node [below left of=A,circle,draw,node distance=1.3cm] (B) {2};
  \node [below right of=A,circle,draw, node distance=1.3cm] (C) {3};
  \draw (A) edge[->] node[above left]{$(a+bd)(cd)^*$} (B);
\draw (A) edge[->] node[above right]{$b+a(cd)^*c$} (C);
\end{tikzpicture}
\caption{$G/3/2$ \label{fig:tarjan-reduce2}}
\end{subfigure}
\caption{An irreducible flow graph $G$ along with the result of eliminating vertices $2$ and $3$ in both possible orders.}
\end{figure}
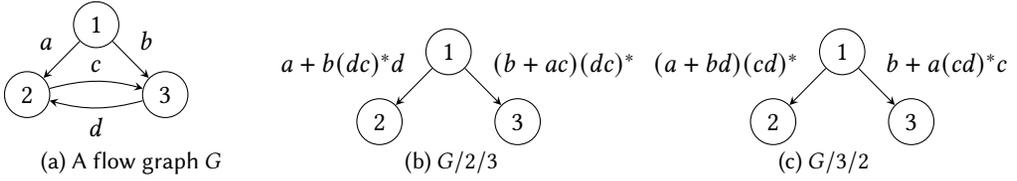

For reducible flow graphs, however, Tarjan's algorithm \textit{is} robust with respect to isomorphisms.  To prove this result, we define a class of \textit{admissible} elimination orders, and show that any two admissible elimination orders of the same reducible graph produce the same analysis result.  It follows that any vertex elimination algorithm that eliminates vertices in an admissible order (of which Tarjan's algorithm is an example) is robust with respect to isomorphisms (for reducible flow graphs).

Let $G$ be a flow graph.  Say that $u$ \textbf{dominates} $v$ if every path from $r_G$ to $v$ passes through $u$.  Dominance forms a partial order, with least element $r_G$.  A flow graph $G$ is said to be \textbf{reducible} if for every cycle in $G$, there is some vertex along the cycle that dominates all other vertices.
A \textbf{local cycle} of a vertex $v$ is a path $\pi$ in $G$ such that (1) $v = \src(\pi) = \dst(\pi)$, and (2) $v$ dominates every vertex along $\pi$.  Define $L_G(v)$ to be the set of vertices that appear in some local cycle of $v$.
Say that a sequence $v_1,\dots,v_n$ of vertices is \textbf{admissible} if there exists no $i < j$ such that $v_j \in L(v_i)$.
Intuitively, one could think that vertices in an inner loop should be eliminated before those in the outer loop.

\begin{restatable}{theorem}{admissibility}
\label{thm:admissibility}
  Let $G$ be a reducible flow graph, and let $v_1,\dots,v_n$ and $u_1,\dots,u_n$ be enumerations of some subset $U \subseteq V_G$.  If both $v_1,\dots,v_n$ and $u_1,\dots,u_n$ are admissible, then $G/v_1,\dots,v_n = G/u_1,\dots,u_n$.
\end{restatable}

In the following, we use $\sem{G}$ to denote the summary assignment produced by a vertex elimination algorithm that eliminates vertices in a reducible order.  Formally,
let $G$ be a reducible flow graph over a PKA $A$.  Define $\sem{G}$ to be the function $V_G \rightarrow A$ that maps each vertex $v \in V_G$ to $w_H(r_G,v)$, where
$H \defeq G/v_1,\dots,v_n$ for some admissible order $v_1,\dots,v_n$ of the non-root vertices of $G$.  $\sem{G}$ is well-defined by Theorem~\ref{thm:admissibility} above.

\subsection{Robustness of Summarization Under Loop-Preserving Stuttering Simulations}
\label{sec:continuous-maps}

Let $\mathbf{A}$ be a robust algebraic program analysis over some base category $\mathbf{B}$.
Recall that the natural base category over which we understand the results of vertex elimination is $\mathbf{Set} \times \mathbf{B}$.
In this section, we wish to define a category $\mathbf{G}_{\mathbf{A}}$ of \textit{reducible} flow graphs over $\mathbf{A}$ such that
$\sem{-}$ can be understood to be
a concrete functor from $\mathbf{G}_\mathbf{A}$ to $\mathbf{S}_{\mathbf{A}}$.  Given that
we wish to view $\mathbf{G}_{\mathbf{A}}$ as a concrete category over $(\mathbf{Set} \times \mathbf{B})$, we can make this question more concrete.  Suppose $G$ and $H$ are flow graphs with weights in $\proj{A}^{-1}(X)$ and $\proj{A}^{-1}(Y)$ respectively (for some $X,Y \in \Ob{B}$).  For which functions $h : V_G \rightarrow V_H$ and arrows $f : X \rightarrow Y$ should the pair $\tuple{h,f}$ be considered an acceptable transformation from $G$ to $H$ (such that we may conclude that $\tuple{h,f}$ is an arrow
from $\sem{G}$ to $\sem{H}$)?  The answer provided in this section is that $\tuple{h,f}$ should be a \textit{stuttering simulation} that \textit{preserves loops}.

Figure~\ref{fig:overview-flowgraphs} depicts an example of a stuttering simulation between two flow graphs.  The intuition behind simulation is that the simulator has \textit{at least} as many behaviors in the sense that every execution of the simulatee can be mapped to a corresponding execution of the simulator.  In the abstract, we view ``executions'' as paths in a flow graph, and correspondence as a consistent approximation relation on weights of those paths.
\begin{definition}
  Let $\mathbf{A}$ be a robust algebraic program analysis over a base category $\mathbf{B}$.
  Let $G$ and $H$ be flow graphs over $\proj{A}^{-1}(X)$ and $\proj{A}^{-1}(Y)$,
  respectively.  A
  \textbf{stuttering simulation} is a pair
  $\tuple{h,f}$ where
  $h : V_G \rightarrow V_H$ and
  $f \in \mathbf{B}(X,Y)$ such that for each edge $\tuple{u,v} \in E_G$, we have
  \begin{itemize}
  \item $\tuple{h(u),h(v)} \in E_H$ and $\barr{w_G(u,v)}{f}{w_H(h(u),h(v))}$, \textit{or}
  \item $h(u) = h(v)$ and
  $\barr{w_G(u,v)}{f}{1}$
  \end{itemize}
\end{definition}

Observe that for any stuttering simulation $\tuple{h,f}$ between flow graphs $G$ and $H$
there exists a function $r$ that maps any path $\pi$ in $G$ to a simulating path $r(\pi)$ in $H$, which can be constructed as:
\begin{align*}
r(\epsilon) &= \epsilon\\
r(\pi \tuple{u,v}) &= \begin{cases}
  r(\pi)\tuple{h(u),h(v)} & \text{if } \tuple{h(u),h(v)} \in E_H \text{ and } \barr{w_G(u,v)}{f}{w_H(h(u),h(v))}\\
  r(\pi) & \text{if } h(u)=h(v) \text{ and } \barr{w_G(u,v)}{f}{1}
\end{cases}
\end{align*}
Note that the above does not define $r$ uniquely, because an edge $\tuple{u,v}$ may be simulated both by $1$ and by $w_H(h(u),h(v))$.  We say that any function meeting the conditions above is a \textbf{witness} for the stuttering simulation $\tuple{h,f}$.

\begin{example} \label{ex:phase-transform}
  \citet{CAV:SDDA2011} introduce a technique for improving the precision of static analyses by splitting loops into phases---an example of the transformation appears below (the program on the left is transformed into the program on the far right).  This kind of transformation \textit{heuristically} improves precision, but in some cases can actually decrease precision.  For robust algebraic program analyses, the transformation always improves precision (not necessarily strictly): there is a stuttering simulation from the transformed program to the original.  A schematic illustration appears below to the right: $\tuple{h,1}$ is a stuttering simulation from the rightmost graph to the one to its left,  where $h = \set{ u \mapsto x, v \mapsto x }$.
  \begin{center}
  \newsavebox{\lstboxSplitA}
  \begin{lrbox}{\lstboxSplitA}\begin{lstlisting}[style=base]
while(x < 100) {
  x++;
  if (x > 50) y++;
}
\end{lstlisting}\end{lrbox}
  \newsavebox{\lstboxSplitB}
  \begin{lrbox}{\lstboxSplitB}\begin{lstlisting}[style=base]
while(x <= 49) { x++; }
while(x < 100 && x > 49) {
  x++; y++;
}
\end{lstlisting}\end{lrbox}
  \begin{tikzpicture}[node distance=1.5cm]
      \node (pre) {\usebox{\lstboxSplitA}};
\node [right of=pre,node distance=5cm] (post) {\usebox{\lstboxSplitB}};
\draw (pre) edge[->] node[above]{\textit{Split}} (post);

    \node [circle,draw, right of =post,node distance=3cm,yshift=0.5cm] (x) {$x$};
    \draw (x) edge[->,loop below] node[below]{$a+b$} (x);

    \node [circle,draw,right of=x] (u) {$u$};
    \node [circle,draw,right of=u] (v) {$v$};
    \draw [->] (u) -- node[above]{$1$} (v);
    \draw (u) edge[->,loop below] node[below]{$a$} (u);
    \draw (v) edge[->,loop below] node[below]{$b$} (v);

  \end{tikzpicture} \qedhere
  \end{center}
\end{example}

\begin{example} \label{ex:unrolling}
Loop unrolling is a transformation widely used in compilers to improve performance by reducing the overhead of loop-control branches and exposing more instructions to the scheduler
\citep{Aho:2006:CPT,Hennessy:2011:CAQ}.
For robust algebraic program analyses, this transformation is (not necessarily strictly) precision-improving  because a stuttering simulation exists from the unrolled program (right) to the original (left): $\tuple{h,1}$ is a stuttering simulation from the transformed graph to the original, where $h = \set{ u \mapsto x, v \mapsto x, w \mapsto x }$.

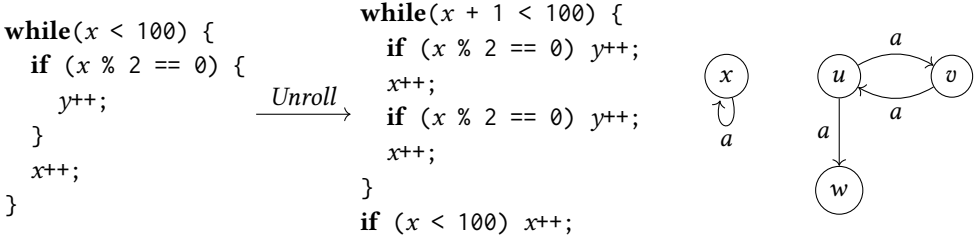
\begin{center}
  \newsavebox{\lstboxUnrollA}
  \begin{lrbox}{\lstboxUnrollA}\begin{lstlisting}[style=base]
while(x < 100) {
  if (x % 2 == 0) {
    y++;
  }
  x++;
}
\end{lstlisting}\end{lrbox}
  \newsavebox{\lstboxUnrollB}
  \begin{lrbox}{\lstboxUnrollB}\begin{lstlisting}[style=base]
while(x + 1 < 100) {
  if (x % 2 == 0) y++;
  x++;
  if (x % 2 == 0) y++;
  x++;
}
if (x < 100) x++;
\end{lstlisting}\end{lrbox}
  \begin{tikzpicture}[node distance=1.5cm]
      \node (pre) {\usebox{\lstboxUnrollA}};
\node [right of=pre,node distance=5cm] (post) {\usebox{\lstboxUnrollB}};
\draw (pre) edge[->] node[above]{\textit{Unroll}} (post);

\node [circle,draw, right of =post,node distance=3cm,yshift=0.5cm] (x) {$x$};
\draw (x) edge[->,loop below] node[below]{$a$} (x);

\node [circle,draw,right of=x] (u) {$u$};
\node [circle,draw,right of=u] (v) {$v$};
\node [circle,draw,below of=u] (w) {$w$};

\draw [->, bend left=30] (u) to node[above]{$a$} (v);
\draw [->, bend left=30] (v) to node[below]{$a$} (u);
\draw [->] (u) -- node[left]{$a$} (w);

  \end{tikzpicture} \qedhere
  \end{center}
\end{example}

\textsc{Informally, } vertex elimination is robust with respect to stuttering simulations that are \textit{also} respect syntactic loop structure, such as nesting hierarchy.  Such simulations (which include Examples~\ref{ex:phase-transform} and \ref{ex:unrolling})
are called \textit{loop-preserving}.  We will now make loop preservation precise.
For a path $\pi = e_1e_2 \dots e_n$ and a vertex $u \in V_G$, define the \textbf{$u$-length} of $\pi$ to be the number of edges $e_i$ such that the source of $e_i$ is $u$.  Denote the $u$-length of a path $\pi$ by $|\pi|_u$.

\begin{definition}
  Let $G$ and $H$ be flow graphs, and let $\tuple{h,f}$ be a stuttering simulation from $G$ to $H$.  Say that $\tuple{h,f}$ is \textbf{loop-preserving} if
  \begin{itemize}
  \item \textit{(Loop membership)} For any $u,v \in V_G$, if $u \in L_G(v)$, then $h(u) \in L_H(h(v))$;
  \item \textit{(Non-nesting)} For any $u,v \in V_G$ such that $h(u) = h(v)$, and $u \in L_G(v)$, we must have $L_G(u) = \emptyset$;
  \item \textit{(Consistent unrolling)} There is a witness $r$ for $\tuple{h,f}$ such that for any $v \in V_G$, there is some $n \geq 1$ such that every primitive local cycle (meaning that it is not obtained by repeating some other local cycles multiple times) $\pi$ of $v$ satisfies $|r(\pi)|_{h(v)} = n$.
  \end{itemize}
\end{definition}

We now present the main result: for any robust algebraic program analysis $\mathbf{A}$, vertex elimination $\sem{-} : \mathbf{G}_{\mathbf{A}} \rightarrow \mathbf{S}_{\mathbf{A}}$ is a robust operation.
\begin{restatable}{theorem}{rpa}
\label{thm:rpa}
  Let $\mathbf{A}$ be a robust algebraic program analysis over a category $\mathbf{B}$.
  Let $G$ and $H$ be reducible flow graphs, and let $\tuple{h,f}$ be a stuttering simulation from $G$ to $H$ that is loop-preserving.  Then for any vertex $v \in V_G$, we have that
  $\sem{G}(v)$ is simulated by $\sem{H}(h(v))$ modulo $f$.
\end{restatable}

%%% Local Variables:
%%% TeX-master: "main"
%%% End:

\section{Conclusion} \label{sec:conclusion}

It is the opinion of the authors that \textit{robustness} should be an important design consideration for practical program analyses.  We provide a unified framework for articulating a class of robustness properties in the language of concrete category theory, by thinking of robustness as preserving certain relationships between programs.

We give a general design strategy for achieving robustness, and establish robustness of several existing analyses by way of showing that they are instances of this strategy.
A natural avenue for future work is to instantiate this strategy in new ways---i.e., for what other sub-Turing models of computation, and what other classes of simulations, can this recipe be followed?
We also study robustness of algebraic program analyses ``in the large''---that is, how guarantees of robustness of individual program operations translate to end-to-end analysis guarantees upon which users may rely.  A limitation of this theory is that we have restricted our attention to \textit{reducible} flow graphs.
  We leave open the question of whether robustness can be established for irreducible flow graphs, and the correct notion of simulation for algebraic termination analysis.

%%% Local Variables:
%%% TeX-master: "main"
%%% End:

\bibliography{references}

@string{cav = "CAV"}

@string{cc = "Comp.\ Construct."}

@string{entcs = "Electr.\ Notes Theor.\ Comp.\ Sci."}

@string{fmcad = "FMCAD"}

@string{pldi = "PLDI"}

@string{popl = "POPL"}

@string{prentice = "Prentice-Hall"}

@string{sas = "SAS"}

@string{vmcai = "VMCAI"}

@book{Aho:2006:CPT,
 author  = {Alfred V. Aho and Monica S. Lam and Ravi Sethi and Jeffrey D. Ullman},
 title   = {Compilers: Principles, Techniques, and Tools},
 edition  = {2nd},
 publisher = {Addison-Wesley},
 year   = {2006}
}

@book{Hennessy:2011:CAQ,
 author  = {John L. Hennessy and David A. Patterson},
 title   = {Computer Architecture: A Quantitative Approach},
 edition  = {5th},
 publisher = {Morgan Kaufmann},
 year   = {2011}
}

@InProceedings{PLDI:KBFR17,
    Author    = "Z. Kincaid and J. Breck and A. {Forouhi Boroujeni} and T. Reps",
    Title     = "Compositional Recurrence Analysis Revisited",
    Booktitle = pldi,
    Year      = 2017
}

@InProceedings{FMCAD:FK2015,
  author    = "A. Farzan and Z. Kincaid",
  title     = "Compositional Recurrence Analysis",
  booktitle = fmcad,
  year      = 2015
}

@inproceedings{Cousot1977,
 author = {Cousot, Patrick and Cousot, Radhia},
 title = {Abstract interpretation: a unified lattice model for static analysis of programs by construction or approximation of fixpoints},
 booktitle = {POPL},
 year = {1977},
 pages = {238--252},
 location = {Los Angeles, California},
 publisher = {ACM},
 }

@article{JACM:Tarjan1981,
 author = {Tarjan, Robert Endre},
 title = {Fast Algorithms for Solving Path Problems},
 journal = {J. ACM},
 issue_date = {July 1981},
 volume = {28},
 number = {3},
 month = jul,
 year = {1981},
 pages = {594--614},
}

@article{JACM:Tarjan1981b,
 author = {Tarjan, Robert Endre},
 title = {A Unified Approach to Path Problems},
 journal = {J. ACM},
 issue_date = {July 1981},
 volume = {28},
 number = {3},
 month = jul,
 year = {1981},
 pages = {577--593},
}

@article{ENTCS:ACI2010,
 author = {Ancourt, C. and Coelho, F. and Irigoin, F.},
 title = {A Modular Static Analysis Approach to Affine Loop Invariants Detection},
 journal = entcs,
 issue_date = {October, 2010},
 volume = {267},
 number = {1},
 month = oct,
 year = {2010},
 pages = {3--16},
}

@inproceedings{POPL:CC1977,
  Author    = "P. Cousot and R. Cousot",
  Title     = "Abstract Interpretation: A Unified Lattice Model for Static Analysis of Programs by Construction or Approximation of Fixpoints",
  Booktitle = popl,
  Year      = 1977
}

@inproceedings{CAV:SK2019,
  author    = {Jake Silverman and
               Zachary Kincaid},
  editor    = {Isil Dillig and
               Serdar Tasiran},
  title     = {Loop Summarization with Rational Vector Addition Systems},
  booktitle = {Computer Aided Verification - 31st International Conference, {CAV}
               2019, New York City, NY, USA, July 15-18, 2019, Proceedings, Part
               {II}},
  series    = {Lecture Notes in Computer Science},
  volume    = {11562},
  pages     = {97--115},
  publisher = {Springer},
  year      = {2019},
  url       = {https://doi.org/10.1007/978-3-030-25543-5\_7},
  doi       = {10.1007/978-3-030-25543-5\_7},
  bibsource = {dblp computer science bibliography, https://dblp.org}
}

@inproceedings{SAS:Kincaid2018,
  author    = {Zachary Kincaid},
  editor    = {Andreas Podelski},
  title     = {Numerical Invariants via Abstract Machines},
  booktitle = {Static Analysis - 25th International Symposium, {SAS} 2018, Freiburg,
               Germany, August 29-31, 2018, Proceedings},
  series    = {Lecture Notes in Computer Science},
  volume    = {11002},
  pages     = {24--42},
  publisher = {Springer},
  year      = {2018},
  url       = {https://doi.org/10.1007/978-3-319-99725-4\_3},
  doi       = {10.1007/978-3-319-99725-4\_3},
  bibsource = {dblp computer science bibliography, https://dblp.org}
}

@InProceedings{CAV:ZK2021,
author="Zhu, Shaowei
and Kincaid, Zachary",
editor="Silva, Alexandra
and Leino, K. Rustan M.",
title="Reflections on Termination of Linear Loops",
booktitle="Computer Aided Verification",
year="2021",
publisher="Springer International Publishing",
address="Cham",
pages="51--74",
isbn="978-3-030-81688-9"
}

@inproceedings{PLDI:ZK2021,
author = {Zhu, Shaowei and Kincaid, Zachary},
title = {Termination Analysis without the Tears},
year = {2021},
isbn = {9781450383912},
publisher = {Association for Computing Machinery},
address = {New York, NY, USA},
url = {https://doi.org/10.1145/3453483.3454110},
doi = {10.1145/3453483.3454110},
booktitle = {Proceedings of the 42nd ACM SIGPLAN International Conference on Programming Language Design and Implementation},
pages = {1296–1311},
numpages = {16},
location = {Virtual, Canada},
series = {PLDI 2021}
}

@incollection{CAV:ZK2024,
  title = {Breaking the {{Mold}}: {{Nonlinear Ranking Function Synthesis Without Templates}}},
  shorttitle = {Breaking the {{Mold}}},
  booktitle = {Computer {{Aided Verification}}},
  author = {Zhu, Shaowei and Kincaid, Zachary},
  editor = {Gurfinkel, Arie and Ganesh, Vijay},
  year = {2024},
  volume = {14681},
  pages = {431--452},
  doi = {10.1007/978-3-031-65627-9_21},
  urldate = {2025-02-28}
}

@article{POPL:CK2024,
  title = {Solvable {{Polynomial Ideals}}: {{The Ideal Reflection}} for {{Program Analysis}}},
  shorttitle = {Solvable {{Polynomial Ideals}}},
  author = {Cyphert, John and Kincaid, Zachary},
  year = {2024},
  month = jan,
  journal = {Proceedings of the ACM on Programming Languages},
  volume = {8},
  number = {POPL},
  pages = {724--752},
  issn = {2475-1421},
  doi = {10.1145/3632867},
  urldate = {2025-02-28}
}

@article{UV:LM2010,
  title = {Usable {{Auto-Active Verification}}},
  author = {Leino, K Rustan M and Moskal, Micha{\l}},
  booktitle={Usable Verification},
  year={2010}
}

@article{ENTCS:ML2012,
  title = {Stratified {{Static Analysis Based}} on {{Variable Dependencies}}},
  author = {Monniaux, David and Le Guen, Julien},
  year = {2012},
  month = dec,
  journal = {Electronic Notes in Theoretical Computer Science},
  volume = {288},
  pages = {61--74},
  issn = {15710661},
  doi = {10.1016/j.entcs.2012.10.008},
  urldate = {2025-04-11},
  copyright = {https://www.elsevier.com/tdm/userlicense/1.0/}
}

@inproceedings{SAS:APS2010,
  title = {Deriving {{Numerical Abstract Domains}} via {{Principal Component Analysis}}},
  booktitle = {Static {{Analysis}}},
  author = {Amato, Gianluca and Parton, Maurizio and Scozzari, Francesca},
  editor = {Hutchison, David and Kanade, Takeo and Kittler, Josef and Kleinberg, Jon M. and Mattern, Friedemann and Mitchell, John C. and Naor, Moni and Nierstrasz, Oscar and Pandu Rangan, C. and Steffen, Bernhard and Sudan, Madhu and Terzopoulos, Demetri and Tygar, Doug and Vardi, Moshe Y. and Weikum, Gerhard and Cousot, Radhia and Martel, Matthieu},
  year = {2010},
  volume = {6337},
  pages = {134--150},
  doi = {10.1007/978-3-642-15769-1_9},
  urldate = {2025-04-11}
}

@inproceedings{SAS:RZ2018,
  title = {Invertible {{Linear Transforms}} of {{Numerical Abstract Domains}}},
  booktitle = {Static {{Analysis}}},
  author = {Ranzato, Francesco and Zanella, Marco},
  editor = {Podelski, Andreas},
  year = {2018},
  volume = {11002},
  pages = {344--363},
  doi = {10.1007/978-3-319-99725-4_21},
  urldate = {2025-03-04}
}

@article{OOPSLA:PK2024,
  title = {Monotone {{Procedure Summarization}} via {{Vector Addition Systems}} and {{Inductive Potentials}}},
  author = {Pimpalkhare, Nikhil and Kincaid, Zachary},
  year = {2024},
  month = oct,
  journal = {Proceedings of the ACM on Programming Languages},
  volume = {8},
  number = {OOPSLA2},
  pages = {1873--1899},
  issn = {2475-1421},
  doi = {10.1145/3689777},
  urldate = {2025-02-28}
}

@InProceedings{CAV:LP2016,
author="Leino, K. R. M.
and Pit-Claudel, Cl{\'e}ment",
editor="Chaudhuri, Swarat
and Farzan, Azadeh",
title="Trigger Selection Strategies to Stabilize Program Verifiers",
booktitle="Computer Aided Verification",
year="2016",
publisher="Springer International Publishing",
address="Cham",
pages="361--381",
}

@article{Coverity,
author = {Bessey, Al and Block, Ken and Chelf, Ben and Chou, Andy and Fulton, Bryan and Hallem, Seth and Henri-Gros, Charles and Kamsky, Asya and McPeak, Scott and Engler, Dawson},
title = {A few billion lines of code later: using static analysis to find bugs in the real world},
year = {2010},
issue_date = {February 2010},
publisher = {Association for Computing Machinery},
address = {New York, NY, USA},
volume = {53},
number = {2},
issn = {0001-0782},
url = {https://doi.org/10.1145/1646353.1646374},
doi = {10.1145/1646353.1646374},
journal = {Commun. ACM},
month = feb,
pages = {66–75},
numpages = {10}
}

@inproceedings{CAV:KRC2021,
  title = {Algebraic {{Program Analysis}}},
  booktitle = {Computer {{Aided Verification}}},
  author = {Kincaid, Zachary and Reps, Thomas and Cyphert, John},
  editor = {Silva, Alexandra and Leino, K. Rustan M.},
  year = {2021},
  volume = {12759},
  pages = {46--83},
  doi = {10.1007/978-3-030-81685-8_3},
  urldate = {2025-02-26}
}

@article{POPL:CBKR2019,
  title = {Refinement of Path Expressions for Static Analysis},
  author = {Cyphert, John and Breck, Jason and Kincaid, Zachary and Reps, Thomas},
  year = {2019},
  month = jan,
  journal = {Proceedings of the ACM on Programming Languages},
  volume = {3},
  number = {POPL},
  pages = {1--29},
  issn = {2475-1421},
  doi = {10.1145/3290358},
  urldate = {2025-06-09}
}

@book{Book:AHS2009,
  author       = {Jir{\'{\i}} Ad{\'{a}}mek and
                  Horst Herrlich and
                  George E. Strecker},
  title        = {Abstract and Concrete Categories - The Joy of Cats},
  publisher    = {Dover Publications},
  year         = {2009},
  url          = {http://store.doverpublications.com/0486469344.html},
  isbn         = {978-0-486-46934-8},
  bibsource    = {dblp computer science bibliography, https://dblp.org}
}

@article{MZ:Kainen1971,
  title = {Weak Adjoint Functors},
  author = {Kainen, Paul C.},
  date = {1971-03-01},
  journaltitle = {Mathematische Zeitschrift},
  shortjournal = {Mathematische Zeitschrift},
  volume = {122},
  number = {1},
  pages = {1--9},
  issn = {1432-1823},
  doi = {10.1007/BF01113560},
  url = {https://doi.org/10.1007/BF01113560},
  year ={1971}
}

@incollection{SGA1,
     author = {Grothendieck, Alexander},
     title = {Technique de descente et th\'eor\`emes d'existence en g\'eom\'etrie alg\'ebrique. {I.} {G\'en\'eralit\'es.} {Descente} par morphismes fid\`element plats},
     booktitle = {S\'eminaire Bourbaki : ann\'ees 1958/59 - 1959/60, expos\'es 169-204},
     series = {S\'eminaire Bourbaki},
     note = {talk:190},
     pages = {299--327},
     publisher = {Soci\'et\'e math\'ematique de France},
     number = {5},
     year = {1960},
     mrnumber = {1603475},
     zbl = {0229.14007},
     language = {fr},
     url = {https://www.numdam.org/item/SB_1958-1960__5__299_0/}
}

@book{Book:BW1995,
  author       = {Michael Barr and
                  Charles Wells},
  title        = {Category theory for computing science {(2.} ed.)},
  series       = {Prentice Hall international series in computer science},
  publisher    = {Prentice Hall},
  year         = {1995},
  isbn         = {978-0-13-323809-9},
  bibsource    = {dblp computer science bibliography, https://dblp.org}
}

@InProceedings{CAV:SDDA2011,
author="Sharma, Rahul
and Dillig, Isil
and Dillig, Thomas
and Aiken, Alex",
editor="Gopalakrishnan, Ganesh
and Qadeer, Shaz",
title="Simplifying Loop Invariant Generation Using Splitter Predicates",
booktitle="Computer Aided Verification",
year="2011",
publisher="Springer Berlin Heidelberg",
address="Berlin, Heidelberg",
pages="703--719",
isbn="978-3-642-22110-1"
}

@inproceedings{MFPS:KRD2023,
  author    = {Katsumata, Shin-ya and Rival, Xavier and Dubut, J\'{e}r\'{e}my},
  title     = {A Categorical Framework for Program Semantics and Semantic Abstraction},
  booktitle = {Math.\ Found.\ of Prog.\ Semantics ({MFPS} {XXXIX})},
  series    = {ENTICS},
  volume    = {3},
  year      = {2023},
  doi       = {10.46298/entics.12288}
}

@article{RAIRO:SJM1992,
  author    = {Steffen, Bernhard and Jay, C. Barry and Mendler, Michael},
  title     = {Compositional Characterization of Observable Program Properties},
  journal   = {RAIRO Theor.\ Informatics Appl.},
  volume    = {26},
  number    = {5},
  pages     = {403--424},
  year      = {1992},
  doi       = {10.1051/ita/1992260504031}
}

@article{PLDI:LL2024,
  author    = {Lesbre, Dorian and Lemerre, Matthieu},
  title     = {Compiling with Abstract Interpretation},
  journal   = {Proc.\ ACM Program.\ Lang.},
  volume    = {8},
  number    = {PLDI},
  pages     = {368--393},
  year      = {2024},
  doi       = {10.1145/3656392}
}

@inproceedings{CC:LF2008,
  author    = {Logozzo, Francesco and F\"{a}hndrich, Manuel},
  title     = {On the Relative Completeness of Bytecode Analysis Versus Source Code Analysis},
  booktitle = cc,
  pages     = {197--212},
  year      = {2008},
  doi       = {10.1007/978-3-540-78791-4_14}
}

@inproceedings{VMCAI:Cousot2015,
  author    = {Cousot, Patrick},
  title     = {Abstracting Induction by Extrapolation and Interpolation},
  booktitle = vmcai,
  pages     = {19--42},
  year      = {2015},
  doi       = {10.1007/978-3-662-46081-8_2}
}

@inproceedings{CAV:CGGMP2005,
  author    = {Costan, Alexandru and Gaubert, St\'{e}phane and Goubault, Eric and Martel, Matthieu and Putot, Sylvie},
  title     = {A Policy Iteration Algorithm for Computing Fixed Points in Static Analysis of Programs},
  booktitle = cav,
  pages     = {462--475},
  year      = {2005},
  doi       = {10.1007/11513988_46}
}

\newpage
\appendix

\section{Algebraic Properties of Lifted Operations}

\subsection{Combinations of Lifted Operations} \label{sec:domain-combinations}

Consider a desirable analysis from $\cat{C}$ to $\cat{R}$, where
$\cat{C}$ is a concrete category that has products (e.g., $\TF$), $\cat{R}$ is a concrete category (for analysis results, e.g., $\TF$ for iteration operations and $\mathbf{SF}$ for non-termination analysis), and both categories are concrete over $\cat{B}$.
Suppose that we have two implementations of the analysis (induced from two subcategories $\cat{D}, \cat{E}$ of $\cat{C}$ that admit best abstractions). 
We now show that a new candidate subcategory $\mathbf{D}\times\mathbf{E}$ also admits best abstractions. 
Therefore, one can also have an induced analysis lifted from $\mathbf{D}\times\mathbf{E}$, with all the robustness and soundness properties we have already proved.

\begin{restatable}{lemma}{wlainducedfromproduct} \label{lem:product}
Suppose the base category $\cat{B}$ has products. Suppose that $\cat{C}$ is a concrete category over $\cat{B}$, also with products. If concrete functors $\gamma_1: \cat{D} \to \cat{C}$ and $\gamma_2: \cat{E} \to \cat{C}$ both have weak left adjoints, then the functor $\gamma: \mathbf{D} \times \mathbf{E} \rightarrow \cat{C}$ defined by $\gamma(\tuple{d,e}) = \gamma_1(d) \times \gamma_2(e)$ also has a weak left adjoint.
\end{restatable}

Before discussing examples, we clarify the notion of products in $\TF$. 
The product of $F_1 \in \TF(X_1)$ and $F_2 \in \TF(X_2)$ is their conjunction over a disjoint union of variables, $F_1 \land F_2 \in \TF(X_1 \cup X_2)$ (assuming $X_1, X_2$ are disjoint sets of variables, renaming if necessary).
One could verify that $F_1 \land F_2$  along with the projections $\pi_1: \exists X_2. F_1 \land F_2$ and $\pi_2: \exists X_1. F_1 \land F_2$ satisfy the universal property of products in $\TF$.

\begin{example} \label{ex:polyhedral-iteration-operator-through-combination}
    In the main text, we have seen two ways to induce robust iteration operators from subcategories with computable transitive closures. The guard analysis captures information about the domain and range of a transition, whereas the recurrence analysis captures relations between pre- and post-states. These analyses can actually be combined into a more powerful analysis.

The reflexive transitive closure obtained through the linear recurrence analysis has form $\exists k \in \mathbb{N}. k \geq 0 \land \bigwedge_{x \in X} x' \leq x + k a_x$ where there is a loop counter $k$.
The reflexive transitive closure based on polyhedral guard analysis has form $ X' = X \lor (\textit{Pre}(F) \land \textit{Post}(F))$,
where there is a disjunction based on whether at least one $F$ transition has taken place. 
To properly combine these analyses, we need them to ``sync'' on the number of loops. 
This requires transforming both into a \emph{exponentiation operator} $\functor{exp}(F, k)$ that over-approximates the effect of exactly $k$ iterations of a transition formula $F$.
Let $G$ be the loop body transition formula. The \textbf{polyhedral guard analysis} can be transformed into:
    \[ \functor{exp}_{\text{PGA}}(G, k) \defeq (k=0 \land X'=X) \lor (k \ge 1 \land \functor{Pre}(G) \land \functor{Post}(G)) \]
and the \textbf{linear recurrence analysis}:
    \[ \functor{exp}_{\text{LRA}}(G, k) \defeq \bigwedge_{i=1}^{n}t_{i}[\Delta_{X}\mapsto X'] \le t_{i}[\Delta_{X}\mapsto X] + k b_{i} \ . \]
Combining these using a shared iteration counter $k$ gives
\[
     G^\pstar \defeq \exists k \in \mathbb{N}. K \geq 0 \land \functor{exp}_{\text{PGA}}(G, k) \land \functor{exp}_{\text{LRA}}(G, k)
\]
which matches the polyhedral iteration operator we introduced in Section~\ref{sec:overview}.
A general theory for combination of induced analysis will be discussed in Section~\ref{sec:domain-combinations}.
\end{example}

\begin{example}
As discussed in Example~\ref{ex:polyhedral-iteration-operator-through-combination}, a powerful way to combine two iteration analyses is to share the iteration counter. This can be seen as an instance of lifting from a product category. Let $\cat{C} = \TF$ and let the result category $\cat{R}$ be the category of functions from $\mathbb{N}$ to $\TF$, where an arrow from $\lambda k. F(k)$ to $\lambda k. G(k)$ is a simulation $s$ such that $\barr{F(k)} {s} {G(k)}$ for all $k$.
Let $\cat{D}$ and $\cat{E}$ be subcategories of $\TF$ admitting best abstractions as exponentiation operators $exp_D$ and $exp_E$. The operations on the subcategories are $F_D(d) = \lambda k. exp_D(d, k)$ and $F_E(e) = \lambda k. exp_E(e, k)$.
The operation on the product category $\cat{D} \times \cat{E}$ is then $F_{D \times E}(\tuple{d,e}) = \lambda k. (exp_D(d, k) \land exp_E(e, k))$.
Lifting this operation gives the combined analysis. For a formula $A \in \TF$, its abstraction in $\cat{D} \times \cat{E}$ is 
$\alpha(A) = \tuple{\alpha_D(A), \alpha_E(A)}$. The lifted result is
\begin{align*}
    \overline{F}_{D \times E}(A) &= \eta_A^{-1}(F_{D \times E}(\alpha(A))) \\
    &= \lambda k. (exp_D(\alpha_D(A), k) \land exp_E(\alpha_E(A), k))[\eta_A, \eta_A'] \\
    &=  \lambda k. exp_D(\alpha_D(A), k)[\eta_{A_1}, \eta_{A_1'}] \land exp_E(\alpha_E(A), k)[\eta_{A_2}, \eta_{A_2'}] 
\end{align*}
where $\eta_{A_1}, \eta_{A_2}$ are the components of $\eta_A$ that act on disjoint variable sets.
\end{example}

\begin{example}
Consider combining two different non-termination analyses. Let $\cat{C} = \TF$ and $\cat{R} = \cat{SF}$. Suppose we have two subcategories $\cat{D}$ and $\cat{E}$ with exact non-termination analyses $\omega_D: \cat{D} \to \cat{SF}$ and $\omega_E: \cat{E} \to \cat{SF}$.
Let us consider the lifting from the product category $\cat{D} \times \cat{E}$,
where the lifted operation is $\overline{\omega}_{D \times E}(A) = \eta_A^{-1}(\omega_{D \times E}(\alpha(A)))$.
Here, $\alpha(A)=\tuple{\alpha_D(A), \alpha_E(A)}$, and $\eta_A$ is from the weak left adjoint to $\gamma: \cat{D}\times \cat{E} \rightarrow \cat{C}$, which corresponds to the union of the individual substitutions, $\eta_A = \eta_{A_1} \cup \eta_{A_2}$. 
The resulting analysis result is $(\omega_D(\alpha_D(A)) \land \omega_E(\alpha_E(A)))[\eta_{A_1} \cup \eta_{A_2}]$. Since the variables of $\omega_D(\alpha_D(A))$ and $\omega_E(\alpha_E(A))$ are disjoint, this is equivalent to $\omega_D(\alpha_D(A))[\eta_{A_1}] \land \omega_E(\alpha_E(A))[\eta_{A_2}]$.
\end{example}

\subsection{Lifting Algebraic Laws} \label{sec:algebraic-laws}

In the context of algebraic program analysis, 
it is sometimes possible to prove that algebras satisfy certain laws, which can be interpreted as robustness with respect to certain kinds of programs transformations.  Equational laws correspond to transformations that do not change the results of an analysis, while inequation laws correspond to transformations that either increase or decrease precision (not necessarily strictly).  We now show that a number of desirable laws are preserved by the lifting strategy described earlier in this section.

Sometimes, the algebraic laws may contain desirable properties that involve more than one operation.
For example, \citet{PLDI:ZK2021} formulates algebraic termination analysis using $\omega$-regular algebra,
where it also uses an $\omega$ operator to compute the (non-)terminating starting states for transition formulas in
addition to the iteration operator $*$ and requires a law of $T^* \cdot T^\omega \leq T^\omega$ for the analysis to be monotone.

The following lemma gives us a general method of proving inequalities of the form $E \leq T^\mathit{op}$ where $E$ is an expression that also involves $T^\mathit{op}$ (e.g., $T^*T^* \leq T^*$, $(T^n)^* \leq T^*$, $(T^n)^\omega \leq T^\omega$), or maybe even more than one operation (e.g., $T^* \cdot T^\omega \leq T^\omega$).  This is formalized by treating the expression $E$ as the composition of three functors $F$, $G$, and $H$ where $G$ represents the operation, $H$ is an expression \textit{inside} the operation, and $F$ is an expression \textit{outside} the operation.
Additionally, two functors $F_1, F_2$ and an intermediate result category $\cat{R}$ are introduced to accommodate the case where
we want to reason about the lifting of laws involving multiple operations.

\generalizedExtension*

The theorem can be used to prove the following properties of lifted operations. 
In the following, for any natural number $n$, $(-)^n$ denotes the functor sending each transition system $T$ to its $n$-fold (relational) composition with itself.

\begin{restatable}[Transitivity]{corollary}{transitivitylaw} \label{cor:law-transitivity}
 Let $\mathbf{D}$ be a subcategory of $\TF$, $\gamma : \cat{D} \rightarrow \cat{TF}$ be a concrete functor with weak left adjoint, 
 and suppose that $(-)^{\star} : \mathbf{D} \rightarrow \mathbf{TF}$ is a concrete functor satisfying $T^{\star}T^{\star} = T^{\star}$ and $1 \preceq T^{\star}$ for all $T \in \Ob{D}$. 
 Let $(-)^{\lifted{\star}}$ denote the lifting of $(-)^{\star}$ to $\mathbf{TF}$. Then we have $T^{\lifted{\star}}T^{\lifted{\star}} = T^{\lifted{\star}}$.
\end{restatable}
\begin{proof}[Proof Sketch]
To show $T^{\lifted{\star}}T^{\lifted{\star}} \preceq T^{\lifted{\star}}$ from the assumption $T^{\star}T^{\star} \preceq T^{\star}$, we use \autoref{thm:generalized_extension}. Let $\mathbf{C} =  \cat{R} = \cat{R'} = \mathbf{TF}$. 
Let $\functor{G} = (-)^{\star}$, so its lifting is $\lifted{\functor{G}} = (-)^{\lifted{\star}}$. 
Let $\functor{H} = 1_{\mathbf{D}}$, $\widetilde{\functor{H}} = 1_{\mathbf{TF}}$, $\functor{F}_1(T) = T \cdot T$, and $\functor{F}_2 = 1_{\TF}$. 
To prove the other direction, we start from $1 \preceq T^{\star}$ which implies $1 \preceq T^{\lifted{\star}}$ by \autoref{thm:lift-soundness}, and subsequently get $T^{\lifted{\star}} \preceq T^{\lifted{\star}} T^{\lifted{\star}}$.
\end{proof}

\begin{restatable}[Unrolling]{corollary}{unrollinglaw} \label{cor:law-unrolling}
 Let $\mathbf{D}$ be a subcategory of $\TF$ that is closed under composition. Suppose that the inclusion functor $\gamma: \mathbf{D} \to \mathbf{TF}$ has a weak left adjoint.
 Let $(-)^{\star} : \mathbf{D} \rightarrow \mathbf{TF}$ be a functor satisfying $(T^n)^{\star} \preceq T^{\star}$ for all $T \in \Ob{D}$ and $n \in \mathbb{N}$. Let $(-)^{\lifted{\star}}$ be the lifting of $(-)^{\star}$ to $\mathbf{TF}$. Then $(T^n)^{\lifted{\star}} \preceq T^{\lifted{\star}}$ for all $T \in \Ob{\TF}$ and $n$.
\end{restatable}
\begin{proof}[Proof Sketch]
Using \autoref{thm:generalized_extension}, where
$\cat{C} = \cat{R} = \cat{R}' = \mathbf{TF}$,
$\functor{H}(T) = T^n$ for $T \in \Ob{\cat{D}}$ (since $\mathbf{D}$ is closed under composition),
$\widetilde{\functor{H}}(T) = T^n$ for $T \in \Ob{\cat{C}}$,
$\functor{G} = (-)^{\star}$,
and $\functor{F}_1 = \functor{F}_2 = 1_{\mathbf{TF}}$.
\end{proof}

\begin{restatable}{corollary}{staromegalaw} \label{cor:star-omega}
 Let $\mathbf{D}$ be a subcategory of $\mathbf{TF}$. Suppose that the inclusion functor $\gamma: \mathbf{D} \to \mathbf{TF}$ has a weak left adjoint, and let $(-)^{\star} : \mathbf{D} \rightarrow \mathbf{TF}$ and $(-)^\omega : \mathbf{D} \rightarrow \mathbf{SF}$ be functors such that $F^{\star} \cdot F^\omega \preceq F^\omega$ for all $F$ in $\mathbf{D}$ (here $\cdot$ is an operator in the $\omega$-regular algebra). Let $(-)^{\lifted{\star}}$ be the lifting of $(-)^{\star}$ to $\mathbf{TF}$, and let $(-)^{\overline{\omega}}$ be the lifting of $(-)^\omega$ to $\mathbf{TF} \to \mathbf{SF}$. Then $F^{\lifted{\star}} \cdot F^{\overline{\omega}} \preceq F^{\overline{\omega}}$ for all $F$ in $\Ob{\mathbf{TF}}$.
\end{restatable}
\begin{proof}[Proof Sketch]
Using  \autoref{thm:generalized_extension}. Put $\cat{C} = \TF$, $\mathbf{R} = \mathbf{TF} \times \mathbf{SF}$, $\mathbf{R}' = \mathbf{SF}$, $\functor{H} = 1_{\mathbf{D}}$, and $\widetilde{\functor{H}} = 1_{\mathbf{TF}}$. Also, let $\functor{G}(F) = (F^{\star}, F^\omega)$ (with its lifting  $\lifted{\functor{G}}(F) = (F^{\lifted{\star}}, F^{\overline{\omega}})$),
$\functor{F}_1(T, S) = T \cdot S$
and $\functor{F}_2(T, S) = S$ (noting that $\functor{F}_2$ is cartesian).
\end{proof}

\section{Proofs} \label{sec:proofs}

\robustImpliesConcreteFunctor*

\begin{proof}
That $\proj{C} = \proj{D} \circ F$ follows immediately from the definition of robustness.  It remains to show that $F$ is a functor.
First, $F$ preserves identities: for any object $X$, $F(1_X) = 1_{F(X)}$.  This follows from the fact that $\proj{D}$ is faithful, and
$ \proj{D}(F(1_X)) = \proj{C}(1_X) = 1_{\proj{C}(X)} = 1_{\proj{D}(F(X))} = \proj{D}(1_{F(X)})\ . $
Also $F$ respects composition: for any compatible arrows $f$ and $g$, $F(f \circ g) = F(f) \circ F(g)$.  This follows from the fact that $\proj{D}$ is faithful, and
\[
\proj{D}(F(f \circ g)) = \proj{C}(f \circ g) = \proj{C}(f) \circ \proj{C}(g) = \proj{D}(F(f)) \circ \proj{D}(F(g)) = \proj{D}(F(f) \circ F(g))\ . \qedhere
\]
\end{proof}

\soundnessfromnattransformation*

\begin{proof}
  The fact that existence of a concrete natural transformation implies soundness is immediate.  For the other direction, for any object $A$, we define $\epsilon_A$ to be the unique arrow in $\mathbf{R}^{\sharp}(F^\natural(A),\gamma(F^\sharp(A)))$ such that
  $\proj{R^\natural}(\epsilon_A) = 1_{\proj{C}(A)}$.  The fact that such an arrow exists follows from soundness, while uniqueness follows from the fact that $\proj{R^\natural}$ is faithful.  Last we must show that $\epsilon$ is a natural transformation: that is, for any arrow $f \in \mathbf{C}(A,B)$, the following diagram commutes:
    \begin{center}
  \begin{tikzpicture}[thick]
    \matrix (m) [matrix of math nodes, row sep=2.5em, column sep=3em,
      text height=1.5ex, text depth=0.25ex] { F^\natural(A) & \gamma(F^\sharp(A))\\
      F^\natural(B)  & \gamma(F^\sharp(B))\\ };

    \draw (m-2-2) edge[<-] node[right]{$\gamma(F^\sharp(f))$} (m-1-2);
    \draw (m-1-1) edge[->] node[above]{$\epsilon_A$} (m-1-2);

    \draw (m-2-1) edge[<-] node[left]{$F^\natural(f)$} (m-1-1);
    \draw (m-2-1) edge[->] node[below]{$\epsilon_B$} (m-2-2);
  \end{tikzpicture}
  \end{center}
  This follows from the fact that the image of this diagram under $\proj{R^\natural}$ commutes (i.e., 
  $\proj{R^\natural}(\gamma(F^\sharp(f)) \circ \epsilon_A) = \proj{C}(f) = \proj{R^\natural}(\epsilon_B \circ F^\natural(f))$ and that $\proj{R^\natural}$ is faithful.
\end{proof}

\begin{corollary}
    Polyhedral guard analysis $(-)^{\star_{\text{PGA}}}$ is sound and robust. \label{cor:polyhedral-guard-analysis-sound-robust}
\end{corollary}
\begin{proof}
We first recall the definition of polyhedral guard analysis as follows:
Define $\cat{PolyCart}$ to be the subcategory of $\TF$
consisting of polyhedral cartesian transition formulas, i.e., formulas of the form
$(\bigwedge_{i=1}^n s_i \geq a_i) \land (\bigwedge_{j=1}^m t_j \geq b_j)$ where each $s_i$ is a linear term over unprimed variables and each $t_i$ is a linear term over primed variables.
Observe that every polyhedral cartesian transition formula is transitively closed, so we can define a reflexive transitive closure operator $(-)^\star : \cat{PolyCart} \rightarrow \TF$ by defining $ F^\star \defeq X' = X \lor F $
for any $F \in \TF(X)$.
Let $\gamma : \cat{PolyCart} \rightarrow \TF$ be an inclusion functor. It remains to show that $\gamma$ has a weak left adjoint $\tuple{\alpha,\eta}$.

Define $\alpha(F) \defeq \textit{Pre}(F) \land \textit{Post}(F)$ and $\eta_F$ to be the identity.  To prove that $\tuple{\alpha,\eta}$ is weak left adjoint to $\gamma$, we must show that for any transition formula of form $G = G_{\textit{pre}} \land G_{\textit{post}}$ and any linear simulation
$s$ such that $\barr{F}{s}{\gamma(G)}$, there is a linear simulation $\overline{s}$ such that $\barr{\alpha(F)}{\overline{s}}{G}$ and $s = \overline{s} \circ \eta_F$.  Since $\eta_F$ is the identity, we must have $\overline{s} = s$, and so we need to show that $\barr{\alpha(F)}{s}{G}$---that is, $s$ is a simulation from $\alpha(F)$ to $G$, or equivalently $\alpha(F) \models G[s,s']$.  Since $s$ is linear, $H \defeq G[s,s']$ is \textit{also} in $\cat{PolyCart}$, taking the form $H_{\textit{pre}} \land H_{\textit{post}}$.  Since $F \models H$ by assumption, we have
$\exists X'.F \models \exists X'. H \equiv H_{\textit{pre}}$; since $H_{\textit{pre}}$ is convex and $\textit{Pre}(F)$ is the convex hull of $\exists X'.F$, we must have $\textit{Pre}(F) \models H_\textit{pre}$.  Similarly, we have
$\textit{Post}(F) \models H_{\textit{post}}$, and thus
$\alpha(F) \defeq \textit{Pre}(F) \land \textit{Post}(F) \models H_{\textit{pre}} \land H_{\textit{post}} = H = G[s,s'], $
and so $s$ is a simulation from $\alpha(F)$ to $G$.

Finally, observe that polyhedral guard analysis is precisely the iteration operator induced by $\cat{PolyCart}$:
$
F^{\star_{\text{PGA}}} \defeq X' = X \lor (\textit{Pre}(F) \land \textit{Post}(F)) = \alpha(F)^\star[\eta_F,\eta_F']\ .
$
It thus follows from Theorem~\ref{thm:lift-robustness} that $(-)^{\star_{\text{PGA}}}$ is robust.
\end{proof}

\begin{corollary}
    Linear recurrence analysis $(-)^{\star_{\text{LRA}}}$ is sound and robust. \label{cor:linear-rec-analysis-sound-robust}
\end{corollary}
\begin{proof}
We first recall the definition of linear recurrence analysis.
Define $\mathbf{LT}_\bot$ (lossy translations) to be the subcategory of $\mathbf{TF}$ whose objects are formulas of the form $\bigwedge_{x \in X} x' \leq x + a_x$ or $\false$. Let $\gamma : \mathbf{LT}_\bot \rightarrow \mathbf{TF}$ be the inclusion functor.
Computing the reflexive transitive closure of a formula in $\mathbf{LT}_\bot$ is straightforward, and gives rise to a concrete functor $(-)^\star : \mathbf{LT}_\bot \rightarrow \mathbf{TF}$, where for a formula $F \in \mathbf{LT}_\bot$ over variables $X$:
\begin{align*}
\false^\star &\defeq \bigwedge_{x \in X} x' = x \\
\left(\bigwedge_{x \in X} x' \leq x + a_x\right)^\star &\defeq \exists k \in \mathbb{N}. k \geq 0 \land \bigwedge_{x \in X} x' \leq x + k a_x
\end{align*}

We must show that any transition formula has a best abstraction in $\mathbf{LT}_\bot$---i.e., the concretization functor $\gamma$ has a weak left adjoint $\tuple{\alpha,\eta}$. Recall that linear recurrence analysis works by extracting a set of linear recurrence inequations from a loop body formula $F$ that is satisfiable, by computing the convex hull of the formula
$\Delta(F) \defeq \exists X,X'. F \land \bigwedge_{x \in X} \delta_x = x' - x$, which takes the form
$\bigwedge_{i=1}^n t_i \leq a_i$, where each $t_i$ is a linear term over $\Delta_X \defeq \set{ \delta_x : x \in X}$, and each $a_i$ is a rational number. Then we may observe that the rational numbers $a_1,\dots,a_n$ define an $\cat{LT}_\bot$ formula $\alpha(F) \defeq \bigwedge_{i=1}^n y_i' \leq y_i + a_i$ (over a fresh set of variables $Y = \set{y_1, \dots, y_n}$), and the terms $t_1,\dots,t_n$ define a linear map $\eta_F : \mathbb{Q}^X \rightarrow \mathbb{Q}^Y$---represented in ``transposed form'' as a linear substitutiom---maps each $y_i$ to $t_i[\delta_X \mapsto X]$.
For an unsatisfiable formula $F$ over $X$, we define $\alpha(F) \defeq \false$ (considered as a transition formula over $X$) and $\eta_F$ is the identity substitution.

Now we must argue that $\tuple{\alpha,\eta}$ is in fact weakly left-adjoint to $\gamma$.
Suppose that $U = \bigwedge_{i=1}^m z_i' \leq z_i + c_i$ is a formula in $\mathbf{LT}_\bot$ and that $s$ is a linear simulation from a satisfiable formula $F$ to $\gamma(U) = U$.
We can think of $s$ as a substitution $[z_i \mapsto s_i]_{i =1}^m$ for linear terms $s_1,\dots,s_m$ over the $X$ variables, and the fact that $s$ is a simulation implies
$F \models \bigwedge_{i=1}^m s_i' \leq s_i + c_i$ and thus
\[\Delta(F) \models s_i[X \mapsto \Delta_X] \leq c_i\ .\] Since $\bigwedge_{i=1}^n t_i \leq a_i$ is the convex hull of
$\Delta(F)$, for each $i$ we have
non-negative $\lambda_{i,1},\dots,\lambda_{i,n}$ such that
$s_i = \lambda_{i,1}t_1 + \dots + \lambda_{i,n}t_n$ by Farkas' lemma.
Define a linear substitution $\overline{s}$ by
$[z_i \mapsto \lambda_{i,1}y_1 + \dots + \lambda_{i,n}y_n]_{i=1}^m$, and observe that
$\overline{s}$ is a linear simulation from $\alpha(F)$ to $U$ and that
$\overline{s} \circ \eta = [z_i \mapsto (\lambda_{i,1}y_1 + \dots + \lambda_{i,n}y_n)[\eta_F]]_{i=1}^m = [z_i \mapsto s_i]_{i=1}^m = s$.
The same argument vacuously holds for unsatisfiable $F$ such that $\alpha(F) = \false$; thus, $\tuple{\alpha,\eta}$ is a weakly left-adjoint to $\gamma$.

Finally, observe that linear recurrence analysis is precisely the iteration operator induced by $\cat{LT}_\bot$:
$F^{\star_{\text{LRA}}} = \alpha(F)^\star[\eta_F,\eta_F']$.
It thus follows from Theorem~\ref{thm:lift-robustness} that $(-)^{\star_{\text{LRA}}}$ is robust.
\end{proof}

\wlainducedfromproduct*
\begin{proof}
For the purpose of this proof, we adopt the standard notion of weak left adjoints (which is equivalent by Lemma~\ref{lem:alternate-left-adjoint}).
For the category $\cat{D} \times \cat{E}$, its objects and arrows are tuples of objects and arrows in $\cat{D}$ and $\cat{E}$.
Let $\langle\alpha_1, \epsilon_1\rangle$ be the weak left adjoint of $\gamma_1$, and $\langle\alpha_2, \epsilon_2\rangle$ be the weak left adjoint of $\gamma_2$.
We define a functor $\gamma: \cat{D} \times \cat{E} \rightarrow \cat{C}$ which maps objects $\tuple{d, e} \to \gamma_1(d) \times \gamma_2(e)$ and arrows $\tuple{f_D, f_E} \to \gamma_1(f_D) \times \gamma_2(f_E)$, which is well-defined since $\cat{C}$ has products.

We now construct its weak left adjoint $\langle\alpha, \epsilon\rangle$.
For any object $A \in \text{Ob}(\cat{C})$, we define its abstraction $\alpha(A) \triangleq \langle \alpha_1(A), \alpha_2(A) \rangle$.
The arrow $\epsilon_A \in \cat{C}(A, \gamma(\alpha(A)))$ is the unique arrow to the product induced by the component arrows $\epsilon_{A,1}$ and $\epsilon_{A,2}$.
We must show that for any object $B = \langle B_D, B_E \rangle \in \text{Ob}(\cat{D} \times \cat{E})$, and any arrow $f \in \cat{C}(A, \gamma(B))$, there exists an arrow $\overline{f} \in (\cat{D} \times \cat{E})(\alpha(A), B)$ such that $f = \gamma(\overline{f}) \circ \epsilon_A$.

The argument is visualized by the following diagram:
\begin{center}
\begin{tikzcd}[row sep=3.5em, column sep=3em]
    & A \arrow[bend left=15]{ddl}[near end]{p_1 \circ f} \arrow[bend right=15]{ddr}[swap, near end]{p_2 \circ f} \arrow{dl}[swap]{\epsilon_{A,1}} \arrow{dr}{\epsilon_{A, 2}} \arrow{d}[near end]{\epsilon_A} & \\
\gamma_1(\alpha_1(A)) \arrow[dashed]{d}[swap]{\gamma_1(\overline{f_1})} & \gamma(\alpha(A)) \arrow{l}[swap]{\pi_1} \arrow{r}{\pi_2} \arrow[dashed]{d}{\gamma(\overline{f})} & \gamma_2(\alpha_2(A)) \arrow[dashed]{d}{\gamma_2(\overline{f_2})} \\
\gamma_1(B_D) & \gamma(B) \arrow{l}[swap]{p_1} \arrow{r}{p_2} & \gamma_2(B_E)
\end{tikzcd}
\end{center}

We consider how to construct $\overline{f} = \langle \overline{f_1}, \overline{f_2} \rangle$.
Consider the arrow $f_1 = p_1 \circ f : A \to \gamma_1(B_D)$. Since $\langle \alpha_1, \epsilon_1 \rangle$ is a weak left adjoint to $\gamma_1$, there exists an arrow $\overline{f_1} \in \cat{D}(\alpha_1(A), B_D)$ such that $f_1 = \gamma_1(\overline{f_1}) \circ \epsilon_{A,1}$. This ensures the left-hand part of the diagram commutes.
Similarly, for the arrow $f_2 = p_2 \circ f : A \to \gamma_2(B_E)$, there exists an arrow $\overline{f_2} \in \cat{E}(\alpha_2(A), B_E)$ such that $f_2 = \gamma_2(\overline{f_2}) \circ \epsilon_{A,2}$.
By the universal property of products in $\cat{C}$, the arrow $\gamma(\overline{f}) = \gamma_1(\overline{f_1}) \times \gamma_2(\overline{f_2})$ is the unique arrow making the outer diagram commute. We have shown it commutes on both components when composed with the projections, therefore $f = \gamma(\overline{f}) \circ \epsilon_A$.
\end{proof}

\subsection*{Definition of Weak Left Adjoints}

In the following Lemma~\ref{lem:alternate-left-adjoint}, we show that our notion of weak left adjoints defined in Section~\ref{sec:best-abstraction-weak-left-adj}
is equivalent to the standard definition.

\begin{lemma} \label{lem:alternate-left-adjoint}
Let $\cat{C}$ and $\cat{D}$ be concrete categories over a common base category $\cat{B}$, with faithful functors $\phi_{\cat{C}}:\cat{C} \to \cat{B}$ and $\phi_{\cat{D}}:\cat{D} \to \cat{B}$. Let $\gamma: \cat{D} \to \cat{C}$ be a concrete functor. The following two notions of weak left adjoints are equivalent:
\begin{enumerate}
    \item (Standard notion of weak left adjoints) There is a weak left adjoint $(\alpha, \epsilon)$ to $\gamma$, where $\alpha: \Ob{\cat{C}} \to \Ob{\cat{D}}$ is a function and for each $A \in \Ob{\cat{C}}$, $\epsilon_A: A \to \gamma(\alpha(A))$ is an arrow in $\cat{C}$ such that for any object $K \in \Ob{\cat{D}}$ and any arrow $f: A \to \gamma(K)$ in $\cat{C}$, there exists an arrow $\overline{f}: \alpha(A) \to K$ in $\cat{D}$ for which the diagram
    \[
    \begin{tikzcd}
    A \arrow{r}{\epsilon_A} \arrow{dr}[swap]{f} & \gamma(\alpha(A)) \arrow[dashed]{d}{\gamma(\overline{f})} \\
    & \gamma(K)
    \end{tikzcd}
    \]
    commutes.
    \item (Our definition) There is a pair $(\alpha, \eta)$ where $\alpha: \Ob{\cat{C}} \to \Ob{\cat{D}}$ is a function and for each $A \in \Ob{\cat{C}}$, $\eta_A: \phi_{\cat{C}}(A) \to \phi_{\cat{D}}(\alpha(A))$ is an arrow in $\cat{B}$ such that $\barr{A}{\eta_A}{ \gamma(\alpha(A))}$, and for any object $K \in \Ob{\cat{D}}$ and any arrow $f: \phi_{\cat{C}}(A) \to \phi_{\cat{C}}(\gamma(K))$ in $\cat{B}$ with $\barr{A}{f}{\gamma(K)}$, there exists an arrow $\overline{f}: \phi_{\cat{D}}(\alpha(A)) \to \phi_{\cat{D}}(K)$ in $\cat{B}$ such that $\barr{\alpha(A)}{\overline{f}}{K}$ and $f = \overline{f} \circ \eta_A$. The following is the diagram in $\cat{B}$.
    \[
    \begin{tikzcd}
    \proj{C}(A) \arrow{r}{\eta_A} \arrow{dr}[swap]{f} & \proj{C}(\gamma(\alpha(A)))=\proj{D}(\alpha(A)) \arrow[dashed]{d}{\overline{f}} \\
    & \proj{C}(\gamma(K)) = \proj{D}(K)
    \end{tikzcd}
    \]
\end{enumerate}
\end{lemma}
\begin{proof}
We first prove that (1) implies (2). Assume $(\alpha, \epsilon)$ is a weak left adjoint to $\gamma$. For each $A \in \Ob{\cat{C}}$, define $\eta_A \defeq \phi_{\cat{C}}(\epsilon_A)$. Thus $\barr{A}{\eta_A}{\gamma(\alpha(A))}$.
Now, consider an object $K \in \Ob{\cat{D}}$ and an arrow $f: \phi_{\cat{C}}(A) \to \phi_{\cat{C}}(\gamma(K))$ in $\cat{B}$ such that $\barr{A} {f} {\gamma(K)}$. Since $\phi_{\cat{C}}$ is faithful, there is a unique arrow $f': A \to \gamma(K)$ in $\cat{C}$ such that $\phi_{\cat{C}}(f') = f$. By the definition of a weak left adjoint, there exists an arrow $\overline{f}': \alpha(A) \to K$ in $\cat{D}$ such that $f' = \gamma(\overline{f}') \circ \epsilon_A$.
Let $\overline{f} \defeq \phi_{\cat{D}}(\overline{f}')$. Since $\overline{f}'$ is an arrow in $\cat{D}$, we have $\barr{\alpha(A)}{\overline{f}}{K}$. Applying the functor $\phi_{\cat{C}}$ to the equation $f' = \gamma(\overline{f}') \circ \epsilon_A$, we get:
\[ \phi_{\cat{C}}(f') = \phi_{\cat{C}}(\gamma(\overline{f}') \circ \epsilon_A) = \phi_{\cat{C}}(\gamma(\overline{f}')) \circ \phi_{\cat{C}}(\epsilon_A) \]
Since $\gamma$ is a concrete functor, $\phi_{\cat{C}} \circ \gamma = \phi_{\cat{D}}$. Thus, $\phi_{\cat{C}}(\gamma(\overline{f}')) = \phi_{\cat{D}}(\overline{f}') = \overline{f}$.
Substituting back, we have $f = \overline{f} \circ \eta_A$. This shows that (2) holds.

We then show that (2) implies (1). Assume the pair $(\alpha, \eta)$ satisfies (2). For each $A \in \Ob{\cat{C}}$, since $\barr{A}{\eta_A} {\gamma(\alpha(A))}$, there exists a unique arrow $\epsilon_A: A \to \gamma(\alpha(A))$ in $\cat{C}$ such that $\phi_{\cat{C}}(\epsilon_A) = \eta_A$. We will show that $(\alpha, \epsilon)$ is a weak left adjoint to $\gamma$.
Consider an object $K \in \Ob{\cat{D}}$ and an arrow $f': A \to \gamma(K)$ in $\cat{C}$. 
Let $f = \phi_{\cat{C}}(f')$, thus by definition $\barr{A} {f} {\gamma(K)}$. According to (2), there exists an arrow $\overline{f}: \phi_{\cat{D}}(\alpha(A)) \to \phi_{\cat{D}}(K)$ in $\cat{B}$ such that $\barr{\alpha(A)} {\overline{f}} {K}$ and $f = \overline{f} \circ \eta_A$.
Since $\barr{\alpha(A)} {\overline{f}} {K}$, there is a unique arrow $\overline{f}': \alpha(A) \to K$ in $\cat{D}$ with $\phi_{\cat{D}}(\overline{f}') = \overline{f}$.
We need to show $f' = \gamma(\overline{f}') \circ \epsilon_A$, by applying the faithful functor $\phi_{\cat{C}}$:
\begin{align*}
    \phi_{\cat{C}}(\gamma(\overline{f}') \circ \epsilon_A) &= \phi_{\cat{C}}(\gamma(\overline{f}')) \circ \phi_{\cat{C}}(\epsilon_A) \\
    &= \phi_{\cat{D}}(\overline{f}') \circ \eta_A \\
    &= \overline{f} \circ \eta_A = f = \phi_{\cat{C}}(f')
\end{align*}
Thus $f' = \gamma(\overline{f}') \circ \epsilon_A$, which means $(\alpha, \epsilon)$ is a weak left adjoint to $\gamma$.
\end{proof}

\subsection*{Robustness and Soundness of Lifted Operations}

The main text contains proofs of Theorems~\ref{thm:lift-robustness} and~\ref{thm:lift-soundness}; below we give slightly more spelled-out versions that may be easier to follow.

Recall Theorem~\ref{thm:lift-robustness}:
\liftRobustness*
\begin{proof}
We must show that $\overline{F}$ is a concrete functor from $\mathbf{C}$ to $\mathbf{R}$.
Observe that we have $\overline{F}(A) = \eta_A^{-1}(F(\alpha(A)))$ by definition, and so $\bcarr{\overline{F}(A)}{\eta_A}{F(\alpha(A))}$.  It follows that $\proj{R}(\overline{F}(A)) = \dom(\eta_A) = \proj{C}(A)$.
It remains to show that for any arrow $f \in \mathbf{B}(\proj{C}(A),\proj{C}(B))$ such that
$\barr{A}{f}{B}$, we have $\barr{\overline{F}(A)}{f}{\overline{F}(B)}$.

  Since $\barr{A}{\eta_B}{\gamma(\alpha(B))}$ and $\tuple{\alpha,\eta}$ is a weak left-adjoint of $\gamma$, there is some
  $ \overline{f}$ in $\cat{B}$ that goes from $\proj{D}(\alpha(A))$ to $\proj{D}(\alpha(B)))$
  such that $\barr{\alpha(A)}{\overline{f}}{\alpha(B)}$ and the following diagram commutes:

 \begin{center}
  \begin{tikzpicture}[thick]
    \matrix (m) [matrix of math nodes, row sep=2.5em, column sep=3em,
      text height=1.5ex, text depth=0.25ex] {
      \proj{C}(A) & \proj{D}(\alpha(A))\\
    \proj{C}(B)  & \proj{D}(\alpha(B))\\ };
    \draw (m-2-2) edge[<-,dashed] node[right]{$\overline{f}$} (m-1-2);
    \draw (m-1-1) edge[->] node[above]{$\eta_A$} (m-1-2);

    \draw (m-2-1) edge[<-] node[left]{$f$} (m-1-1);
    \draw (m-2-1) edge[->] node[below]{$\eta_B$} (m-2-2);
  \end{tikzpicture}
  \end{center}

  Since $F$ is a concrete functor, we have
  $\barr{F(\alpha(A))}{\overline{f}}{F(\alpha(B))}$
  and thus $\barr{\overline{F}(A)}{\overline{f} \circ \eta_A}{F(\alpha(B))}$, and since
  $\overline{f} \circ \eta_A = \eta_B \circ f$ we have $\barr{\overline{F}(A)}{\eta_B \circ f}{F(\alpha(B))}$.
  Observing that $\bcarr{f^{-1}(\overline{F}(B))}{f}{\bcarr{\overline{F}(B)}{\eta_B}{F(\alpha(B))}}$, we have
  $\bcarr{f^{-1}(\overline{F}(B))}{\eta_B \circ f}{F(\alpha(B))}$.
  We have thus exhibited two arrows pointing to $F(\alpha(B))$, with the latter one cartesian.
  It follows from the universal property of cartesian arrows ($\dagger$) that
  $\barr{\overline{F}(A)}{1}{f^{-1}(\overline{F}(B))}$.  We thus have
  $\barr{\overline{F}(A)}{1}{\barr{f^{-1}(\overline{F}(B))}{f}{\overline{F}(B)}}$, and so
  $\barr{\overline{F}(A)}{f}{\overline{F}(B)}$.
\end{proof}

We now prove Theorem~\ref{thm:lift-soundness}.
\begin{theorem*}[Soundness of lifted operations]
    Suppose that we have the following (not commuting) diagram of concrete categories over $\mathbf{B}$:

 \begin{minipage}{0.2\textwidth}
  \begin{tikzpicture}[thick]
    \matrix (m) [matrix of math nodes, row sep=2.5em, column sep=3em,
      text height=1.5ex, text depth=0.25ex] { \mathbf{D} & \mathbf{R}^\sharp \\
      \mathbf{C}  & \mathbf{R}^\natural \\ };
    \draw (m-2-2) edge[<-] node[right]{$\gamma_R$} (m-1-2);
    \draw (m-1-1) edge[->] node[above]{$F^\sharp$} (m-1-2);
    \draw (m-1-1) edge[->] node[left]{$\gamma_D$} (m-2-1);
    \draw (m-2-1) edge[->] node[below]{$F^\natural$} (m-2-2);
    \draw (m-2-1) edge[->,dashed] node[above left]{$\overline{F}^\sharp$} (m-1-2);
  \end{tikzpicture}
  \end{minipage}
  \begin{minipage}{0.75\textwidth}
  where:
  \begin{itemize}
      \item $\mathbf{C}$, $\mathbf{D}$, $\mathbf{R}^\natural$, and $\mathbf{R}^\sharp$ are concrete categories over $\mathbf{B}$;
      \item $\gamma_D: \mathbf{D} \to \mathbf{C}$ is a concrete functor that has a weak left adjoint $\langle \alpha, \eta \rangle$;
      \item $F^\natural : \mathbf{C} \rightarrow \mathbf{R}^\natural$, $F^\sharp : \mathbf{D} \rightarrow \mathbf{R}^\sharp$; and $\gamma_R : \mathbf{R}^\sharp \rightarrow \mathbf{R}^\natural$ are concrete functors;
      \item $F^\sharp$ is sound with respect to $F^\natural \circ \gamma_D$.
  \end{itemize}
  \end{minipage}

   Supposing that $\mathbf{R}^\sharp$ is fibred (so that the lifted operation $\overline{F}^\sharp : \mathbf{C} \rightarrow \mathbf{R}^\sharp$ is well-defined) and $\gamma_R$ is cartesian, then $\overline{F}^\sharp$ is sound with respect to $F^\natural$.
\end{theorem*}
\begin{proof}
    For any object $A$ of $\mathbf{C}$, we must show that $F^\natural(A) \preceq \gamma_R(\overline{F}^\sharp(A))$.

Since $F^\natural$ is a concrete functor, the arrow $\eta_A : \phi_C(A) \to \phi_C(\gamma_D(\alpha(A)))$ from the weak left adjoint gives us the relation $\barr{F^\natural(A)}{\eta_A}{F^\natural(\gamma_D(\alpha(A)))}$.
We have as a condition that $F^\sharp$ is sound with respect to $F^\natural \circ \gamma_D$. Applied to the object $\alpha(A) \in Ob(\mathbf{D})$, this gives us $F^\natural(\gamma_D(\alpha(A))) \preceq \gamma_R(F^\sharp(\alpha(A)))$, which is equivalent to the relation $\barr{F^\natural(\gamma_D(\alpha(A)))}{1}{\gamma_R(F^\sharp(\alpha(A)))}$.
Combining these, we have $\barr{F^\natural(A)}{\eta_A}{\barr{F^\natural(\gamma_D(\alpha(A)))} {1}{ \gamma_R(F^\sharp(\alpha(A)))}}$, and so by composition, we get $\barr{F^\natural(A)}{\eta_A}{\gamma_R(F^\sharp(\alpha(A)))}$.

By definition of the lifted operator, we have $\overline{F}^\sharp(A) = \eta_A^{-1}(F^\sharp(\alpha(A)))$, which means there exists a cartesian arrow $\bcarr{\overline{F}^\sharp(A)}{\eta_A}{F^\sharp(\alpha(A))}$. Since $\gamma_R$ is a cartesian functor, it preserves this arrow, giving us a cartesian arrow $\bcarr{\gamma_R(\overline{F}^\sharp(A))}{\eta_A}{\gamma_R(F^\sharp(\alpha(A)))}$.

We have now established two relations over the same base arrow $\eta_A$: (1) $\barr{F^\natural(A)}{\eta_A}{\gamma_R(F^\sharp(\alpha(A)))}$;
(2) A cartesian arrow $\bcarr{\gamma_R(\overline{F}^\sharp(A))}{\eta_A}{\gamma_R(F^\sharp(\alpha(A)))}$.
By the universal property of cartesian arrows (\dag), there must exist an arrow $\barr{F^\natural(A)}{1}{\gamma_R(\overline{F}^\sharp(A))}$. Equivalently, $F^\natural(A) \preceq \gamma_R(\overline{F}^\sharp(A))$.
\end{proof}

\subsection*{Proof of the Extension Theorem}

Recall the extension theorem for lifting algebraic laws on the operators:
\generalizedExtension*

\begin{proof}[Proof of Theorem~\ref{thm:generalized_extension}]

Let $A \in \Ob{C}$.   We must show that
$F_1(\lifted{\functor{G}}(\widetilde{\functor{H}}(A))) \preceq
\lifted{\functor{G}}(A)$.

Let $A_H \defeq \widetilde{\functor{H}}(A)$.
Since $\barr{A}{\eta_A}{\gamma(\alpha(A))}$ and $\widetilde{\functor{H}}$ is robust, we have
\[A_H = \barr{\widetilde{\functor{H}}(A)}{\eta_A}{\widetilde{\functor{H}}(\gamma(\alpha(A)))}=\gamma(H(\alpha(A)))\] (where the second equation follows from the assumption $\widetilde{\functor{H}}\gamma = \gamma H$).  It follows from the universal property of weak left adjoints that there is some $\lifted{\eta_A}$ such that $\barr{\alpha(A_H)}{\lifted{\eta_A}}{H(\alpha(A))}$ and the following diagram commutes:
\[
\begin{tikzpicture}[thick]
 \matrix (m) [matrix of math nodes, row sep=3em, column sep=3em, text height=1.5ex, text depth=0.25ex]
 {
  & \proj{D}(\functor{H}(\alpha(A))) \\
  \proj{C}(A_H) &  \proj{D}(\alpha(A_H))\\
  % & \\
 };
 \draw (m-2-1) edge[->] node[above left] {$\eta_A$} (m-1-2);
 \draw (m-2-1) edge[->] node[below] {$\eta_{A_H}$} (m-2-2);
 \draw (m-2-2) edge[->,dashed] node[right] {$\lifted{\eta_A}$} (m-1-2);
\end{tikzpicture}
\]

The argument for why $F_1(\lifted{\functor{G}}(\widetilde{\functor{H}}(A))) \preceq
\lifted{\functor{G}}(A)$ (or equivalently,
$\barr{F_1(\lifted{\functor{G}}(\widetilde{\functor{H}}(A)))}{1}{\lifted{\functor{G}}(A)}$) is expressed in the following diagram:
\begin{center}
\begin{tikzpicture}[thick]
  \matrix (m) [matrix of math nodes, row sep=2.5em, column sep=4em, text height=1.5ex, text depth=0.25ex]
  {
    \functor{F}_1(\lifted{\functor{G}}(A_H)) & \functor{F}_1(\functor{G}(\alpha(A_H))) & \functor{F}_1(\functor{G}(\functor{H}(\alpha(A)))) & \functor{F}_2(\functor{G}(\alpha(A))) \\
    & & & \functor{F}_2(\lifted{\functor{G}}(A)) \\
  };
  \draw (m-1-1) edge[{||[sep=1pt]}->] node[above]{$\eta_{A_H}$} (m-1-2);
  \draw (m-1-2) edge[{||[sep=1pt]}->] node[above]{$\overline{\eta_A}$} (m-1-3);
  \draw (m-1-3) edge[{||[sep=1pt]}->] node[above]{$1$} (m-1-4);
  \draw (m-2-4) edge[{||[sep=1pt]}->>] node[right]{$\eta_A$} (m-1-4);
  \draw (m-1-1) edge[{||[sep=1pt]}->,dashed,bend right=5] node[below left]{$1$} (m-2-4);
\end{tikzpicture}
\end{center}
The idea is by that by composing the propositions of the top portion of the diagram
we have
$\barr{F_1(\overline{G}(\widetilde{H}(A)))}{\overline{\eta_A}\circ \eta_{A_H}}{F_2(G(\alpha(A)))}$, so by ($\dagger$) the universal property of the cartesian arrow
$\bcarr{F_2(\overline{G}(A))}{\eta_A}{F_2(G(\alpha(A)))}$ and the fact that $\overline{\eta_A}\circ \eta_{A_H} = \eta_A$, we have
the desired result (dashed arrow)
\[F_1(\overline{G}(\widetilde{H}(A))) = \barr{F_1(\overline{G}(A_H))}{1}{F_2(\overline{G}(A))}\]
We argue for each component of this diagram in turn:
\begin{itemize}
\item $\barr{\functor{F}_1(\lifted{\functor{G}}(A_H))}{\eta_{A_H}}{\functor{F}_1(\functor{G}(\alpha(A_H)))}$: By definition of $\lifted{\functor{G}}$ we have
$\lifted{\functor{G}}(A_H) = \eta_{A_H}^{-1}(G(\alpha(A_H))$ and thus $\bcarr{\lifted{\functor{G}}(A_H)}{\eta_{A_H}}{G(\alpha(A_H))}$.  The result follows because $F_1$ is a concrete functor.
\item $\barr{\functor{F}_1(\functor{G}(\alpha(A_H)))}{\overline{\eta_A}}{\functor{F}_1(\functor{G}(H(\alpha(A))))}$:
we have $\barr{\alpha(A_H)}{\lifted{\eta_A}}{H(\alpha(A))}$ by the construction of $\lifted{\eta_A}$.  The result follows because $F_1$ and $G$ are concrete functors.
\item $\barr{\functor{F}_1(\functor{G}(H(\alpha(A))))}{1}{\functor{F}_2(\functor{G}(\alpha(A)))}$: by the assumption
$F_1GH \preceq F_2G$.
\item $\bcarr{\functor{F}_2(\lifted{\functor{G}}(A))}{\eta_A}{\functor{F}_2(\functor{G}(\alpha(A)))}$:  By definition of $\overline{G}$, we have
$\lifted{\functor{G}}(A) = \eta_A^{-1}(G(\alpha(A))$ and thus $\bcarr{\lifted{\functor{G}}(A)}{\eta_A}{G(\alpha(A))}$; the result follows since $F_2$ is a cartesian functor. \qedhere
\end{itemize}
\end{proof}

\subsection*{Proof of Theorem~\ref{thm:admissibility}}

The essential argument is that if two vertices are incomparable in the sense that neither appears in a local cycle of the other (and therefore can appear in either order in an admissible order), then either way the vertices are eliminated yields identical results:
\begin{lemma} \label{lem:commute-elim}
For any reducible flow graph $G$ and vertices $u$ and $v$, if neither $u \in L_G(v)$ nor $v \in L_G(u)$,  then
$G/u,v = G/v,u$.
\end{lemma}
\begin{proof}
Observe that we must have either $w_G(u,v) = 0$ or $w_G(v,u) = 0$.  If both are non-zero then $\tuple{u,v}\tuple{v,u}$ is a cycle, and since $G$ is reducible, either $u$ dominates $v$ and $v \in L_G(u)$ or $v$ dominates $u$ and $u \in L_G(v)$, a contradiction. Without loss of generality, suppose that $w_G(v,u) = 0$.
Observe that for any vertex $z \neq u$, we have
$w_{G/u}(v,z) = w_{G}(v,z) + w_G(v,u)w_G(u,u)^*w_G(u,z) = w_G(v,z)$.
Similarly, for any $z \neq v$, we have
$w_{G/v}(z,u) = w_G(z,u) + w_G(z,v)w_G(v,v)^*w_G(v,u) = w_G(z,u)$.

Let $z$ and $z'$ be vertices.  Suppose that $z'$ is neither $u$ nor $v$---the other cases are similar.  Then we have
\begin{align*}
 w_{G/u/v}(z,z') &= w_{G/u}(z,z') + w_{G/u}(z,v)w_{G/u}(v,v)^*w_{G/u}(v,z')\\
 &= w_{G/u}(z,z') + w_{G/u}(z,v)w_G(v,v)^*w_G(v,z')\\
 &= \left(\begin{array}{ll}&w_G(z,z')\\ + & w_G(z,u)w_G(u,u)^*w_G(u,z')\\ +&
 w_G(z,v)w_G(v,v)^*w_G(v,z')\\
 + & w_G(z,u)w_G(u,u)^*w_G(u,v)w_G(v,v)^*w_G(v,z')
\end{array}\right)
 \end{align*}
By a similar calculation for $w_{G/v/u}(z,z')$, we arrive at the same result and thus $w_{G/u/v}(z,z') = w_{G/v/u}(z,z')$.  Since this holds for all $z$ and $z'$, we have $G/u/v = G/v/u$.
\end{proof}

We may now prove Theorem~\ref{thm:admissibility}. Recall:
\admissibility*
\begin{proof}
By induction on $n$.  The base case is trivial.  For the induction step, suppose that it is true for all admissible sequences of length $n-1$, and show that it is true for length $n$.

Let $i$ be such that $u_i = v_1$.  By repeated application of Lemma~\ref{lem:commute-elim}, we have that
$G/u_1,\dots,u_n = G/v_1,u_1,\dots,u_{i-1},u_{i+1},\dots,u_n$.  By the induction hypothesis, we have
\[ (G/v_1)/v_2,\dots,v_n = (G/v_1)/u_1,\dots,u_{i-1},u_{i+1},\dots,u_n \]
and thus $G/v_1,\dots,v_n = G/u_1,\dots,u_n$.
\end{proof}

\subsection*{Proof of Theorem~\ref{thm:rpa}}

Recall Theorem~\ref{thm:rpa}:
\rpa*

Recall that 
$\sem{H}(h(v))$ is the weight of the edge from $r_H$ to $h(v)$ in $H/u_1,\dots,u_n$ (for some admissible sequence $u_1,\dots,u_n$).  Our high-level strategy to prove Theorem~\ref{thm:rpa} is by induction on the sequence $u_1,\dots,u_n$.  For each $i$,
let $\vec{v}_i$ denote an admissible enumeration of $f^{-1}(u_i)$.
We construct an enumeration of $V_{G}$ by concatenating the sequences $\vec{v_1},\dots,\vec{v_n}$, and show that for each $i$,
$\tuple{h,f}$ is a loop-preserving stuttering simulation from
$G/\vec{v_1},\dots,\vec{v_i}$ to $H/u_1,\dots,u_i$.
From this we may conclude that
$\tuple{h,f}$ is a stuttering simulation from $G/\vec{v_1},\dots,\vec{v_n}$ to $H/u_1,\dots,u_n$, from which Theorem~\ref{thm:rpa} easily follows.

This intuitive argument is not quite correct as stated---such ``lock step'' elimination of vertices does not in fact preserve loop-preserving stuttering simulations.  An example appears in Figure~\ref{fig:ex-multi-flow-graph}.  Letting $h \defeq \set{ A \mapsto X, B \mapsto Y, C \mapsto X}$, we can observe that $\tuple{h,1}$ is loop-preserving stuttering simulation from $G$ to $H$.  However, it fails to be a stuttering simulation from $G/B$ to $H/Y$, because the edge from $A$ to $C$ is not stuttering, nor is it simulated by the self-loop on $X$.  The essential problem is that eliminating $B$ collapses a non-stuttering path $A\xrightarrow{a}B\xrightarrow{b}C$ with a stuttering edge $A\xrightarrow{1}C$.  To side-step this problem, we introduce \textit{multi-flow graphs}, which allow for parallel edges.

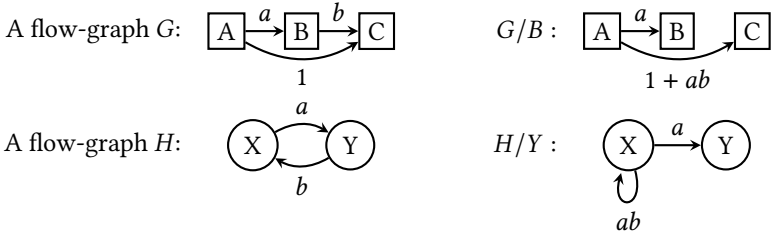
\begin{figure}
\begin{tikzpicture}[>=stealth,thick]
  \node [rectangle,draw](A) {A};
  \node [rectangle,draw,right of=A](B) {B};
  \node [rectangle,draw,right of=B](C) {C};
  \draw (A) edge[->] node[above]{$a$} (B);
  \draw (B) edge[->] node[above]{$b$} (C);
  \draw (A) edge[->,bend right] node[below]{$1$} (C);

  \node [circle,draw,below of=A,node distance=1.5cm,xshift=10pt](X) {X};
  \node [circle,draw,below of=C,node distance=1.5cm,xshift=-10pt](Y) {Y};
  \draw (X) edge[->,bend left] node[above]{$a$} (Y);
  \draw (Y) edge[->,bend left] node[below]{$b$} (X);

  \node [right of=C, node distance=3cm,rectangle,draw] (A') {A};
  \node [rectangle,draw,right of=A'](B') {B};
  \node [rectangle,draw,right of=B'](C') {C};
 \draw (A') edge[->] node[above]{$a$} (B');
  \draw (A') edge[->,bend right] node[below]{$1+ab$} (C');

    \node [circle,draw,below of=A',node distance=1.5cm,xshift=10pt](X') {X};
  \node [circle,draw,below of=C',node distance=1.5cm,xshift=-10pt](Y') {Y};
  \draw (X') edge[->] node[above]{$a$} (Y');
  \draw (X') edge[->,loop below] node[below]{$ab$} (X');

  \node [left of=A,anchor=east,xshift=15pt] (G1) {A flow-graph $G$:};
\node [left of=X,anchor=east,xshift=5pt] (H1) {A flow-graph $H$:};
  \node [left of=A',anchor=east,xshift=15pt] (G2) {$G/B:$};
\node [left of=X',anchor=east,xshift=5pt] (H2) {$H/Y:$};
\end{tikzpicture}
\caption{An example illustrating that stuttering simulations are not preserved by vertex elimination. \label{fig:ex-multi-flow-graph}}
\end{figure}

\begin{definition}
  Let $A$ be a regular algebra.  A \textbf{multi-flow graph} $\mathcal{G}$ over $A$ consists of a finite set of vertices $V_{\mathcal{G}}$, a finite set of weighted directed edges $E_{\mathcal{G}} \subseteq V_{\mathcal{G}} \times A \times V_{\mathcal{G}}$, and a designated root vertex $r_{\mathcal{G}} \in V_{\mathcal{G}}$.
\end{definition}

For a multi-flow graph $\mathcal{G}$ and vertices $u,v \in V_{\mathcal{G}}$, we use $w_{\mathcal{G}}(u,v)$ to denote the sum of the weights on all edges with source $u$ and destination $v$.  For a path $\pi = \tuple{v_0,a_0,v_1}\tuple{v_1,a_1,v_2}\dots\tuple{v_n,a_n,v_{n+1}}$ in a multi-flow graph, we use
$W(\pi)$ to denote the sequential composition of the weights along $\pi$: $W(\pi) \defeq a_0 \cdot a_1 \cdot \dots \cdot a_n$

A flow graph $G$ can be associated with a multi-flow graph $\mathcal{G}$ with the same root and vertices, and where $E_{\mathcal{G}} = \set{ \tuple{u, w_G(u,v), v} : \tuple{u,v} \in G}$.  In the reverse direction, a multi-flow graph $\mathcal{G}$ can be associated with a flow graph $G$ with the same root and vertices, and where $E_G = \set{ \tuple{u,v} : \exists a. \tuple{u,a,v} \in E_{\mathcal{G}}}$
and $w_G(u,v) = w_{\mathcal{G}}(u,v)$.  Thus one may think of a multi-flow graph as a particular representation of a flow graph.  We extend definitions from flow graphs to multi-flow graphs (paths, dominance, reducibility, vertex elimination, stuttering simulations) in the natural way.  In particular, for a vertex $u$ is in a multi-flow graph $\mathcal{H}$, $\mathcal{G}/u$ is the multi-flow graph with the same vertices and root as $\mathcal{G}$, but with
\[ E_{\mathcal{G}/u} = \begin{array}{l@{}l}&\set{ \tuple{z,a,z'} \in E_{\mathcal{G}} : z \neq u }\\ \cup & \set{ \tuple{p,a \cdot w_{\mathcal{G}}(u,u)^* \cdot b,s} : \tuple{p,a,u},\tuple{u,b,s} \in E_{\mathcal{G}}}\\
\cup & \set{ \tuple{p,a \cdot w_{\mathcal{G}}(u,u)^*, u} : \tuple{p,a,u} \in E_{\mathcal{G}} }
\end{array}\]
Observe multi-flow graphs avoid the issue pointed out in Figure~\ref{fig:ex-multi-flow-graph}: if $G$ is considered as a multi-flow graph, then $G/B$ has two parallel edges $A \xrightarrow{1}C$ and $A \xrightarrow{ab} C$; the first edge is stuttering, while the second is simulated by the self-loop on $X$ in $H/Y$.

To prove Theorem~\ref{thm:rpa}, we will need several technical results that relate a multi-flow graph $\mathcal{G}$ to the multi-flow graph $\mathcal{G}/v$ that results from eliminating a vertex.  The following definition relates the paths of these two graphs:
\begin{definition}
Let $\mathcal{G}$ be a multi-flow graph, and let $v \in V_{\mathcal{G}}$.  For paths
$\pi$ in $\mathcal{G}$ and $\pi' \in \mathcal{G}/v$, we say that $\pi'$ is a \textbf{$v$-contraction} of $\pi$ if
\begin{enumerate}
\item $\pi = \pi'$, or
\item $\pi = \tuple{p,a,v}\tuple{v,b_1,v}\dots\tuple{v,b_n,v}\tuple{v,c,s}$ and $\pi' = \tuple{p,a\cdot w_{\mathcal{H}}(v,v)^* \cdot c,s}$, or
\item We can decompose $\pi$ as $\pi_1\pi_2$, and $\pi'$ as $\pi_1'\pi_2'$ such that $\pi_1'$ is a $v$-contraction of $\pi_1$ and $\pi_2'$ is a $v$-contraction of $\pi_2$.
\end{enumerate}
We say that $\pi'$ is a \textbf{simple $v$-contraction} of $\pi$ if it contains no self-loops of $v$.  
\end{definition}
Observe that every path in $\mathcal{G}/v$ is a $v$-contraction of at least one path in $\mathcal{G}$, and every path in $\mathcal{G}$ has a $v$-contraction in $\mathcal{G}/v$.

The following lemmas formalize the sense in which eliminating a vertex preserves various features of multi-flow graph:
dominance (Lemma~\ref{lem:dominance-preservation}), reducibility
(Lemma~\ref{lem:reducible-preservation}), loop structure (Lemma~\ref{lem:loop-invariance}).

\begin{lemma}[Dominance preservation] \label{lem:dominance-preservation}
  Let $\mathcal{G}$ be a multi-flow graph and $v \in V_{\mathcal{G}}$ be a vertex.  For any vertices $u,u' \in V_{\mathcal{G}}$, if $u$ dominates $u'$ in $\mathcal{G}/v$, then $u$ dominates $u'$ in $\mathcal{G}$.  If $u$ dominates $u'$ in $\mathcal{G}$ and $u \neq v$, then $u$ dominates $u'$ in $\mathcal{G}/v$.
\end{lemma}
\begin{proof}
  First, suppose that $u$ dominates $u'$ in $\mathcal{G}/v$.   For a contradiction, suppose that $u$ does not dominate $u'$ in $\mathcal{G}$.  Then there is a path $\pi$ from $r_{\mathcal{G}}$ to $u'$ in $\mathcal{G}$ that does not pass through $u$.  But then the $v$-contraction of $\pi$ is a path from $r_{\mathcal{G}}$ to $u'$ in $\mathcal{G}/u$ that does not pass through $u$, contradicting the assumption that $u$ dominates $u'$.
  
  Now suppose that $u$ dominates $u'$ in $\mathcal{G}$ and $u \neq v$.  For a contradiction, suppose that $u$ does not dominate $u'$ in $\mathcal{G}/v$.  Then there is a path $\pi$ from $r_{\mathcal{G}}$ to $u'$ in $\mathcal{G}/v$ that does not pass through $u$. Let $\tau$ be a path in $\mathcal{G}$ such that $\pi$ is a $v$-contraction of $\tau$.  Then $\tau$ is a path from $r_{\mathcal{G}}$ to $u'$ in $\mathcal{G}$ that does not pass through $u$, contradicting the assumption that $u$ dominates $u'$ in $\mathcal{G}$.
\end{proof}

\begin{lemma}
  \label{lem:reducible-preservation}
  Let $\mathcal{G}$ be a reducible multi-flow graph and $v \in V_{\mathcal{G}}$ be a vertex with $|L_{\mathcal{G}}(v)| \leq 1$. Then $\mathcal{G}/v$ is reducible.
\end{lemma}
\begin{proof}
  Let $\pi$ be a cycle of $\mathcal{G}/v$; we must show that some vertex in $\pi$ dominates all others.  Since $\pi$ is a cycle of $\mathcal{G}/v$, it 
  is the contraction of some cycle $\tau$ of $\mathcal{G}$.  Since $\mathcal{G}$ is reducible, there is some vertex $\ell$ along $\tau$ that dominates all others.  If $\ell = v$, then either $v$ is the \textit{only} vertex that appears in $\tau$ (contradicting the fact $\pi$ is a contraction of $\tau$) or $|L_{\mathcal{G}}(v)| \geq 2$ (contradicting the assumption $|L_{\mathcal{G}}(v)| \leq 1$).  Thus $\ell$ must not be $v$, and so $\ell$ also appears in $\pi$, and by Lemma~\ref{lem:reducible-preservation}, $\ell$ must dominate all other vertices on $\pi$.
\end{proof}

\begin{lemma} \label{lem:loop-invariance}
  Let $\mathcal{G}$ be a reducible multi-flow graph, and let $\ell$ and $v$ be distinct vertices.  Then $L_{\mathcal{G}/v}(\ell) = L_{\mathcal{G}}(\ell) \setminus \set{v}$.
\end{lemma}
\begin{proof}
We show that each side of the equation is included in the other.
\begin{itemize}
    \item $L_{\mathcal{G}}(\ell) \setminus \set{v} \subseteq L_{\mathcal{G}/v}(\ell)$:  Suppose that $u \in L_{\mathcal{G}}(\ell)$ and $u \neq v$.  We must show that $u \in L_{\mathcal{G}/v}(\ell)$.

    Since $u \in L_{\mathcal{G}}(\ell)$, there is a local cycle $\pi$ of $\ell$ that includes $u$.  Let $\tau$ be the $v$-contraction of $\pi$.  Then we have that $\tau$ is a cycle of $\ell$ in $\mathcal{G}/v$ that includes $u$ (since $u \neq v$).  By Lemma~\ref{lem:dominance-preservation}, $\tau$ is local, and so $u \in L_{\mathcal{G}/v}(\ell)$.
    \item $L_{\mathcal{G}/v}(\ell) \subseteq L_{\mathcal{G}}(\ell) \setminus \set{v}$:
    Suppose that $u \in L_{\mathcal{G}/v}(\ell)$.
    Then there is a local cycle $\pi$ of $\ell$ in $\mathcal{G}/v$ that includes $u$. Let $\tau$ be a path in $\mathcal{G}$ such that $\pi$ is a $v$-contraction of $\tau$.  
    Since $\tau$ is a cycle  and $\mathcal{G}$ is reducible, there is some vertex along $\tau$ that dominates all others.  That vertex cannot be $v$, because then we would have $\set{v,\ell} \subseteq L_{\mathcal{G}}(v)$ and $|L_{\mathcal{G}}(v)| \leq 1$ and $\ell \neq v$ by assumption.  It follows that $\ell$ dominates all vertices along $\tau$, and thus $\tau$ is a local cycle of $\ell$ in $\mathcal{G}$ and $u \in L_{\mathcal{G}}(\ell)$. \qedhere
 \end{itemize}
\end{proof}

Recall that our argument for Theorem~\ref{thm:rpa} is based on showing that ``lock step'' elimination of vertices preserves loop-preserving stuttering simulations (formalized in Lemma~\ref{lem:elimination-step}).  Intuitively, we would like the argument to proceed by induction on the admissible enumeration $h^{-1}(u)$---i.e., given an admissible enumeration $v_1,\dots,v_n$ of $h^{-1}(u)$, 
we would like to show that
$\tuple{h,f}$ is a loop-preserving stuttering simulation
from $\mathcal{G}/v_1,\dots,v_i$ to $\mathcal{H}/u$ for all $i$.  However, if $n > 1$ this is not necessarily the case.
Firstly, the vertices $v_{i+1},\dots,v_{n}$ may have outgoing edges in $\mathcal{G}/v_1,\dots,v_i$,
$u$ has no outgoing edges in $\mathcal{H}/u$ with which to simulate them.  Secondly, eliminating a vertex $v_i$ may violate the \textit{Consistent unrolling} property of loop-preserving stuttering simulations, since it shortens the $u$-length of the simulating path of any primitive local cycle of $v_j$ (with $j>i$) that passes through $i$.

To solve the first problem, our argument will establish that $\tuple{h,f}$ is a stuttering simulation from $\mathcal{G}/v_1,\dotsi,v_i$ to
an intermediate object
$\mathcal{K} = \mathcal{H} + \mathcal{H}/u$ (where $\mathcal{H}+\mathcal{H}/u$ denotes the graph with all the edges of both $\mathcal{H}$ and $\mathcal{H}/u$).  The crucial lemma concerning $\mathcal{K}$ is the following,
which establishes that paths through $\mathcal{K}$ that pass through $u$ are simulated by paths that do not.
\begin{lemma} \label{lem:transit}
  Let $\mathcal{H}$ be a multi-flow graph and let $u \in V_{\mathcal{H}}$, and let $\mathcal{K} = \mathcal{H}+(\mathcal{H}/u)$.  Then we have:
  \begin{enumerate}
      \item For any $\tuple{p,a,u} \in E_{\mathcal{K}}$, we have $\tuple{p,a \cdot w_{\mathcal{H}}(u,u)^*,u} \in E_{\mathcal{K}}$.
      \item For any $\tuple{p,a,u}$ and $\tuple{u,b,s}$ in $E_{\mathcal{K}}$, we have
      $\tuple{p,a \cdot w_{\mathcal{H}}(u,u)^* \cdot v, s} \in E_{\mathcal{K}}$.
  \end{enumerate}
\end{lemma}
\begin{proof}
We prove (1); case (2) is similar.
If $\tuple{p,a,u} \in E_{\mathcal{H}}$, then (1) holds simply by construction of $\mathcal{K}$.
If $\tuple{p,a,u} \notin E_{\mathcal{H}}$, then by construction it takes the form
$\tuple{p,b \cdot \textit{loop}_{\mathcal{H}}(u),u}$ for some $\tuple{p,b,u} \in E_{\mathcal{H}}$,
where $\textit{loop}_{\mathcal{H}}(u)$ is $ w_{\mathcal{H}}(u,u)^*$.
By transitivity,
$\textit{loop}_{\mathcal{H}}(u)\cdot\textit{loop}_{\mathcal{H}}(u) = \textit{loop}_{\mathcal{H}}(u)$, and so 
$\tuple{p,a \cdot \textit{loop}_{\mathcal{H}}(u), u} = \tuple{p,a,u} \in E_{\mathcal{K}}$.
\end{proof}

To solve the second problem, we develop a more relaxed variant of the \textit{Consistent unrolling} property.
In the following, for any multi-flow graph $\mathcal{G}$ and any vertex $\ell \in V_{\mathcal{G}}$ with $L_{\mathcal{G}}(\ell) \neq \emptyset$, we use $\mathcal{G}[\ell]$ to denote the induced sub-multi flow graph of $\mathcal{G}$ with vertices $L_{\mathcal{G}}(\ell)$, and root $\ell$.   A path $\pi \in \mathcal{G}[\ell]$ is said to be \textbf{forward} if it does not pass through $\ell$; i.e., $\pi$ cannot be decomposed into two non-trivial paths $\pi = \pi_1\pi_2$ such that $\dst(\pi_1) = \src(\pi_2) = \ell$.  Observe that a local cycle of $\ell$ is primitive exactly when it is forward.
The following lemma gives our relaxation of the \textit{Consistent unrolling} property---namely, that there exists a distance function for every loop.
\begin{lemma} \label{lem:distance-function}
  Let $\mathcal{G}$ and $\mathcal{H}$ be multi-flow graphs and let $\tuple{h,f}$ be a loop-preserving stuttering simulation from $\mathcal{G}$ to $\mathcal{H}$.  Then for any vertex $\ell$ with $|L_{\mathcal{H}}(h(\ell))| \leq 1$, there is a ``distance function'' $d_\ell : L_{\mathcal{G}}(\ell) \times L_{\mathcal{G}}(\ell) \rightarrow \mathbb{N}$ such that for every forward path $z \xrightarrow{\pi} z'$ in $\mathcal{G}[\ell]$, we have $\barr{W(\pi)}{f}{w_{\mathcal{H}}(h(\ell),h(\ell))^{d_\ell(z,z')}}$.
\end{lemma}
\begin{proof}
Let $u = h(\ell)$.  First we show that $L_{\mathcal{G}}(\ell) \subseteq h^{-1}(u)$.   Suppose $v \in L_{\mathcal{G}}(\ell)$.  Since $\tuple{h,f}$ preserves loops, we have $h(v) \in L_{\mathcal{H}}(h(\ell))$ by \textit{(Loop membership)}.  Since $|L_{\mathcal{H}}(h(\ell))| \leq 1$, we must have $L_{\mathcal{H}}(h(\ell)) = \set{ h(\ell) } = \set { u }$, and since $h(v) \in L_{\mathcal{H}}(h(\ell))$ we must have $h(v) = u$, and thus $v \in h^{-1}(u)$.

Next, we define a ``distance function'' $d_\ell : L_{\mathcal{G}}(\ell) \times L_{\mathcal{G}}(\ell) \rightarrow \mathbb{N}$ such that for every $z,z' \in L_{\mathcal{G}}(\ell)$, 
every forward path $z \xrightarrow{\pi} z'$ in $\mathcal{G}[\ell]$ satisfies $\barr{W(\pi)}{f}{w_{\mathcal{H}}(u,u)^{d_\ell(z,z')}}$.   If
there is no non-trivial forward path from $z$ to $z'$ in $\mathcal{G}[\ell]$, we may simply take $d_\ell(z,z') = 0$.  
If at least one such path $\pi$ exists, then we may take $d_\ell(z,z') = |r(\pi)|$.  To show that $d_\ell$ is well-defined, we must show that
(1) $\barr{W(\pi)}{f}{w_{\mathcal{H}}(u,u)^{|r(\pi)|}}$, and
(2) if there is any other non-trivial forward path $z \xrightarrow{\pi'} z'$ in $\mathcal{G}[\ell]$, we have $|r(\pi)| = |r(\pi')|$
\begin{enumerate}
\item Let $\pi = e_1 \dots e_m$.  Since every vertex along $\pi$ belongs to $L_{\mathcal{G}}(\ell) \subseteq h^{-1}(u)$,
for each edge $e_i$, we must have either $r(e_i) = \epsilon$, or $r(e_i)$ is a self-loop of $u$, and thus 
$\barr{w(e_i)}{f}{w_{\mathcal{H}}(u,u)}$ by compatibility and idempotence of addition.  It follows that
$\barr{w(\pi)}{f}{w_{\mathcal{H}}(u,u)^{|r(\pi)|}}$.
\item Let $\pi'$ be a forward path from $z$ to $z'$ in $\mathcal{G}[\ell]$.
Since $z$ belongs to $L_{\mathcal{G}}(\ell)$, there is a forward path $\tau$ in
$\mathcal{G}[\ell]$ from $\ell$ to $z$; 
since $z'$ belongs to $L_{\mathcal{G}}(\ell)$, there is a forward path $\tau'$ in
$\mathcal{G}[\ell]$ from $z'$ to $\ell$ (if $z=\ell$, choose $\tau=\epsilon$, and if $z' = \ell$, choose $\tau' = \epsilon$).
Then $\tau\pi\tau'$ and $\tau\pi'\tau'$ are primitive local cycles of $\ell$, and so
$|r(\tau\pi\tau')|_u = |r(\tau\pi'\tau')|_u$ the assumption that $\tuple{h,f}$ is loop-preserving.
It follows that $|r(\pi)|_u = |r(\pi')|_u$.  Since all edges in $r(\pi)$ and $r(\pi')$ have source $u$, their $u$-lengths coincide with their lengths, and so $|r(\pi)| = |r(\pi')|$. \qedhere
\end{enumerate}
\end{proof}

The following lemma is used to establish that existence of distance functions is preserved by vertex elimination (in contrast to the Consistent unrolling property).
\begin{lemma} \label{lem:forward-path}
Let $\mathcal{G}$ be a multi-flow graph, and let $\ell,v \in V_{\mathcal{G}}$ such that $\ell \notin L_{\mathcal{G}}(v)$.  For any 
path $\tau$ in $\mathcal{G}$ and forward path $\pi$ in $\mathcal{G}/v[\ell]$ such that $\pi$ is a $v$-contraction of $\tau$, we have that $\tau$ is a forward path in $\mathcal{G}[\ell]$
\end{lemma}
\begin{proof}
Since $\pi$ is a forward path in $\mathcal{G}/v[\ell]$, we must have $v \neq \ell$ because $v$ has out-degree zero in $\mathcal{G}/v$ and thus $L_{\mathcal{G}}(v)$ is empty while $L_{\mathcal{G}}(\ell)$ is non-empty.  Since $\pi$ does not pass through $\ell$, neither does $\tau$ because $\pi$ is a $v$-contraction of $\tau$ and $v \neq \ell$.
It remains to show that $\tau$ belongs to $\mathcal{G}[\ell]$---i.e., $\ell$ dominates
all vertices that appear in $\tau$.  Since $\ell$ dominates all vertices of $\pi$ in $\mathcal{G}/v$ by assumption and $\pi$ is a $v$-contraction of $\tau$, then by Lemma~\ref{lem:dominance-preservation}, it is sufficient to show that should $v$ appear in $\tau$, then $\ell$ dominates $v$ in $\mathcal{G}$.  Suppose that $v$ does appear in $\tau$.  Since $z$ and $z'$ belong to $L_{\mathcal{G}/v}(\ell)$ and therefore 
$L_{\mathcal{G}}(\ell)$, there exist paths 
$\ell \xrightarrow{\rho} z$ and $z' \xrightarrow{\rho'} \ell$ in $\mathcal{G}[\ell]$.  Since $\mathcal{G}$ is reducible, some vertex must dominate all others in the cycle $\rho\pi\rho'$.  Since $\ell \notin L_{\mathcal{G}}(v)$ by assumption, this vertex cannot be $v$ and so must be $\ell$, and thus $v$ appears in the local cycle $\rho\pi\rho'$ of $\ell$ and thus $\pi$ is a forward path in $\mathcal{G}[\ell]$.
\end{proof}

We may now prove the critical lemma in the argument for Theorem~\ref{thm:rpa}:

\begin{lemma}[Lock-step elimination]\label{lem:elimination-step}
  Let $\mathcal{G}$ and $\mathcal{H}$ be multi-flow graphs,
  and let $\tuple{h,f}$ be a loop-preserving stuttering simulation from $\mathcal{G}$ to $\mathcal{H}$.
  Then for any vertex $u \in V_H$ with $|L_{\mathcal{H}}(u)| \leq 1$ and any admissible enumeration $v_1,\dots,v_n$ of $h^{-1}(u)$, we have $\tuple{h,f} \in \mathbf{G}_{\mathbf{A}}(G/v_1,\dots,v_n,H/u)$.
\end{lemma}
\begin{proof}

Let $r$ be a witness function for the stuttering simulation $\tuple{h,f}$, let $\mathcal{K} = \mathcal{H}+(\mathcal{H}/u)$, and for any vertex $\ell \in V_{\mathcal{G}}$, let $d_\ell$ be as in Lemma~\ref{lem:distance-function}.

  For any $i \in \set{0,\dots,n}$, define
  $\mathcal{G}_i \defeq \mathcal{G}/v_1,\dots,v_i$ to be the graph $\mathcal{G}$ after $i$ elimination steps.
  For each $i$, and each vertex $\ell \in V_{\mathcal{G}}$, we will define a function $r_i : E_{\mathcal{G}_i} \rightarrow E_{\mathcal{K}}$ such that the following hold
  \begin{enumerate}
  \item[$A(i)$]: $\tuple{h,f}$ is a stuttering simulation from $\mathcal{G}_i$ to $\mathcal{K}$ with witness $r_i$
  \item[$B(i)$]: For any vertex $\ell \in h^{-1}(u)$ and any pair of vertices $z,z' \in L_{\mathcal{G}_i}(\ell)$, every non-trivial forward path $z \xrightarrow{\pi} z' \in \mathcal{G}[\ell]$ satisfies
  $\barr{w(\pi)}{f}{w_{\mathcal{H}}(u,u)^{d_\ell(z,z')}}$
  \end{enumerate}
  From property $A(n)$, we have that $\tuple{h,f}$ is a stuttering simulation from $\mathcal{G}_n = \mathcal{G}/v_1,\dots,v_n$ to $\mathcal{K}$.  Since every vertex in $f^{-1}(u)$ has out-degree zero in $\mathcal{G}_n$ (i.e., the outgoing edges of $u$ in $\mathcal{K}$ do not simulate any edges of $\mathcal{G}_n$), $\tuple{h,f}$ is also a stuttering simulation from $\mathcal{G}_n$ to $H/u$.  
    
  We now show that $A(i)$ and $B(i)$ hold for all $i$ by induction on $i$.
  For the base case $\mathcal{G}_0 = \mathcal{G}$, define $r_0 = r$; $A(0)$ holds by assumption, $B(0)$ holds by the definition of $d_\ell$.
  
  We now show the induction step.  Suppose that condition $A$ and $B$ hold up to $i-1$.
\begin{enumerate}
\item[$A(i)$]
  Observe that by $B(i-1)$, 
  for any self-loop $\tuple{v_i,a,v_i} \in \mathcal{G}_{i-1}$, we have
  $\barr{a}{f}{w_{\mathcal{H}}(u,u)^{d_{v_i}(v_i,v_i)}}$, and thus
\begin{align*}
\barr{w_{\mathcal{G}_{i-1}}(v_i,v_i)&}{f}{w_{\mathcal{H}}(u,u)^{d_{v_i}(v_i,v_i)}} & \text{Compatibility}\\
\barr{w_{\mathcal{G}_{i-1}}(v_i,v_i)^*&}{f}{(w_{\mathcal{H}}(u,u)^{d_{v_i}(v_i,v_i)})^*} & \text{Compatibility}\\
&\leq w_{\mathcal{H}}(u,u)^* & \text{Unrolling}\\
\barr{w_{\mathcal{G}_{i-1}}(v_i,v_i)^*&}{f}{w_{\mathcal{H}}(u,u)^*} & \text{Compatibility}
\end{align*}

  We construct $r_i$ as follows.  For the edges
  $e \in E_{\mathcal{G}_i} \cap E_{\mathcal{G}_{i-1}}$, we may define
  $r_i(e) = r_{i-1}(e)$.  The edges that are present in $\mathcal{G}_{i}$ that are not present in $\mathcal{G}_{i-1}$ take one of two forms
  \begin{itemize}
  \item $e = \tuple{p,a \cdot w_{\mathcal{G}_{i-1}}(v_i,v_i)^*, v_i}$ where $\tuple{p,a,v_i} \in E_{\mathcal{G}_{i-1}}$:
  take \[r_i(e) \defeq \tuple{h(p), W(r_{i-1}(p,a,v_i))\cdot w_{\mathcal{H}}(u,u)^*, h(v_i)} \]
  which belongs to $E_{\mathcal{K}}$ by Lemma~\ref{lem:transit}.
  Observe that 
 by $A(i-1)$, $r_{i-1}$ is a witness that $\tuple{h,f}$ is a stuttering simulation from $\mathcal{G}_{i-1}$ to $\mathcal{K}$, and so $\barr{a}{f}{W(r_{i-1}(p,a,v_i))}$.  Since
$\barr{w_{\mathcal{G}_{i-1}}(v_i,v_i)^*}{f}{w_{\mathcal{H}}(u,u)^*}$ we have
 \[ W(e) = \barr{a \cdot w_{\mathcal{G}_{i-1}}(v_i,v_i)^*}{f}{r_{i-1}(p,a,v_i)\cdot w_{\mathcal{H}}(u,u)^*} = W(r_i(e)) \]
  by compatibility.
    \item $e = \tuple{p,a \cdot w_{\mathcal{G}_{i-1}}(v_i,v_i)^* \cdot b, s}$ where $\tuple{p,a,v_i},\tuple{v_i,b,s} \in E_{\mathcal{G}_{i-1}}$: take \[ r_i(e) \defeq \tuple{h(p), w(r_{i-1}(p,a,v_i)) \cdot w_{\mathcal{H}}(u,u)^* \cdot w(r_{i-1}(v_i,b,s)), h(s)} \]
    which belongs to $E_{\mathcal{K}}$ by Lemma~\ref{lem:transit}.
    Observe that by $A(i-1)$, $r_{i-1}$ is a stuttering simulation from
$\mathcal{G}_{i-1}$ to $\mathcal{K}$, and so $\barr{a}{f}{W(r_{i-1}(p,a,v_i))}$
and $\barr{b}{f}{W(r_{i-1}(v_i,b,s))}$.
    Since
$\barr{w_{\mathcal{G}_{i-1}}(v_i,v_i)^*}{f}{w_{\mathcal{H}}(u,u)^*}$ we have
\[ W(e) = \barr{a \cdot w_{\mathcal{G}_{i-1}}(v_i,v_i)^* \cdot b}{f}{W(r_{i-1}(p,a,v_i)) \cdot w_{\mathcal{H}}(u,u)^* \cdot W(r_{i-1}(v_i,b,s))} = W(r_i(e)) \]
  by compatibility.
  \end{itemize}
  Since for any edge $e$ in $\mathcal{G}_i$ we have $\barr{W(e)}{f}{W(r_i(e))}$, $\tuple{h,f}$ is a stuttering simulation with witness $r_i$.
\item[$B(i)$]
Let $\ell \in h^{-1}(u)$, and let $z \xrightarrow{\pi} z'$ be a non-trivial forward path in $\mathcal{G}_{i}[\ell]$.
Let $\tau$ be a path in $\mathcal{G}_{i-1}$ such that $\pi$ is a $v_i$-contraction of $\tau$.  By Lemma~\ref{lem:forward-path}, $\tau$ is a forward path of $\mathcal{G}_{i-1}[\ell]$, and so by $B(i-1)$ we have $\barr{W(\tau)}{f}{w_{\mathcal{H}}(u,u)^{d_\ell(z,z')}}$.  Thus it suffices to show that $W(\pi) = W(\tau)$.
Consider two cases:
\begin{itemize}
\item $v_i \in L_{\mathcal{G}_{i-1}}(\ell)$:
Since $\pi$ is a $v_i$-contraction of $\tau$, to show $W(\pi) = W(\tau)$ it is sufficient to show that $v_i$ has no self-loops in $\mathcal{G}_{i-1}$.  For a contradiction, suppose that $v_i$ has a self-loop, and thus $v_i \in \mathcal{G}_{i-1}(v_i)$.  By 
Lemma~\ref{lem:loop-invariance}, we have
$v_i \in \mathcal{G}_{i-1}(v_i) \subseteq \mathcal{G}(v_i)$ and 
$v_i \in \mathcal{G}_{i-1}(\ell) \subseteq \mathcal{G}(\ell)$, contradicting the assumption that $\tuple{h,f}$ preserves loops from $\mathcal{G}$ to $\mathcal{H}$ (Non-nesting).
\item $v_i \notin L_{\mathcal{G}_{i-1}}(\ell)$:
By Lemma~\ref{lem:forward-path}, $\tau$ is a forward path in $\mathcal{G}_{i-1}[\ell]$.  Since $v_i \notin L_{\mathcal{G}_{i-1}}(\ell)$, $\tau$ may not pass through $v_i$, and thus must be $\pi$, whence $W(\pi) = W(\tau)$.
\end{itemize}

  \end{enumerate}

Finally, we must show that $\tuple{h,f}$ is loop preserving from $\mathcal{G}_n$ to $\mathcal{H}/u$.
\begin{enumerate}
\item (Loop membership) For any $z,z' \in V_{\mathcal{G}_n}$ such that $z \in L_{\mathcal{G}_n}(z')$, we must have
$z \in L_{\mathcal{G}}(z')$, and thus $h(z) \in L_{\mathcal{H}}(z')$.  We can have neither $h(z) = u$ nor $h(z') = u$, because all such vertices have out-degree zero in $\mathcal{G}_n$ and thus do not belong to any loop.  It follows that $h(z) \in L_{\mathcal{H}/u}(h(z'))$, since 
$L_{\mathcal{H}/u}(\ell) = L_{\mathcal{H}}(\ell) \setminus \set{u}$ for any vertex $\ell$ by Lemma~\ref{lem:loop-invariance}.

\item (Non-nesting) Suppose $z,z' \in V_{\mathcal{G}_n}$ such that $h(z) = h(z')$ and $z \in L_{\mathcal{G}_n}(z')$; we must show that $L_{\mathcal{G}_n}(z) = \emptyset$.
Since $z \in L_{\mathcal{G}_n}(z') \subseteq L_{\mathcal{G}}(z)$ and $\tuple{h,f}$ is loop-preserving from $\mathcal{G}$ to $\mathcal{H}$, 
we must have $L_{\mathcal{G}}(z) = \emptyset$.
Since $L_{\mathcal{G}_n}(z) \subseteq L_{\mathcal{G}}(z) = \emptyset$, we have
$L_{\mathcal{G}_n}(z) = \emptyset$.

\item (Consistent unrolling)
  Let $z \in V_{\mathcal{G}_n}$.  
  By the assumption that $r$ is a witness that $\tuple{h,f}$ is loop-preserving from $\mathcal{G}$ to $\mathcal{H}$, there is some $k$ such that every primitive local cycle $\pi$ of $z$ in $\mathcal{G}$ satisfies $|r(\pi)|_{h(z)} = k$.
We will show that every primitive local cycle $\pi$ of $z$ in $\mathcal{G}_n$ also satisfies $|r_n(\pi)|_{h(z)} = k$.

We make use of the following lemma:
\begin{lemma} \label{lem:contraction-preserves-length}
For any $i < n$,
any path $\tau$ of $\mathcal{G}_i$ and any $v_{i+1}$-contraction $\tau'$ of $\tau$, we have 
$|r_i(\tau)|_{h(z)} = |r_{i+1}(\tau')|_{h(z)}$.
\end{lemma}

Using this lemma, we have that
$|r_n(\pi)|_{h(z)} = k$ by the following argument.
Let $\pi_0,\pi_1,\dots,\pi_n$ be a sequence of cycles of $z$ such that for each $i<n$, $\pi_{i+1}$ is a $v_{i+1}$-contraction of $\pi_i$ and such that $\pi_n = \pi$.
By Lemma~\ref{lem:forward-path},
$\pi_0$ is a forward path in
$\mathcal{G}_0[z]$, and thus a primitive local cycle of $z$ in $\mathcal{G}_0 = \mathcal{G}$ and thus $|r_0(\pi_0)|_{h(z)} = k$ by assumption.  By Lemma~\ref{lem:contraction-preserves-length}, we have
\[ |r_n(\pi)|_{h(z)} = |r_n(\pi_n)|_{h(z)} = |r_{n-1}(\pi_{n-1})|_{h(z)} = \dots = |r_0(\pi_0)|_{h(z)} = k\ .\]
Thus it is sufficient to prove that
$|r_n(\pi)|_{h(z)} = |r_0(\pi_0)|_{h(z)}$.

It remains to prove Lemma~\ref{lem:contraction-preserves-length}. We do so by induction on the judgement that $\tau'$ is a $v_{i+1}$-contraction of $\tau$.
\begin{itemize}
\item Case $\tau = \tau'$: trivial, since $r_i$ and $r_{i+1}$ agree on edges that belong to both $\mathcal{G}_i$ and $\mathcal{G}_{i+1}$
\item Case
\begin{align*}
    \tau &= \tuple{p,a,v_{i+1}}\tuple{v_{i+1},b_1,v_{i+1}}\dots \tuple{v_{i+1},b_r,v_{i+1}}\tuple{v_{i+1},v,s} \\
    \tau' &= \tuple{p , a \cdot w_{\mathcal{G}_{i+1}}(v_{i+1}, v_{i+1})^* \cdot c,s}
\end{align*}

Since within $\mathcal{G}_n$, $z$ has a local cycle and all vertices in $f^{-1}(u)$ have out-degree zero, we must have $h(z) \neq u$.
Observe that
$|r_i(\tau)|_{h(z)} = |r_i(\tuple{p,a,v_{i+1}})|_{h(z)}$
since $h(v_{i+1}) = u \neq h(z)$, and $\tuple{p,a,v_{i+1}}$ is the only edge of $\tau$ with source not equal to $v_{i+1}$.
If $h(p) \neq h(z)$, then clearly
$|r_i(\tau)|_{h(z)} = 0 = |r_{i+1}(\tau')|_{h(z)}$.  If $h(p) = h(z)$, then $|r_{i+1}(\tau')|_{h(z)} = 1$, because by construction $r_{i+1}$ does not map any nascent edge in $\mathcal{G}_{i+1}$ to $\epsilon$.  We must also have 
$|r_{i}(\tau)|_{h(z)} = 1$, because $r_i$ cannot map $\tuple{p,a,v_{i+1}}$ to $\epsilon$ since $h(p) \neq h(v_{i+1})$.  Thus in any case
$|r_i(\tau)|_{h(z)} = |r_n(\tau')|_{h(z)}$

\item Case $\tau = \tau_1\tau_2$, $\tau' = \tau'_1\tau'_2$, with $\tau_1'$ and $\tau_2'$ $v_{i+1}$-contractions of $\tau_1$ and $\tau_2$ respectively: by the induction hypothesis and the fact that $r_i$ and $|-|_{h(z)}$ are homomorphisms, we have
\[ |r_i(\tau)|_{h(z)} = |r_i(\tau_1)|_{h(z)} + |r_i(\tau_2)|_{h(z)} = |r_{i+1}(\tau'_1)|_{h(z)} + |r_{i+1}(\tau'_2)|_{h(z)} = |r_{i+1}(\tau')|_{h(z)} \qedhere\]
\end{itemize}
\end{enumerate}
  \end{proof}

Finally, we may prove the main result.  Recall Theorem~\ref{thm:rpa}:
\rpa*
\begin{proof}
  Suppose $u_1,\dots,u_n$ is an admissible enumeration of $V_H\setminus \set{r_H}$.  For each $i$, let $v_{i,1},\dots,v_{i,m_i}$ be an admissible enumeration of $f^{-1}(u_i)$.
  Define \begin{align*}
      \vec{v} &\defeq v_{1,1},\dots,v_{1,m_1},\dots,v_{n,1},\dots,v_{n,m_n}\\
      G' &\defeq G/\vec{v}\\
      H' &\defeq H/u_1,\dots,u_n
  \end{align*}
  By Lemma~\ref{lem:elimination-step}, 
  $\tuple{h,f}$ is a stuttering simulation from $G'$ to $H'$. 
  It follows that for every non-root vertex $v$, we have
  $\barr{w_{G'}(r_G,v)}{f}{\barr{H'}{f(r_G)}{h(v)}}$.
  Assuming that $\vec{v}$ is admissible, we have
  $w_{G'}(r_H,v) = \sem{G}(v)$ and thus
  \[\sem{G}(v) = \barr{w_{G'}(r_G,v)}{f}{w_{H'}(r_H,h(v))} = \sem{H}(h(v))\ . \]  

  It remains to show that $\vec{v}$ is admissible.   Suppose that $v_{i,j} \in L_G(v_{k,\ell})$---we must show that $\tuple{i,j}$ precedes $\tuple{k,\ell}$ in lexicographic order.  Since $h$ preserves loops, we must have $u_i = f(v_{i,j}) \in L_H(f(v_{k,\ell}))  = L_H(u_k)$.  Since $u_1,\dots,u_n$ is admissible, we must have $i \leq k$.  If $i < k$, then we have the result; if $i = k$, then since $v_{i,1},\dots v_{i,m_i}$ is admissible, we must have
  $j \leq \ell$.
\end{proof}

\subsection*{The Category of Flow Graphs}

This section justifies Theorem~\ref{thm:rpa} as a robustness property in the sense of Section~\ref{sec:robustness}, by showing that $\mathbf{G}_{\mathbf{A}}$ is indeed a category.

\begin{lemma} \label{lem:local-cycle-preservation}
  Let $G$ and $H$ be flow graphs, and let $h : V_G \rightarrow V_H$ be loop-preserving.  Then for any vertex $v$ and local cycle $\pi$ of $v$, $h(\pi)$ is a local cycle.
\end{lemma}
\begin{proof}
  It is easy to see that $h(\pi)$ is a cycle of $v$; we need only to show that each vertex in $h(\pi)$ is dominated by $v$.  For each vertex $z$ along $h(\pi)$, there is some corresponding vertex $u$ along $\pi$ such that $h(u) = z$.  Since $\pi$ is a local cycle, we have $u \in L_G(v)$.  Since $h$ is loop-preserving, we have
  $z = h(u) \in L_H(h(v))$.  Since $z$ appears on a local cycle of $f(v)$ in $H$, $h(v)$ must dominate $z$ in $H$.  Since this holds for all $z$ along $h(\pi)$, $h(\pi)$ is a local cycle of $h(v)$.
  \end{proof}

\begin{proposition}
  $\mathbf{G}_\mathbf{A}$ forms a category---i.e., arrows are closed under composition.
\end{proposition}
\begin{proof}
  Let $G_1$, $G_2$, and $G_3$ be flow graphs, and let $\tuple{h_1,f_1} \in \mathbf{G}_{\mathbf{A}}(G_1,G_2)$ and $\tuple{h_2,f_2} \in \mathbf{G}_{\mathbf{A}}(G_2,G_3)$ be arrows.  We must show that
  $\tuple{h_2 \circ h_1,f_2 \circ f_1}$ is a loop-preserving stuttering simulation from $G_1$ to $G_3$.

  \begin{itemize}
  \item Simulation: Let $\tuple{u,v}$ be an edge in $G_1$.  If $h_1(u) = h_1(v)$ and $\barr{w_{G_1}(u,v)}{f_1}{1}$, then
  $h_2(h_1(u)) = h_2(h_1(v))$ and $\barr{w_{G_1}(u,v)}{f_1}{1}$.
  Otherwise, we must have $\tuple{h_1(u),h_1(v)} \in E_{G_2}$ and
  $\barr{w_{G_1}(u,v)}{f_1}{w_{G_2}(h_1(u),h_1(v))}$.
  Since $\tuple{h_2,f_2}$ is a stuttering simulation,
  either 
  (1) $h_2(h_1(u)) = h_2(h_1(v))$ and $\barr{w_{G_2}(h_1(u),h_1(v))}{f_2}{1}$ (in which case
  $\barr{w_{G_1}(u,v)}{f_2 \circ f_1}{1}$), or
  $\tuple{h_2(h_1(u)), h_2(h_1(v))} \in E_{G_3}$ and
  $\barr{w_{G_2}(h_1(u),h_1(v))}{f_2}{w_{G_3}(h_2(h_1(u)), h_2(h_1(v)))}$
  (in which case 
  $\barr{w_{G_1}(u,v)}{f_2 \circ f_1}{w_{G_3}(h_2(h_1(u)), h_2(h_1(v)))}$).
  \item Loop-preserving: the fact that for every $u,v \in V_{G_1}$ such that $u \in L_{G_1}(v)$ implies 
  $h_2(h_1(u)) \in L_{G_3}(h_2(h_1(v)))$ is immediate. 

  Suppose that there is some $u,v \in V_G$ with $h_2(h_1(u)) = h_2(h_1(v))$ and $u \in L_{G_1}(v)$.  We must show that $L_{G_1}(u) = \emptyset$.  For a contradiction, suppose that $z \in L_{G_1}(u)$.  Since $h_1$ is loop-preserving, we have $h_1(z) \in L_{G_2}(h_1(u))$.  Since $h_2$ is loop-preserving and we have
  $h_2(h_1(u)) = h_2(h_1(v))$, $h_1(u) \in L_{G_2}(h_1(v))$, and
  $L_{G_1}(h_1(u))$ non-empty we must have $h_1(u) = h_1(v)$.  This contradicts the assumption that $h_1$ is loop-preserving, since $h_1(u) = h_1(v)$, $u \in L_{G_1}(v)$, and $L_{G_1}(u)$ is non-empty.

  Suppose that $\pi$ and $\pi'$ are primitive local cycles of some vertex $v$.  
  Since $h_1$ is loop-preserving,
  we have $|r_1(\pi)|_{h_1(v)} = |r_1(\pi')|_{h_1(v)}$, and thus we may decompose each cycle into the same number of primitive cycles:
  \begin{align*}
  r_1(\pi) &= \tau_1 \tau_2 \dots \tau_k\\
  r_1(\pi') &= \tau_1' \tau_2' \dots \tau_k'
  \end{align*}
  By Lemma~\ref{lem:local-cycle-preservation}, 
  each $\tau_i$ and $\tau_i'$ is a local cycle of $h_1(v)$.  Since $\tuple{h_2,f_2}$ is loop-preserving, there is some natural number $t$ such that each primitive local cycle $\rho$ of $h_1(v)$ has $|r_1(\pi)|_{h_2(h_1(v))} = t$.  Then we have
  that
  $|r_2(r_1(\pi))|_{h_2(h_1(v))} = tk = |r_2(r_1(\pi'))|_{h_2(h_1(v))}$. \qedhere
  \end{itemize}
\end{proof}

%%% Local Variables:
%%% TeX-master: "main"
%%% End:

\end{document}